\def\confversion{0}
\def\ifconf{\ifnum\confversion=1}
\def\ifnotconf{\ifnum\confversion=0}

\documentclass[11 pt, fleqn, reqno]{article}
\def\showauthornotes{0}
\def\showkeys{0}
\def\showdraftbox{0}

\usepackage{xspace,xcolor,enumerate}
\usepackage{amsmath,amssymb}
\usepackage{amsthm}
\usepackage[toc,page]{appendix}
\usepackage{thmtools}
\usepackage{thm-restate}
\usepackage{color,graphicx}
\usepackage{boxedminipage}
\usepackage{makecell}
\usepackage{tabularx}
\usepackage{enumitem}
\usepackage[linesnumbered,ruled,vlined]{algorithm2e}
\ifnum\showkeys=1
\usepackage[color]{showkeys}
\fi

\definecolor{darkred}{rgb}{0.5,0,0}
\definecolor{darkgreen}{rgb}{0,0.35,0}
\definecolor{darkblue}{rgb}{0,0,0.55}

\usepackage[pdfstartview=FitH,pdfpagemode=UseNone,colorlinks,linkcolor=darkblue,filecolor=darkred,citecolor=darkgreen,urlcolor=darkred,pagebackref]{hyperref}

\usepackage[capitalise,nameinlink]{cleveref}
\usepackage[T1]{fontenc}
\usepackage{mathtools,dsfont,bbm}
\ifnum\showauthornotes=1
\newcommand{\Authornote}[2]{{\sf\small\color{red}{[#1: #2]}}}
\newcommand{\Authorcomment}[2]{{\sf \small\color{gray}{[#1: #2]}}}
\newcommand{\Authorfnote}[2]{\footnote{\color{red}{#1: #2}}}
\else
\newcommand{\Authornote}[2]{}
\newcommand{\Authorcomment}[2]{}
\newcommand{\Authorfnote}[2]{}
\fi

\ifnum\showdraftbox=1
\newcommand{\draftbox}{\begin{center}
  \fbox{%
    \begin{minipage}{2in}%
      \begin{center}%
        \begin{Large}%
          \textsc{Working Draft}%
        \end{Large}\\
        Please do not distribute%
      \end{center}%
    \end{minipage}%
  }%
\end{center}
\vspace{0.2cm}}
\else
\newcommand{\draftbox}{}
\fi

\newtheorem{theorem}{Theorem}[section]

\newtheorem{definition}[theorem]{Definition}

\newtheorem{lemma}[theorem]{Lemma}
\newtheorem{remark}[theorem]{Remark}
\newtheorem{proposition}[theorem]{Proposition}
\newtheorem{corollary}[theorem]{Corollary}
\newtheorem{claim}[theorem]{Claim}
\newtheorem{fact}[theorem]{Fact}

\theoremstyle{remark}
\newtheorem{algo}[theorem]{Algorithm}

\def\FullBox{\hbox{\vrule width 6pt height 6pt depth 0pt}}

\def\qedsketch{\ifmmode\Box\else{\unskip\nobreak\hfil
\penalty50\hskip1em\null\nobreak\hfil$\Box$
\parfillskip=0pt\finalhyphendemerits=0\endgraf}\fi}
\def\to{\rightarrow}

\def\epsilon{\varepsilon}

\def\d{\delta}
\def\cal{\mathcal}

\def\implies{\Rightarrow}

\renewcommand{\bar}{\overline} 

\newcommand{\ie}{i.e.,\xspace}
\newcommand{\eg}{e.g.,\xspace}
\newcommand{\etal}{et al.\xspace}

\newcommand{\mper}{\,.}
\newcommand{\mcom}{\,,}

\newcommand{\R}{{\mathbb R}}
\newcommand{\E}{{\mathbb E}}
\newcommand{\C}{{\mathbb C}}
\newcommand{\N}{{\mathbb{N}}}

\newcommand{\FF}{{\mathbb F}}

\newcommand{\cA}{\mathcal{A}}
\newcommand{\cB}{\mathcal{B}}

\newcommand{\cH}{\mathcal{H}}

\newcommand{\cE}{\mathcal{E}}
\newcommand{\cN}{\mathcal{N}}
\newcommand{\cG}{\mathcal{G}}
\newcommand{\cU}{\mathcal{U}}

\usepackage{nicefrac}
\newcommand{\abs}[1]{\ensuremath{\left\lvert #1 \right\rvert}}

\newcommand{\norm}[1]{\ensuremath{\left\lVert #1 \right\rVert}}

\newcommand{\ip}[2] {\ensuremath{\left\langle #1 , #2 \right\rangle}}

\newcommand{\one}{{\mathbf{1}}}

\newcommand{\Esymb}{\mathbb{E}}

\makeatletter

\def\Ex#1{%
    \ProbabilityRender{\Esymb}{#1}%
}

\def\ProbabilityRender#1#2{%fancy probability command
  \@ifnextchar\bgroup%
  {\renderwithdist{#1}{#2}}
   {\singlervrender{#1}{#2}}
}
\def\singlervrender#1#2{%
   \ensuremath{\mathchoice
       {{#1}\left[ #2 \right]}
       {{#1}[ #2 ]}
       {{#1}[ #2 ]}
       {{#1}[ #2 ]}
   }
}
\def\renderwithdist#1#2#3{%
   \@ifnextchar\bgroup
   {\superfancyrender{#1}{#2}{#3}}
   {\ensuremath{\mathchoice
      {\underset{#2}{#1}\left[ #3 \right]}
      {{#1}_{#2}[ #3 ]}
      {{#1}_{#2}[ #3 ]}
      {{#1}_{#2}[ #3 ]}
     }
   }
}
\def\superfancyrender#1#2#3#4#5{
   \ensuremath{\mathchoice
      {\underset{#1}{{#1}}\left#4 #3 \right#5}
      {{#1}_{#2}#4 #3 #5}
      {{#1}_{#2}#4 #3 #5}
      {{#1}_{#2}#4 #3 #5}
   }
}
\makeatother

\newcommand{\conv}[1]{\mathrm{conv}\inparen{#1}}

\newfont{\inhead}{eufm10 scaled\magstep1}

\newcommand{\calP}{{\cal P}}

\newcommand{\poly}{{\mathrm{poly}}}

\DeclareMathOperator\supp{supp}

\DeclareMathOperator{\sgn}{\operatorname{sgn}}

\DeclareMathOperator{\sym}{\operatorname{Sym}}

\renewcommand{\bar}[1]{\ensuremath{\overline{#1}}}

\newcommand{\1}[1]{\mathds{1}\left[#1\right]}

\newcommand{\val}{{\sf val}}

\newcommand{\inparen}[1]{\left(#1\right)}             %\inparen{x+y}  is (x+y)
\newcommand{\ket}[1]{\lvert #1\rangle}
\newcommand{\braket}[2]{\left\langle #1 \,\middle\vert\, #2\right\rangle}
\newcommand{\QMA}{\textup{QMA}}
\newcommand{\BQP}{\textup{BQP}}
\newcommand{\NP}{\textup{NP}}
\newcommand{\NEXP}{\textup{NEXP}}
\newcommand{\PCP}{\textup{PCP}}

\DeclareSymbolFont{extraup}{U}{zavm}{m}{n}
\DeclareMathSymbol{\varheart}{\mathalpha}{extraup}{86}
\DeclareMathSymbol{\vardiamond}{\mathalpha}{extraup}{87}

\def\rand#1{\mathbf{#1}}
\def\rv#1{\rand #1}

\def\one{\mathbf 1}

\DeclarePairedDelimiter\set{\lbrace}{\rbrace}

\usepackage{tikz}

\newcommand{\gapLarge}{3/4+1/\poly(n)}
\newcommand{\gapSmall}{2/9}
\newcommand{\gapNaive}{7/9+e^{-\poly(n)}} % 1 - gapSmall

\newcommand{\SEP}{\textup{SEP}}

\newcommand{\Symm}[2]{\vee^{#1}(\C^{#2})}

\usepackage{tikz}
\usetikzlibrary{shapes.symbols}
\usepackage{float}
\usepackage{graphicx}
\usepackage{pstyle}

\usepackage[top=1in, bottom=1in, left=1.25in, right=1.25in]{geometry}
\begin{document}

\title{The Power of Unentangled Quantum Proofs with Non-negative Amplitudes}

\author{Fernando Granha Jeronimo\thanks{{\tt IAS \& Simons Institute}. {\tt granha@ias.edu}. This material is based upon work supported by the National Science Foundation under Grant No. CCF-1900460. The work is done while the author was at IAS.} \and Pei Wu\thanks{{\tt Weizmann Institute of Science}. {\tt pwu@ias.edu}. This material is based upon work supported by the National Science Foundation under Grant No. CCF-1900460. The work is done while the author was at IAS.}}

\date{\vspace{-5ex}}

\maketitle
\draftbox
\thispagestyle{empty}

\begin{abstract}
Quantum entanglement is a fundamental property of quantum mechanics
and plays a crucial role in quantum computation and
information.  Despite its importance, the power and limitations of
quantum entanglement are far from being fully understood. Here, we
study entanglement via the lens of computational complexity. This is
done by studying quantum generalizations of the class $\textup{NP}$
with multiple \emph{unentangled} quantum proofs, the so-called
$\textup{QMA}(2)$ and its variants. The complexity of
$\textup{QMA}(2)$ is known to be closely connected to a variety of
problems such as deciding if a state is entangled and several
classical optimization problems. However, determining the complexity
of $\textup{QMA}(2)$ is a longstanding open problem, and only the trivial
complexity bounds $\textup{QMA} \subseteq \textup{QMA}(2) \subseteq \textup{NEXP}$
are known.

In this work, we study the power of \emph{unentangled} quantum proofs
with \emph{non-negative} amplitudes, a class which we denote
$\textup{QMA}^+(2)$.  In this setting, we are able to design proof
verification protocols for (increasingly) hard problems both
using \emph{logarithmic} size quantum proofs and having
a \emph{constant} probability gap in distinguishing yes from no
instances.  In particular, we design \emph{global} protocols for small
set expansion (SSE), unique games (UG), and PCP verification. As a
consequence, we obtain
$\textup{NP} \subseteq \textup{QMA}^+_{\log}(2)$ with a constant gap.
By virtue of the new \emph{constant} gap, we are able to ``scale up'' this
result to $\textup{QMA}^+(2)$, obtaining the full characterization
$\textup{QMA}^+(2)=\textup{NEXP}$ by establishing stronger explicitness properties
of the $\PCP$ for $\textup{NEXP}$.  We believe that our protocols are
interesting examples of proof verification and property testing in
their own right.  Moreover, each of our protocols has a single
isolated property testing task relying on non-negative amplitudes
which if generalized would allow transferring our results to
$\textup{QMA}(2)$.

One key novelty of these protocols is the manipulation of quantum
proofs in a \emph{global} and \emph{coherent} way yielding constant
gaps.  Previous protocols (only available for general amplitudes) are
either \emph{local} having vanishingly small gaps or treating the quantum
proofs as classical probability distributions requiring polynomially many
proofs. In both cases, these known protocols do not imply non-trivial
bounds on $\textup{QMA}(2)$.

Finally, we show that $\textup{QMA}(2)$ is equal to $\textup{QMA}^+(2)$ provided
the gap of the latter is a sufficiently large constant. In particular, as a result of the
above characterization of $\textup{NEXP}$, we obtain that if $\textup{QMA}^+(2)$ admits
strong gap amplification for the completeness and soundness gap,
then $\textup{QMA}(2)=\textup{NEXP}$.
\end{abstract}

\newpage

\ifnotconf
\pagenumbering{roman}
\tableofcontents
\clearpage
\fi

\pagenumbering{arabic}
\setcounter{page}{1}

\section{Introduction}

Quantum entanglement is a fundamental property of quantum mechanics
and it plays a major role in several fields such as quantum
computation, information, cryptography, condensed matter physics,
etc~\cite{HHHH09,NC10,Watrous18,O19}. Roughly speaking, quantum
entanglement is a distinctive form of quantum correlation that is
stronger than classical correlations.  Entanglement can lead to
surprising (and sometimes counter-intuitive) phenomena as presented in
the celebrated EPR paradox \cite{EPR35} and the violation of Bell's
(style) inequalities~\cite{B64,CHSH69}. In a sense, entanglement is
necessary to access the full power of quantum computation since it is
known that quantum computations requiring ``little'' entanglement
can be simulated classically with small overhead
\cite{Vidal03}. Entanglement is also crucial in a variety of protocols
such as quantum key distribution \cite{BB14}, teleportation
\cite{BBCJPW93}, interactive proof systems \cite{JNVWY20}, and so
on. However, despite this central role, the power and limitations of
quantum entanglement are far from being understood. Here, we
study quantum entanglement via the lens of computational
complexity. More precisely, we investigate the role of entanglement in
the context of quantum proof verification.

The notions of provers, proofs, and proof verification play a central
role in our understanding of classical complexity theory
\cite{AroraBarak09}. The quantum setting allows for various and vast
generalizations of these classical notions \cite{VW16}.  For instance,
by allowing the proof to be a quantum state of polynomial size and the
verifier to be an efficient quantum machine, one obtains the class
$\QMA$ which is a natural generalization of the class $\NP$
\cite{Watrous00}. The $\QMA$ proof verification model can be further
generalized to two quantum proofs from two \emph{unentangled}
provers. This generalization gives rise to a class known as $\QMA(2)$
\cite{KMY03} (see~\cref{def:qma_k}). This latter complexity class is
known to be closely connected to a variety of computational problems
such as the fundamental problem of deciding whether a quantum state
(given its classical description) is entangled or not. It is also
connected to a variety of classical optimization problems such as
polynomial and tensor optimization over the sphere as well as some
norm computation problems~\cite{HM13}.

Determining the complexity of $\QMA(2)$ is a major open 
problem in quantum complexity. Contrary to many other quantum proof
systems (\eg $\textup{QIP}$ \cite{JJUW11} and $\textup{MIP}^*$
\cite{JNVWY20}), we still do not know any non-trivial complexity
bounds for $\QMA(2)$.  On one hand, we trivially have $\QMA \subseteq
\QMA(2)$ since a $\QMA(2)$ verifier can simply ignore one of the
proofs. On the other hand, a NEXP verifier can guess exponentially
large classical descriptions of two quantum proofs of polynomially
many qubits and simulate the verification protocol classically in
exponential time. Hence, we also have $\QMA(2) \subseteq
\NEXP$. Despite considerable effort with a variety of powerful
techniques brought to bear on this question, such as semi-definite
programming hierarchies \cite{DPS04,BKS17,HNW17}, quantum de Finetti
theorems~\cite{KM09,BH13,BCY11}, and carefully designed nets
\cite{BH15,SW12}, only the trivial bounds $\QMA \subseteq \QMA(2)
\subseteq \NEXP$ are known.

Even though there are no non-trivial complexity bounds for
$\QMA(2)$, there are results showing surprisingly powerful consequences
of \emph{unentangled} proofs. An early result by Blier and Tapp
\cite{BT09} shows that two \emph{unentangled} proofs of a logarithmic
number of qubits suffice to verify the NP-complete problem of graph
$3$-coloring.  The version of $\QMA(2)$ with logarithmic-size proofs
is known as $\QMA_{\log}(2)$. Since $\QMA_{\log}(1)
\subseteq \BQP$ due to Marriott and Watrous~\cite{MW05}, Blier and Tapp's work provides some evidence that having
two unentangled proofs of logarithmic size is more powerful than
having a single one. This suggests that the lack of quantum
entanglement across the proofs can play an important role in proof
verification.  Furthermore, note that this situation is in sharp
contrast with the classical setting where having two classical proofs
of logarithmic size is no more powerful than having a single one
since two proofs can be combined into a larger one.

The above protocol has a critical drawback, namely, the
verifier only distinguishes yes from no instances with a
polynomially small probability. This distinguishing probability is
known as the \emph{gap} of the protocol.  These weak gaps are 
undesirable for two reasons. First, we cannot obtain
tighter bounds on $\QMA(2)$ from these protocols since scaling up
these results to $\QMA(2)$ leads to exponentially small gaps. Such tiny
gaps fall short to imply $ \QMA(2) = \NEXP$ as the definition of $\QMA(2)$
can tolerate up to only polynomially small gaps. Second, the strength of
the various hardness results that can be deduced from these protocols
depends on how large the gap is. For instance, we do not know if
several of these problems are also hard to approximate within say a
more robust universal constant. A series of subsequent works extended
Blier and Tapp's result in the context of $\QMA_{\log}(2)$ protocols for
NP-complete problems \cite{Beigi10,GNN12,CF13}. However, all these
results suffer from a polynomially small gap.

Another piece of evidence pointing to the additional power of unentangled
proofs appears in the work of Aaronson \etal~\cite{ABDFS08}. They show
that $\widetilde{O}(\sqrt{n})$ quantum proofs of logarithmic
size suffice to decide an NP-complete variant of the SAT problem of
size $n$ with a constant gap. Due to the work of Harrow and Montanaro
\cite{HM13}, it is possible to convert this protocol into a two-proof
protocol where each one has size $\widetilde{O}(\sqrt{n})$ and the gap
remains constant. Unfortunately, this converted protocol does not
imply tighter bounds for $\QMA(2)$ since it only shows $\NP \subseteq
\QMA(2)$.

In this work, we study \emph{unentangled} quantum proofs with
\emph{non-negative} amplitudes. We name the associated complexity
classes introduced here as $\QMA^+(2)$ and $\QMA_{\log}^+(2)$
(see~\cref{def:qma_plus_k}) in analogy to $\QMA(2)$ and
$\QMA_{\log}(2)$, respectively. The main question we consider is the
following:

\begin{center}
  What is the power of \emph{unentangled} proofs with \emph{non-negative} amplitudes?
\end{center}

This non-negative amplitude setting is intended to capture several
structural properties of the general $\QMA(2)$ model while providing
some restriction on the adversarial provers in order to gain a better
understanding of unentangled proof verification. In this non-negative
amplitude setting, we are able to derive much stronger results and
fully characterize $\QMA^+(2)$. In particular, we are able to design
$\QMA_{\log}^+(2)$ protocols with \emph{constant} gaps for
(increasingly) hard(er) problems. Each of these protocols contributes
to our understanding of proof verification and leads to different sets
of techniques, property testing routines, and analyses.

Our first protocol is for the small set expansion (SSE) problem
\cite{RS10,BarakBHKSZ12}. Roughly speaking, the SSE problem asks
whether all small sets of an input graph are very
expanding\footnote{In terms of edge expansion.} or if there is a
small non-expanding set. The SSE problem arises in the context of the unique
games (UG) conjecture. This conjecture plays an important role in the
classical theory of hardness of approximation
\cite{Khot02:unique,KR03,KhotKMO04,Raghavendra08,KhotOD09,Khot10:icm}. One
key reason is that the unique games problem is a (seemingly) more
structured computational problem as opposed to more general and
provably NP-hard constraint satisfaction problems (CSPs) making it
easier to reduce UG to other problems. In this context, the SSE
problem is considered an even more structured problem than UG since
some of its variants can be reduced to UG. This extra structure of SSE
compared to UG can make it even easier to reduce SSE to other
problems. So far the hardness of SSE remains an open problem---it has evaded the best-known algorithmic techniques~\cite{RST10}.

\begin{theorem}[Informal]\label{theo:inform:protocol:sse}
  Small set expansion is in $\QMA^+_{\log}(2)$ with a constant gap.
\end{theorem}

Our second protocol is for the unique games problem. The UG problem is
a special kind of CSP wherein the constraints are permutations and
it is enough to distinguish almost fully satisfiable instances from
those that are almost fully unsatisfiable. The fact that the
constraints of a UG instance are bijections which in turn can be
implemented as valid (\ie unitary operators) is explored in our
protocol. Although the hardness of UG remains an open problem, a
weaker version of the UG problem was recently proven to be
NP-hard~\cite{DKKMS18,KMS18,BKS19}. From our UG protocol and this
weaker version of the problem, we obtain $\NP \subseteq
\QMA_{\log}^+(2)$ with a \emph{constant} gap
(see~\cref{cor:np_qma_plus_log} below).

\begin{theorem}[Informal]\label{theo:inform:protocol:ug}
   Unique Games is in $\QMA^+_{\log}(2)$ with a constant gap.
\end{theorem}

A key novelty of our protocols is their \emph{global} and
\emph{coherent} manipulation of quantum proofs leading to
\emph{constant} gaps. The previous protocols for $\QMA_{\log}(2)$ with
a logarithmic proof size are
\emph{local} in the sense that they need to read \emph{local}
information\footnote{Roughly speaking, they treat a quantum proof as
  quantum random access codes that encodes $n$ bits using $\log_2(n)$
  qubits. By Nayak's bound the probability of recovering a queried
  position is polynomially small in $n$.} from the quantum proofs
thereby suffering from vanishingly small gaps. Furthermore,
the previous protocol with a constant gap treats the quantum proofs as
classical probability distributions (\eg relying on the birthday paradox)
and this classical treatment ends up requiring polynomially many
proofs to achieve the constant gap.

Another interesting feature of our protocols is that they already
almost work in the general amplitude case in the sense that each
protocol isolates a single property testing task relying on
non-negative amplitudes. If such a property testing task can be
generalized to general amplitudes, then the corresponding protocol
works in $\QMA_{\log}(2)$ as well.

As discussed earlier, by \cref{theo:inform:protocol:ug} together with the work
on the $2$-to-$2$ conjecture, we obtain that NP is contained in
$\QMA^+_{\log}(2)$ with a \emph{constant} gap.

\begin{corollary}[Informal]\label{cor:np_qma_plus_log}
  $\NP \subseteq \QMA^+_{\log}(2)$ with a constant gap.
\end{corollary}
By virtue of the \emph{constant} gaps of our protocols for
$\QMA_{\log}^+(2)$, we can ``scale up'' our results to give an
exact characterization of $\QMA^+(2)$ building on top of ideas of 
very efficient classical PCP verifiers.

\begin{theorem}\label{thm:qmaplus-nexp-intro}
  $\QMA^+(2) = \NEXP$.
\end{theorem}

The above characterization is shown by designing a \emph{global}
$\QMA^+(2)$ protocol for NEXP. To design this \emph{global} protocol,
we not only rely on the properties of the known efficient classical
PCP verification for NEXP, but we need additional explicitness and
regularity properties. Regarding the explicitness, we call \emph{doubly
  explicit} the kind of PCP required in our \emph{global}
protocol (in analogy to the terminology of graphs). Roughly speaking,
doubly explicitness means that we can very efficiently not only
determine the variables appearing in any given constraint but also
reverse this mapping by very efficiently determining the
constraints in which a variable appears. Here, we prove that these
properties can be indeed obtained by carefully combining known PCP
constructions.

An intriguing next step is to explore the improved understanding of
the unentangled proof verification from our protocols in the general
amplitude case. Investigating problems like SSE and UG might provide
more structure towards this goal. Characterizing the complexity of
$\QMA(2)$ would be extremely interesting whatever this
characterization turns out to be.

At this moment, a natural question remains is what is the relationship between $\QMA^+(2)$ and $\QMA(2)$. We prove that they can simulate each other to some extent. In particular, a $\QMA(2)$ protocol can always be simulated by a $\QMA^+(2)$ protocol without any loss in the completeness and soundness. This direction is not  surprising at this point, as $\QMA^+(2)=\NEXP.$ On the other direction, we show that a $\QMA(2)$ protocol can also simulate the $\QMA^+(2)$ protocol at the cost of worsening the soundness by a multiplicative factor ofat most 4. An immediate corollary is the following:
\begin{corollary}
If $\NEXP\subseteq \QMA^+(2)$ with a completeness and soundness gap at least $\gapLarge$. Then
\[
    \QMA(2) = \NEXP.
\]
\end{corollary}
Therefore, a strong enough gap amplification for $\QMA^+(2)$ would solve the long-standing open problem of characterizing $\QMA(2)$.

\paragraph{Organization.}This document is organized as follows. In~\cref{sec:strategy}, we give
an overview of our global protocols. In~\cref{sec:prelim}, we formally
define $\QMA^+(2)$ and its variants as well as fix some notation and
recall basic facts. In~\cref{sec:prop_test}, we develop some quantum
property testing primitives that will be common to our
protocols. In~\cref{sec:SSE-protocol}, we present our global protocol
for SSE. In~\cref{sec:ug-protocol}, we present our global protocol for
UG and we use it to prove $\NP \subseteq \QMA_{\log}^+(2)$ with
a constant gap. In~\cref{sec:nexp}, we prove the characterization
$\QMA^+(2) = \NEXP$. In the last section, \cref{sec:qma2-vs-qma2+}, we discuss the relationship between the complexity class of $\QMA^+(2)$ and $\QMA(2)$, pointing out a potential direction towards the ``$\QMA(2)=\NEXP$?'' problem.

% Prelim
\section{Preliminaries}\label{sec:prelim}

Let $\N, \R, \C $ stand for the natural, real, and complex numbers. $\N^+$ denotes the positive natural numbers. For any real number $x$,
\[
\sgn(x)= \begin{cases}
    1 & x > 0;\\
    0 & x = 0;\\
    -1 & x<0.
\end{cases}
\]
In this paper, $\log$ stands for the logarithm to base 2.
For $p \in [1,\infty)$, we denote the $\ell_p$-norm of $u \in \C^n$ as
$\norm{u}_p$, \ie $\norm{u}_p = \left(\sum_{i=1}^n
\abs{u_i}^p\right)^{1/p}$. We omit the subscript for the $\ell_2$-norm, \ie $\norm{u} := \norm{2}_2$.
We denote the $\ell_\infty$-norm of $u \in \C^n$ as $\norm{u}_\infty$, \ie
$\norm{u}_\infty = \max_{i \in [n]} \abs{u_i}$.
Let  $\unitC_{n}:=\{u\in\C^{n+1}:\|u\|=1\}$ be the $n$-dimensional sphere and 
$\unitR_{n}:=\{u\in(\posR)^{n+1}:\|u\|=1\}$ be the intersection of the $n$-dimensional sphere and the non-negative orthant. 
The subscript will almost always be omitted in this manuscript
since it can be confusing and the dimension is normally clear from the context.
Adopt the short-hand notation $[n]=\{1,2,\ldots,n\}$. For any universe $U$ and any subset $S\subseteq U$, let $\bar S:=U \setminus S$. Denote the characteristic vector of $S$ by $\one_S$, i.e., $\one_S\in \R ^U$ and 
\[
\one_S(x) = \begin{cases}
    1 & \text{if } x\in S,\\
    0 & \text{otherwise.}
\end{cases}
\]
For a logical condition $C,$ we use the Iverson bracket
\[
\1{C}=\begin{cases}
1 & \text{if \ensuremath{C} holds,}\\
0 & \text{otherwise.}
\end{cases}
\]

Let $\Sigma$ be an arbitrary non-empty alphabet. For any strings $s\in \Sigma^*$, $|s|$ denotes its length. For and $I\subset \N$, we denote the substring of $s$ with index in $I$ by $s|_I$. Thus, $s|_{\{i_1,i_2,\ldots,i_m\}} = s_{i_1}s_{i_2}\cdots s_{i_m}$. For two strings $s, t$, we use $s\prec t$ to mean that $t$ is a prefix of $s$.

We adopt the Dirac notation for vectors representing quantum states, e.g., $\ket\psi, \ket\phi$, etc. In this paper, all the vectors of the form $\ket\psi$ are unit vectors. Given any pure state $\ket\psi$, we adopt the convention that its density operator is denoted by the Greek letter without the ``ket'', e.g. $\psi = \ket\psi\bra\psi$. 
\iffalse
A symmetric state $\ket\psi \in (\C^d)^{\otimes n}$ is that invariant under any permutation $\pi\in \sym_n$ where $\sym_n$ is the symmetric group. The action of $\pi$ on $(\C^d)^{\otimes n}$ is  
\[
\pi:|\psi_1, \psi_2, \ldots, \psi_n\rangle
\mapsto |\psi_{\pi(1)}, \psi_{\pi(2)}, \ldots, \psi_{\pi(n)} \rangle.
\]
The \emph{symmetric subspace} is the subspace of $(\C^d)^{\otimes n}$ that is invariant under $S_n$, denoted by $\Symm{n}{d}$. 
\fi
Given any set $H\subseteq \cH$ for some Hilbert space $\cH$, $\conv H$ is the convex hull of $H$. 
One particularly interesting set of states to us is the \emph{separable} states. We adopt the following notation for the set of density operators regarding separable states,
\begin{align*}
  \SEP(d,r) \coloneqq \conv{\psi_1 \otimes \cdots \otimes \psi_r \mid \ket{\psi_1},\ldots,\ket{\psi_r} \in \C^d}.
\end{align*}
A related notion is that of separable measurement, whose formal definition is given below.
\begin{definition}[Separable measurement]
A measurement $M = (M_0, M_1)$ is separable if in the yes case, the corresponding Hermitian matrix $M_1$ can be represented as a conical combination of two operators acting on the first and second parts, i.e.,
for some distribution $\mu$ over the tensor product of PSD matrices $\alpha$ and $\beta$ on the corresponding space,
\[
M_1 = \int  \alpha \otimes \beta ~\mathrm{d}\mu.
\]
\end{definition}
We record the following well-known fact. An interested reader is referred to~\cite{harrow2013church} for a formal proof.
\begin{fact}[Folklore] \label{fact:swap-test-sep}
The swap test is separable.
\end{fact}

\paragraph{Matrix Analysis} Given any matrix $M\in \C^{n\times n}$, $M^\dagger$ is its  conjugate transpose. Let $\sigma_1 \ge \sigma_2 \ge \ldots \ge \sigma_n$ denote its singular values. Then the trace norm $\|\cdot\|_1$, Frobenius norm $\|\cdot\|_F$ are defined as below
\begin{align*}
    \|M\|_1 = \sum_i \sigma_i, \qquad \|M\|_F = \sqrt{\sum_i \sigma_i^2} .
\end{align*}
The Frobenius norm also equals the square root of sum of squared lengths, i.e., $\|M\|_F=\sqrt{\sum_{i,j} |M(i,j)|^2}$.

For a positive semi-definite (PSD) matrix $M$, $\|M\|_F = \sqrt{\trace M ^2}$.  
For two PSD matrices, there is one (of many) analogous matrix Cauchy-Schwarz inequality.
\begin{equation}\label{eq:matrix-cauchy-schwarz}
\trace(\sigma\rho)\le \|\sigma\|_F \cdot\|\rho\|_F.
\end{equation}
We adopt the notation $\succeq$ to denote the partial order  that $\sigma \succeq \rho$ if $\sigma-\rho$ is  positive semi-definite.

\subsection{Quantum Merlin-Arthur with Multiple Provers}
The class $\QMA(k)$ can be formally defined in more generality as
follows.

\begin{definition}[$\QMA_\ell(k,c,s)$]\label{def:qma_k}
  Let $k \colon \mathbb{N} \to \mathbb{N}$ and $c,s,\ell \colon \mathbb{N} \to \mathbb{R}^{+}$
  be polynomial time computable functions.
  A promise problem $\mathcal{L}_{\textup{yes}},\mathcal{L}_{\textup{no}} \subseteq \set{0,1}^*$ is in $\QMA_\ell(k,c,s)$ if there exists a $\BQP$
  verifier $V$ such that for every $n \in \mathbb{N}$ and every $x \in \set{0,1}^n$, 
  \begin{itemize}[label={}]
    \item \textbf{\textup{Completeness:}} If $x \in \mathcal{L}_{\textup{yes}}$, then there exist unentangled states $\ket{\psi_1}, \ldots, \ket{\psi_{k(n)}}$, each on at most $\ell(n)$ qubits,
                                        s.t.\ $\Pr[V(x, \ket{\psi_1} \otimes \cdots \otimes \ket{\psi_{k(n)}})\textup{ accepts}] \ge c(n)$.
\vspace{\parsep}
    \item \textbf{\textup{Soundness:}} If $x \in \mathcal{L}_{\textup{no}}$, then for every unentangled states $\ket{\psi_1}, \ldots, \ket{\psi_{k(n)}}$, each on at most $\ell(n)$ qubits,
                                        we have \ $\Pr[V(x, \ket{\psi_1} \otimes \cdots \otimes \ket{\psi_{k(n)}})\textup{ accepts}] \le s(n)$.
  \end{itemize}
\end{definition}

Harrow and Montanaro proved that: For any state $\ket \psi \in \C^{d_1}\otimes \C^{d_2}\otimes \cdots\otimes \C^{d_k}$, if 
\[
    \max_{\phi_i \in \C^{d_i}} \langle \psi \mid \phi_1 \phi_2\ldots \phi_k\rangle = 1-\epsilon,
\]
then the \emph{product test} rejects $\ket\psi^{\otimes 2}$ with probability $\Omega(\epsilon)$. Based on this 
product test, Harrow and Montanaro further showed in the $\QMA$ protocols, the number
of provers can always be reduced to $2$.
\begin{theorem}[Harrow and Montanaro~\cite{HM13}]\label{thm:HM}
For any $\ell, k, 0\le s < c \le 1$,
\[
\QMA_\ell(k, c,s) \subseteq \QMA_{k\ell}(2, c',s'),
\]
where $c'=(1+c)/2$ and $s'=1-(1-s)^2/100$.
\end{theorem}

The class $\QMA(k)^+$ is formally defined in more generality as
follows.

\begin{definition}[$\QMA_\ell^+(k,c,s)$]\label{def:qma_plus_k}
  Let $k \colon \mathbb{N} \to \mathbb{N}$ and $c,s,\ell \colon \mathbb{N} \to \mathbb{R}^{+}$
  be polynomial time computable functions.
  A promise problem $\mathcal{L}_{\textup{yes}},\mathcal{L}_{\textup{no}} \subseteq \set{0,1}^*$ is in $\QMA_\ell^+(k,c,s)$ if there exists a $\BQP$
  verifier $V$ such that for every $n \in \mathbb{N}$ and every $x \in \set{0,1}^n$,
  \begin{itemize}[label={}]
    \item \textbf{\textup{Completeness:}} If $x \in \mathcal{L}_{\textup{yes}}$, then there exist unentangled states $\ket{\psi_1}, \ldots, \ket{\psi_{k(n)}}$, each on at most $\ell(n)$ qubits and with real non-negative amplitudes,
                                        s.t.\ $\Pr[V(x, \ket{\psi_1} \otimes \cdots \otimes \ket{\psi_{k(n)}})\textup{ accepts}] \ge c(n)$.
\vspace{\parsep}
    \item \textbf{\textup{Soundness:}} If $x \in \mathcal{L}_{\textup{no}}$, then for every unentangled states $\ket{\psi_1}, \ldots, \ket{\psi_{k(n)}}$, each on at most $\ell(n)$ qubits and with real non-negative amplitudes,
                                        we have \ $\Pr[V(x, \ket{\psi_1} \otimes \cdots \otimes \ket{\psi_{k(n)}})\textup{ accepts}] \le s(n)$.
  \end{itemize}
\end{definition}

In our work, we are only interested in 
\begin{align*}
&\QMA^+_{\log}(2) := \bigcup_{c-s = \Omega(1)}\QMA^+_{O(\log n)}(2, c,s),\\
&\QMA^+(2): = \bigcup_{i \in \N, ~c-s = \Omega(1)}  \QMA^+_{O(n^i)}(2,c,s).
\end{align*}
Instead of having only $2$ provers, it is much more convenient to consider $k$ provers for some  large constant $k$. This is without loss of generality, as
\cref{thm:HM} generalizes to $\QMA^+(k)$ as well. Because the product test works for general states, it in particular works
for states with non-negative amplitudes. 
Furthermore, the closest product state to a state with non-negative amplitudes also has non-negative amplitudes. 
\begin{theorem}
For any $\ell, k, 0\le s < c \le 1$,
\[
\QMA^+_\ell(k, c,s) \subseteq \QMA^+_{k\ell}(2, c',s'),
\]
where $c'=(1+c)/2$ and $s'=1-(1-s)^2/100$.
\end{theorem}
As a result, as long as the $\QMA^+(k,c,s)$ protocol is such that $c >1- (1-s)^2 /50$, it can be converted back to a $\QMA^+(2)$ protocol with a constant gap. The condition that $c >1- (1-s)^2 /50$ is also not much of an issue, since by a repetition involving more provers, we can amplify any constant $(c, s)$ gap to a $(1-\epsilon, \delta)$ gap
for $\epsilon,\delta$ close to 0.
In the remainder of the paper, we will use constantly many proofs without further referring to this result.

\subsection{Trace Distances}
A standard notion of distance for quantum states is that of the \emph{trace distance}. The trace distance between $\psi$ and $\phi$, denoted $\TD(\psi, \phi)$, is 
\begin{align}
\frac{1}{2}\|\psi-\phi\|_1 = \frac{1}{2}\trace\sqrt{( \psi - \phi)^\dagger (\psi-\phi)}.\label{eq:trace-dist}
\end{align}
We also use the notation $\TD(\ket\psi, \ket\phi)$ if we want to emphasize that $\psi$ and $\phi$ are pure states. The following fact provides an alternative definition for trace distance between pure states.
\begin{fact}\label{fact:trace-innerproduct}
  The trace distance between $\ket{\phi}$ and $\ket{\psi}$ is given by $\TD(\ket{\phi},\ket{\psi}) = \sqrt{1-\abs{\braket{\phi}{\psi}}^2}$.
\end{fact}
The trace distance remains small under the tensor product.
\begin{fact}  \label{fact:trace-tensor} Let $\ket {\psi_0}, \ket{\phi_0} \in \unitC_n$ and $\ket{\psi_1},  \ket{\phi_1} \in \unitC_m$ for arbitrary $n,m\in\N$. Then
\[
\TD(\ket {\psi_0} \otimes \ket {\psi_1} , \ket{\phi_0} \otimes \ket{\phi_1})^2
\le \TD(\ket{\psi_0}, \ket{\phi_0})^2 + \TD(\ket{\psi_1}, \ket{\phi_1})^2.
\]
\begin{proof}
By the alternative definition of the trace distance,
    \begin{align}
        \TD(\ket {\psi_0} \otimes \ket {\psi_1} , \ket{\phi_0} \otimes \ket{\phi_1})^2
        & = ( 1 - |\langle \psi_0, \phi_0 \rangle|^2|\langle \psi_1, \phi_1 \rangle|^2) \nonumber \\
        &\le (1-|\langle \psi_0, \phi_0 \rangle|^2+1-|\langle \psi_1, \phi_1 \rangle|^2) \nonumber \\
        & = \TD(\ket{\psi_0}, \ket{\phi_0})^2 + \TD(\ket{\psi_1}, \ket{\phi_1})^2,\nonumber
    \end{align}
where the second step can be easily verified as $-a^2b^2 + b^2 \le 1-a^2$ for any $a,b\in[0,1]$.
\end{proof}
\end{fact}
Two states with small trace distance are indistinguishable to quantum protocols.
\begin{fact}\label{fact:trace_norm_acc}
  If a quantum protocol accepts a state $\ket{\phi}$ with probability at most $p$, then it accepts $\ket{\psi}$  with
  probability at most $p + \TD(\ket{\phi},\ket{\psi})$.
\end{fact}

We will use the well-known swap test to compare unentangled quantum states.

\begin{fact}[Swap Test]
  Let $\ket{\phi}$ and $\ket{\psi}$ be two quantum states on the same Hilbert space.
  Then the acceptance probability of SwapTest$(\ket{\phi}, \ket{\psi})$ is
  \begin{align*}
    \frac{1}{2} + \frac{\abs{\braket{\phi}{\psi}}^2}{2}\mper
  \end{align*}
\end{fact}

We can equivalently state the acceptance probability of the swap test
in terms of the trace distance as follows.

\begin{remark}\label{remark:swap_test}
  Any two quantum states $\ket{\phi}$ and $\ket{\psi}$ pass the swap test with probability $ 1-\frac{1}{2}\TD(\ket{\phi},\ket{\psi})^2$.
\end{remark}

We record the following elementary facts. They are special cases of trace distance for states with nonnegative amplitudes.
\begin{claim}\label{claim:correlation-dist}
Let $u,v,z\in\unitR_d$ for any natural number $d$. Let $\epsilon>0$
be some small real constant.
\begin{enumerate}
    \item \label{enu:correlation-closeness}(Closeness preservation) If 
        $\langle u,v\rangle ^{2}\ge 1-\epsilon$.
    Then
    \[
        \left| \langle u,z\rangle^{2}- \langle v,z\rangle^{2}\right|\le 3 \sqrt{\epsilon}.
    \]
    \item \label{enu:correlation-Triangle-inquality}(Triangle inequality) If
        $ \langle u,z\rangle^{2}\ge1-\epsilon,$ and $\langle v,z\rangle^{2}\ge1-\epsilon.$
    Then
    \[
        \langle u,v\rangle^{2}\ge1-2\epsilon.
    \]
\end{enumerate}
\end{claim}

\begin{proof}
The first item is bounded as below
\begin{align*}
|\langle u,z\rangle^{2}-\langle v,z\rangle^{2}| & =|\langle u-v,z\rangle|\cdot|\langle u,z\rangle+\langle v,z\rangle|\\
 & \le2\|u-v\|\\
 & \le 2\sqrt{2-2\langle u,v\rangle}\\
 & \le2\sqrt{2-2\sqrt{1-\epsilon}}\\
 & \le3\sqrt{\epsilon},
\end{align*}
where the last step can be verified by elementary calculus.

Next, we prove the second item as follows
\begin{align*}
\langle u,v\rangle^{2} & =\left(\frac{2-\|u-v\|^{2}}{2}\right)^{2}\\
 & \ge\left(\frac{2-\|u-z\|^{2}-\|v-z\|^{2}}{2}\right)^{2}\\
 & =(\langle u,z\rangle+\langle v,z\rangle-1)^{2}\\
 & \ge(2\sqrt{1-\epsilon}-1)^{2}\\
 & =5-4\epsilon-4\sqrt{1-\epsilon}\\
 & \ge1-2\epsilon,
\end{align*}
where the last step holds because $\sqrt{1-\epsilon}\le1-\epsilon/2$.
\end{proof}

\subsection{Expander Graphs}

Let $G=(V,E)$ be a $d$-regular graph. For non-empty sets $S,T
\subseteq V$, we denote by $E(S,T)$ the following set of edges $E(S,T)
= \{(x,y) \in E \mid x \in S, y \in T\}$.\footnote{The graphs are usually undirected. In this case, $E(S, S)$ actually counts the same edge twice by the definition.}
The edge expansion  of a non-empty $S \subseteq V$, denoted $\Phi_G(S)$, is
defined as
\begin{align*}
  \Phi_G(S) \coloneqq \frac{\abs{E(S,V\setminus S)}}{d\abs{S}}\mcom
\end{align*}
and it is a number in the interval $[0,1]$. For $S \subseteq V$, we
refer to relative size $\abs{S}/\abs{V}$ as the \emph{measure} of $S$.
A closely related notion called Cheeger constant for $G$, is defined as
\[
    \min_{S\subseteq G : |S|\le |G|/2} \frac{|E(S, V\setminus S)|}{|S|}.
\]

\iffalse
\begin{definition}[$(\eta,\delta)$-SSE graph]
  Let $\eta,\delta \in (0,1)$.
  %
  We say that $G$ is a $(\eta,\delta)$ small set expander, or simply
  $(\eta,\delta)$-SSE for short, if for every $\emptyset \ne S \subseteq V$ of size $\abs{S}
  \le \delta \abs{V}$ we have $\Phi_G(S) \ge 1 - \eta$.
\end{definition}
\fi

% Protocol overview
\section{Overview of Global Protocols}\label{sec:strategy}

We now give an overview of our \emph{global} protocols for SSE
in~\cref{sec:overview:sse}, for UG in~\cref{sec:overview:ug} and for
NEXP in~\cref{sec:overview:nexp}. As alluded earlier, a key insight of
these protocols is the manipulation of quantum proofs in a
\emph{global} and \emph{coherent} way in order to achieve a
\emph{constant} gap. For the problems considered here, there is always
an underlying graph to the problem whose edge set can be (or almost) decomposed into
perfect matchings. Taking advantage of this collection
of perfect matchings will be one of the aspects in allowing for a
\emph{global} manipulation of the quantum proofs in these protocols. 
It will be more convenient to design
protocols with constantly many unentangled proofs rather than just two.
Recall that due to the result of Harrow and Montanaro~\cite{HM13}, these
protocols can be converted into two-proof protocols with
a constant multiplicative increase in the proof size.

\subsection{Small Set Expansion Protocol}\label{sec:overview:sse}

We provide an overview of the SSE protocol in $\QMA_{\log}^+(2)$ with
a \emph{constant} gap from~\cref{sec:SSE-protocol}. Suppose that we are
given an input $n$-vertex graph $G$ on the vertex set $V$. Our goal is
to decide whether $G$ is a yes or no instance of
$(\eta,\delta)$-SSE. Recall that, in the yes case, there exists a set
$S$ of measure $\delta$, such that $S$ essentially does not expand,
\ie $\Phi_G(S) \le \eta \approx 0$. Nonetheless, in the no case, every
set $S$ of measure at most $\delta$ has near-perfect expansion, \ie
$\Phi_G(S) \ge 1-\eta \approx 1$.

In the design of the protocol, we are allowed two \emph{unentangled}
proofs on $O_{\eta,\delta}(\log(n))$ qubits. It is natural to ask for
one of these proofs to be a state $\ket{\psi}$ ``encoding'' a uniform
superposition of elements of a purported non-expanding set $S$ of the
form
\begin{align*}
  \ket{\psi} ~=~ \frac{1}{\sqrt{S}} \sum_{i \in S} \ket{i}\mper
\end{align*}
We now check the non-expansion of the support set of $\ket{\psi}$ as 
follows.
Suppose we could apply the adjacency matrix $A$ of $G$ directly to the
vector $\ket{\psi}$. While $A$ is not necessarily a valid quantum operation,
it will not be difficult to resolve this issue later.
If we are in the yes case and the support of $\ket{\psi}$ indeed
encodes a non-expanding set, we would have $\supp(A\ket{\psi}) \cap
\supp(\ket{\psi}) \approx \supp(\ket{\psi})$.
However, if we are in the no case, provided the size of
$\supp(\ket{\psi})$ is small (at most a $\delta$ fraction of the
vertices), the small set expansion property of $G$ would imply
$\supp(A\ket{\psi}) \cap \supp(\ket{\psi}) \approx \emptyset$.

How can we check the support conditions above? For this, suppose that
we have not only one copy of $\ket{\psi}$ but rather two equal
\emph{unentangled} copies $\ket{\psi_1} = \ket{\psi_2}$. We apply $A$
to $\ket{\psi_1}$ and then measure the correlation between
$A\ket{\psi_1}$ and $\ket{\psi_2}$.  In the yes case, the two vectors
are almost co-linear, whereas in the no case they are almost 
orthogonal. It is well-known that co-linearity versus orthogonality 
of two \emph{unentangled} quantum states can be tested via the swap test.

We now address the issue that the adjacency matrix $A$ may not be a
unitary matrix, and hence it is not necessarily a valid quantum
operation. Nonetheless, the adjacency matrix of a $d$-regular graph
can always be written as a sum of $d$ permutation matrices
$P_1, \ldots, P_d$, which are special unitary matrices. In terms of the
support guarantees described above, it is possible to show that
applying one of these permutation matrices uniformly at random in the
protocol leads to a similar behavior as applying $A$.

In the yes case, it can be shown that all goes well with the above
strategy. However, in the no case, things become more delicate
starting with the fact that $\ket{\psi}$ is an arbitrary adversarial
state of the form
\begin{align*}
  \ket{\psi} =  \sum_{i \in S} \alpha_i \ket{i}\mcom
\end{align*} 
where we have no control over the amplitudes $\alpha_i$'s magnitudes and
phases.

One important issue is that the support of $\ket{\psi}$ may not be
small (\ie at most a $\delta$ fraction), and the graph $G$ may have
large non-expanding sets even in the no case. We design a sparsity
test to enforce that its support is indeed small. The soundness of
this sparsity test takes advantage of the non-negative amplitudes
assumption to achieve dimension-independent parameters and this is the
only test of the protocol that relies on the non-negative assumption.
This points to a very natural question in quantum property testing:
how efficiently can we test sparsity\footnote{For this task, we can
have multiple unentangled copies of the state to be tested as well
multiple unentangled proofs to help the tester.} with the help of a
prover in the general amplitude case?

In our protocol, the support conditions from above are actually
checked by considering the average magnitude of the overlap between
$P_r\ket{\psi}$ and $\ket{\psi}$. This overlap governs (part of) the
acceptance probability of the protocol which can be bounded as
\begin{align*}
  \Ex{r \in [d]}{\abs{\braket{P_r \psi}{\psi}}} \le \frac{1}{d} \sum_{i,j} A_{i,j} \abs{\alpha_i} \abs{\alpha_j} = \frac{1}{d} \braket{A \abs{\psi}}{\abs{\psi}}\mcom
\end{align*}
where $\ket{\abs{\psi}} = \sum_{i \in S} \abs{\alpha_i} \ket{i}$. With
this bound, phases are no longer relevant. 
%}

A second important and more
delicate issue is that the magnitude of the amplitudes $\alpha_i$'s of
$\ket{\psi}$ may be very far from flat. By definition, the SSE
property of the graph $G$ only states that for every ``flat''
indicator vector $\one_S$, where $S$ is any vertex set of measure at
most $\delta$, we have
\begin{align*}
  \frac{1}{d} \braket{A \frac{\one_S}{\sqrt{\abs{S}}}}{\frac{\one_S}{\sqrt{\abs{S}}}} ~\approx_{\eta,d}~ 0\mper
\end{align*}
Nonetheless, in order to not be fooled by the provers, we need a stronger \emph{analytic} condition
\begin{align*}
  \max_{u \colon \norm{u}_2 =1, \abs{\supp(u)} \le \delta \abs{V}} ~\frac{1}{d} \braket{A u}{u} \approx 0\mcom
\end{align*}
where $u$ ranges over arbitrary unit vectors. For every disjoint set $S, T \subseteq V$ of combined measure at
most $\delta$, the SSE property of $G$ allows us to deduce
\begin{align}\label{eq:strategy_bound}
  \frac{1}{d} \braket{A \frac{\one_S}{\sqrt{\abs{S}}}}{\frac{\one_T}{\sqrt{\abs{T}}}} ~\approx_{\eta,d}~ 0\mper
\end{align}
Ideally, we would like to leverage the bounds we have for flat
indicator vectors of small sets from~(\ref{eq:strategy_bound}) to
conclude that arbitrary unit vectors of small support have a bounded
quadratic form. The seminal
work on $2$-lifts~\cite{BL06} of Bilu and Linial dealt with a similar
question, but without the sparse support conditions.  Surprisingly, they gave sufficient
conditions for this phenomenon.  Here, we prove that the same phenomenon
also happens for the sparse version of the problem. In particular,
this shows that SSE graphs satisfy the more ``robust''
\emph{analytic} SSE property. Using this robust property, we
conclude the soundness of the protocol.

\subsection{Unique Games Protocol}\label{sec:overview:ug}

We provide an overview of the UG protocol in $\QMA_{\log}^+(2)$ with
a \emph{constant} gap from~\cref{sec:ug-protocol}. Suppose that we are
given an input UG instance with alphabet $\Sigma$, namely, an
$n$-vertex $d$-regular graph $G = (V, E)$, where each
directed\footnote{The reverse edge of $e$ is typically associated with
  the constraint $f_e^{-1}$.} edge $e \in E$ is associated with a
permutation constraint $f_e \colon \Sigma \to \Sigma$. We say that
an assignment $\ell \colon V \to \Sigma$ satisfies an edge $e=(i,j)$
if $f_e(\ell(i)) = \ell(j)$. This means that for each assigned
value for $i$ there is a unique value for $j$ and vice versa
satisfying the permutation constraint of edge $e$. The goal is to
distinguish between (yes) there exists an assignment satisfying at
least $1-\eta$ fraction of the constraints, and (no) every assignment
satisfies at most a $\delta$ fraction of constraints.

In the yes case, the protocol expects from the unentangled provers
copies of a quantum state $\ket{\psi}$ encoding an assignment $\ell$
of value at least $1-\eta$ of the form
\begin{align}\label{eq:ug_assignment_state}
  \ket{\psi} = \sum_{i=1}^n \frac{1}{\sqrt{n}} \ket{i}\ket{\ell(i)}\mper
\end{align}
We will again explore the underlying graph structure of the problem to
make the proof verification \emph{global} leading to a constant gap.
Similarly to the SSE protocol, we will also use the fact that the
adjacency matrix $A$ of a $d$-regular graph can be written as a sum of
$d$ permutation matrices $P_1, \ldots, P_d$ and these matrices are 
special cases of unitary operators. Using a permutation matrix $P_r$
and the UG constraints, we will define a unitary operator $\Pi_r$
intended to help us check the constraints along the edges of
$P_r$. Each $\Pi_r$ is defined as follows
\begin{align*}
  \Pi_r \ket{i} \ket{a} \mapsto \left( P_r \ket{i} \right) \ket{f_{(i, P_r i)}(a)}\mcom
\end{align*}
where $i$ ranges in $V$ and $a$ ranges in $\Sigma$. The crucial
observation is that if the constraints along the edges of $P_r$ are
almost fully satisfied by $\ell$, we should have $\ket{\psi} \approx
\Pi_r \ket{\psi}$ whereas if they almost fully unsatisfied by $\ell$,
we should have $\ket{\psi}$ almost orthogonal to $\Pi_r \ket{\psi}$. By
sampling a uniformly random $\Pi_r$ and checking this approximate
co-linearity versus orthogonality property, we obtain a \emph{global}
test to check if an assignment is good.

In the no case, there is no reason the adversarial provers
will provide proofs of the form~(\ref{eq:ug_assignment_state})
encoding a valid assignment. In general, we will have an arbitrary state of
the form
\begin{align*}
  \ket{\psi} = \sum_{i=1}^n \alpha_i \ket{i} \left(\sum_{a \in \Sigma} \beta_{i,a} \ket{a} \right)\mper
\end{align*}
There are two main issues. First, the adversary can omit the
assignment to several vertices by making $\alpha_i \approx 0$.
Second, even if all the vertices are present in the superposition with
amplitudes $\alpha_i = 1/\sqrt{n}$, the prover can assign a
superposition of multiple values to each position as in
\begin{align*}
  \ket{\psi} = \sum_{i=1}^n \frac{1}{\sqrt{n}} \ket{i} \left(\sum_{a \in \Sigma} \beta_{i,a} \ket{a} \right)\mper
\end{align*}
Fortunately, both of these issues can be handled in a global way. In
addressing the second issue, we currently rely on the non-negative
amplitudes assumption. To give a flavor of why non-negative amplitudes
can be helpful, consider the following simplified scenario that $\Sigma
= \set{0,1}$ and
\begin{align*}
  \ket{\psi} = \sum_{i=1}^n \frac{1}{\sqrt{n}} \ket{i} \left(\frac{1}{\sqrt{2}} \ket{0} + \frac{1}{\sqrt{2}}\ket{1} \right)\mper
\end{align*}
Suppose that we measure the second register (containing the values in
$\Sigma$) of two copies of $\ket{\psi}$ obtaining $\ket{0}$ and
$\ket{1}$, and let $\ket{\psi_0}$ and $\ket{\psi_1}$ be the resulting
states on the first register containing the indices of the vertices,
respectively. In the ideal scenario, if each vertex has a single well
defined value in $\ket{\psi}$ (which is not the case in this example),
we should have $\ket{\psi_0} \perp \ket{\psi_1}$. If not (as in this
toy example), the supports of $\ket{\psi_0}$ and $\ket{\psi_1}$ are
not disjoint. With non-negative amplitudes, if there is substantial
``mass'' in the intersection of their supports, then this condition
can be tested using a swap test since $\braket{\psi_0}{\psi_1}$ will
be large (in this toy example it is 1 as $\ket{\psi_0}= \ket{\psi_1} =
\sum_{i=1}^n 1/\sqrt{n} \ket{i}$).

With this UG protocol and the recent proof\footnote{Coming from the
  proof of the $2$-to-$2$ conjecture.} of the NP-hardness of deciding
UG with parameters $\eta = 1/2$ and $\delta > 0$ an arbitrarily small
chosen constant, we can deduce that $\NP \subseteq \QMA_{\log}^+(2)$.

\subsection{PCP Verification Protocol for NEXP}\label{sec:overview:nexp}

We provide an overview of the NEXP protocol in $\QMA^+(2)$ with
\emph{constant} gap from~\cref{sec:nexp}. Recall that scaling up to
$\QMA(2)$ the previous protocols for $\QMA_{\log}(2)$ from literature
leads to exponentially small gaps which are intolerable to $\QMA(2)$.
This motivates our study of \emph{constant} gap protocols for hard
problems in $\QMA_{\log}^+(2)$. Our new constant gap protocols can be
indeed scaled up to $\QMA^+(2)$ and the gap remains constant! Another
issue unresolved in the previous work is that if we scale up the protocol
naively, the running time of the verifier becomes exponential and this
is also intolerable to $\QMA(2)$ (or $\QMA^+(2)$) which requires a
polynomial-time $\BQP$ verifier. Simultaneously achieving a constant gap with a
polynomial-time verifier is quite interesting since this requires
considering very efficient forms of quantum proof verification.

Classically, it is known that NEXP admits polynomial-time proof
verification protocols with a constant gap, \ie very efficient
PCPs. Note that the proof size is exponentially large in the input
size and the verification runs in \emph{polylogarithmic} time in the
size of the proof.  These protocols manipulate exponentially large
objects given in very succinct and explicit forms. We will build on some
of these PCPs results to design our $\QMA^+(2)$ protocol for NEXP, but
our \emph{global} verification of quantum proofs will require even
stronger explicitness and regularity properties of these objects. In
this work, we prove these additional properties by carefully
investigating the composition of known PCP constructions.

A PCP protocol naturally gives rise to a label cover CSP (via a simple
and standard argument). We give a \emph{global} $\QMA^+(2)$ protocol
for label cover arising from the PCP for NEXP with the additional
explicitness and regularity properties alluded above. Recall that a label
cover instance is given by a bipartite graph $G=(L\sqcup R,E)$ with a
left and right vertex partitions $L$ and $R$, left and right alphabets
$\Sigma_L$ and $\Sigma_R$ and constraints $f_e \colon \Sigma_L \to
\Sigma_R$ on the edges $e \in E$. Given assignments to the left and
right partitions $\ell_L \colon L \to \Sigma_L$ and $\ell_R \colon
R \to \Sigma_R$, a constraint on edge $e=(i,j)$ is satisfied if
$f_e(\ell_L(i))=\ell_R(j)$. In this correspondence of PCP and
label cover, the left vertices correspond to the constraints of the PCP
verifier and the right vertices correspond to the symbols of the proof
which are the variables in the PCP constraints.

We now give an abstract simplified description of our protocol to
convey some intuition and general ideas. The precise protocol is
actually more involved and somewhat different (see~\cref{sec:nexp} for
its full description). In the yes case our $\QMA^+(2)$ protocol
expects to receive copies of the state $\ket{\psi_L}$ and from it
obtain copies of a state similar to $\ket{\psi_R}$ both described below
\begin{align}\label{eq:prof_over_nexp}
  \ket{\psi_L} = \sum_{i \in L} \frac{1}{\sqrt{\abs{L}}} \ket{i} \ket{\ell_L(i)} \qquad \textup{and} \qquad \ket{\psi_R} = \sum_{j \in R} \frac{1}{\sqrt{\abs{R}}} \ket{j} \ket{\ell_R(j)}.
\end{align}
Note that the left assignment $\ell_L$ specifies the values of all
variables appearing in each PCP constraint, and $\ell_R$ specifies
the values of variables appearing in the PCP proof. In this case,
checking the constraints (essentially) amounts to testing consistency
of these various assignments to the variables. To accomplish this
goal, we design two operations\footnote{We stress that this is a
simplistic view of the protocol. See~\cref{sec:nexp} for the precise
technical details.}  $\Gamma_L$ and $\Gamma_R$ such
that,\footnote{Assuming $\ket{\psi_L}$ and $\ket{\psi_R}$ are of the
  above form.} if the label cover instance is fully satisfiable (with
$\ell_L$ and $\ell_R$), then $\Gamma_L(\ket{\psi_L}) \approx
\Gamma_R(\ket{\psi_R})$, otherwise $\Gamma_L(\ket{\psi_L})$ will be
approximately orthogonal to $\Gamma_R(\ket{\psi_R})$. In a vague
sense, $\Gamma_L$ tries to extract the value of some variables in the
constraints and $\Gamma_R$ tries to replicate the values of each
variable in a quantum superposition so that $\Gamma_L(\ket{\psi_L})$
and $\Gamma_R(\ket{\psi_R})$ become equal if $\ell_L,\ell_R$ are
fully satisfying assignments and they become close to orthogonal if
the CSP instance is far from satisfiable (regardless of
$\ell_L,\ell_R$). At a high level, there is some
parallel\footnote{As in the SSE and UG protocols, there is also
  distribution on pairs of operator $(\Gamma_L, \Gamma_R)$ here.} with
the SSE and UG protocols. There, we had $\ket{\psi_L} = \ket{\psi_R}$,
$\Gamma_L$ being the identity and $\Gamma_R$ being either $P_r$ (in
SSE) or $\Pi_r$ (in UG).

A crucial point is that to make the operations $\Gamma_L$ and
$\Gamma_R$ efficient, we need to be able to determine (1) the
neighbors of any given vertex in $L$ in polynomial time, and (2) the
neighbors of any given vertex in $R$ in polynomial time. We call an
instance satisfying (1) and (2) \emph{doubly explicit}. While
(1) follows easily from the definition of PCP, to get property (2) we
need to carefully compose known PCP protocols and prove that this
property holds.

Similarly to the UG protocol, we also need to check that the quantum
proofs are close to a valid encoding of an assignment to the
variables. The provers should not (substantially) omit the values of
variables nor provide a superposition of multiple values for the same
variable. Similarly, checking this second condition is the part of the
protocol that currently relies on non-negative amplitudes.

% symmetry test + sparsity test
\section{Property Testing Primitives}\label{sec:prop_test}
In this section, we prove some property testing primitives that we will use as the building blocks in designing protocols for general problems.

The first test is the \emph{symmetry} test. In many situations, we ask the prover to provide a supply of constantly many copies of a state. 
To make sure that all copies are approximately the same state, the symmetry test will be invoked. The symmetry test in general can be applied in any quantum protocol. A similar symmetry test has been considered previously in~\cite{ABDFS08}. Here we provide a stronger version.

The second test is the \emph{sparsity} test. Consider the scenario where we ask the prover to provide a state that is supposed to be some
\emph{subset state}. In particular, let $\SSS_{\gamma}\subseteq \C ^n$ be the set of subset state spanning a $\gamma$
fraction of computational basis, i.e.,
\[
\SSS_{\gamma}:=\left\{ \frac{1}{\sqrt{\gamma n}}\sum_{i\in S}|i\rangle:S\subseteq[n],|S|=\gamma n\right\} .
\]
We call $\gamma$ the \emph{sparsity} of the subset state in $\SSS_\gamma$. The sparsity test is used to determine whether a given state is close to $\SSS_\gamma$. Our sparsity test relies on the fact that the amplitudes of the quantum proofs are non-negative.

The third test is the \emph{validity} test. A natural quantum proof for many problems like the 3-SAT or 3COLOR problem is to put the variables/vertices together with their values/colors in superpositions. For example, for 3-SAT on $n$ variables, such that variable $i$ has value $x_i$, a valid proof should look like
\[
    \qphi = \frac{1}{\sqrt n} \sum_{i\in [n]} \ket i \ket {x_i}.
\]
This can be generalized for an arbitrary set of variables $X$ and an arbitrary value domain $\Sigma$ of the variables. Then the valid set would be
\begin{equation*}
    \cV = \left\{ 
        \frac{1}{\sqrt {|X|}} \sum_{i\in X} \ket i \ket {x_i}: \forall i\in X, x_i \in \Sigma
    \right\}.
\end{equation*}
The validity test tells whether a given state is close to a valid state. Our validity test works only in the situation when the given state is close to a state in $\SSS_{|\Sigma|^{-1}}$, which is guaranteed by the sparsity test. Thus, this validity test does not generalize.

\subsection{\texorpdfstring{$\epsilon$}{epsilon}-tilted States}
Before we discuss the tests, let's make the following definition first.
\begin{definition}[$\epsilon$-tilted states]
A family of states $|\psi_1\rangle, |\psi_2\rangle, \ldots, |\psi_k \rangle$ defined on a same space is an $\epsilon$-tilted state
if there is a subset $R\subseteq[k]$ such that $|R|\ge (1-\epsilon)k$ and for any $i,j \in R$,
\iffalse
\[
|\langle \psi_i \mid \psi_j \rangle|^2 \ge 1-\epsilon.
\]
\fi
\[
\TD(|\psi_i\rangle, |\psi_j\rangle) \le \sqrt\epsilon.
\]
Furthermore,  we call $ \ket{\psi_i}$ a \emph{representative state} for any $i\in R$, and the subset $\{ |\psi_i\rangle: i\in R\}$ the
\emph{representative set}.
\end{definition}

Note that a 0-tilted state is simply a set of equal states, and any $\epsilon$-tilted state is also a $\delta$-tilted state for any $\delta>\epsilon$. The name $\epsilon$-tilted state may be confusing. Our
message is that instead of treating this object as a set of states, we should simply treat them as a single state conceptually (for example, think of it as a representative state tilted a little bit). As we will see later in \cref{sec:symmetry-test},
when the symmetry test passes, we are supplied with an $\epsilon$-tilted state with high probability. Having a large number of (almost) equal states is very convenient, therefore we
always take advantage of the symmetry test and work with $\epsilon$-tilted states.
We reserve the capital letters, i.e., $|\Psi\rangle$ or simply $\Psi$,\footnote{In this paper, we never use the density operator, so there should be no confusion.}  to denote an $\epsilon$-tilted
state. The \emph{size} of $\Psi$, denoted $|\Psi|$, is the size of $\Psi$ viewed as a set
of states. 

The tilted states tensorize. In particular, for two sets of states $\Psi=\{\ket{\psi_1}, \ket{\psi_2},\ldots, \ket{\psi_k} \}$ and $\Phi=\{\ket{\phi_1},\ket{\phi_2},\ldots,\ket{\phi_k}\}$ of the same size, let $\Psi \otimes \Phi$ denote the set of states $ \{ \ket{\psi_1, \phi_1}, \ldots, \ket{\psi_k, \phi_k}\}$ (if there is not a default order,  the order can be set arbitrarily). 
\begin{proposition}[Tensorization of tilted states]\label{prop:tensor-tilted-states}
If $\Psi$ is an $\epsilon$-tilted state and $\Phi$ is a $\gamma$-tilted state, and $|\Psi|=|\Phi|=k$. Then $\Psi\otimes\Phi$ is an $(\epsilon+\gamma)$-tilted state.
\end{proposition}
\begin{proof}
Let $S$ and $T$ be the representative set of $\Psi$ and $\Phi$, respectively. Simply note that
\[
    |S \cap T|\ge (1-\epsilon-\gamma)k,
\]
and for any $i, j\in S \cap T$, 
\[
    \TD(\ket{\psi_i}\otimes\ket{\phi_i}, \ket{\psi_j}\otimes\ket{\phi_j})^2 \le \TD(\ket{\psi_i},\ket{\phi_i})^2 + \TD(\ket{\psi_j}, \ket{\phi_j})^2 \le \epsilon + \gamma,
\]
where the first inequality is due to~\cref{fact:trace-tensor}.
\end{proof}

As commented earlier that we should treat an $\epsilon$-tilted state as a single state conceptually. Now we make this comment more formal. When we apply some quantum algorithm $\cA$ to $\Psi$, we mean apply $\cA$ to all the states in $\Psi$. For any $f:\C^n \to \C$, when we evaluate $f$ on $\Psi$, we mean the expected value of $f$ on all states in $\Psi$, i.e.,
\[
    f(\Psi) = \Exp_{\ket{\psi}\in\Psi} [f(\ket{\psi})].
\]
\begin{proposition}\label{prop:tilted-state-approx}
For any quantum algorithm $\cA$, let $\cA(\ket{\psi})$ denote the probability that $\cA$ accepts $\ket{\psi}$. Let $\Psi$ be an $\epsilon$-tilted state, and $\ket{\psi}$ any representative state of $\Psi$. Then
\begin{equation}\label{eq:approx-tilted-state}
    |\cA(\ket{\psi}) - \cA(\Psi)| \le 3\sqrt \epsilon. 
\end{equation}
Furthermore, when apply $\cA$ to $\Psi$, let $\alpha$ be the fraction of accepted executions of $\cA$. Then
\begin{equation}\label{eq:concentration-tilted-state}
    \Pr[ |\alpha - \cA(\Psi)| \ge \sqrt \epsilon] \le \exp(-\epsilon|\Psi|/2),
\end{equation}
and therefore,
\begin{equation}\label{eq:approx-reprensetative-state}
    \Pr[ |\alpha - \cA(\ket{\psi})| \ge 4\sqrt \epsilon] \le \exp(-\epsilon|\Psi|/2).
\end{equation}
\end{proposition}
\begin{proof}
Let $\Psi = \{\ket{\psi_1}, \ket{\psi_2}, \ldots, \ket{\psi_k} \}$ and $S$ be the representative set for $\Psi$. Then
\begin{align*}
    \cA(\Psi) &= \frac{1}{k} \sum_{i=1}^k \cA(\ket{\psi_i})
              = \frac{|S|}{k} \Exp_{i\in S} \cA(\ket{\psi_i}) + \frac{k-|S|}{k} \Exp_{i\not\in S} \cA(\ket{\psi_i}).
\end{align*}
It follows that
\[
    (1-\epsilon)\Exp_{i\in S} \cA(\ket{\psi_i}) \le \cA(\Psi) \le (1-\epsilon)\Exp_{i\in S} \cA(\ket{\psi_i}) + \epsilon.
\]
Therefore,
\begin{equation}\label{eq:avgtypcal-expectation}
    \left| \Exp_{i\in S} \cA(\ket{\psi_i}) - \cA(\Psi) \right| \le \epsilon.
\end{equation}
By Fact~\ref{fact:trace_norm_acc} and the definition of $\epsilon$-tilted state, for any $j\in S$,
\begin{equation}\label{eq:typical-avgtypical}
    \left| \Exp_{i\in S} \cA(\ket{\psi_i}) - \cA(\ket{\psi_j}) \right| \le \sqrt\epsilon.
\end{equation}
Combining (\ref{eq:avgtypcal-expectation}) and (\ref{eq:typical-avgtypical}), we obtain (\ref{eq:approx-tilted-state}). 
The furthermore part follows by Chernoff bound.
\end{proof}
By (\ref{eq:approx-reprensetative-state}), it suffices to understand the typical behavior of the representative state in an $\epsilon$-tilted state.

\subsection{Symmetry Test}\label{sec:symmetry-test}
The symmetry test is described below.

\noindent\fbox{\begin{minipage}[t]{1\columnwidth - 2\fboxsep - 2\fboxrule}%
\uline{Symmetry Test}

\textbf{Input:} $ \Psi = \{ a_{1},a_{2},\ldots,a_{k} \}\subseteq \unitC$ for some even number $k$.
\begin{enumerate}
\item Sample a random matching $\pi$ within $1, 2, \ldots, k$.
\item SwapTest on the pairs based on the matching $\pi$.
\end{enumerate}
\emph{Accept} if all SwapTests accept. 

\end{minipage}}

\begin{theorem}[Symmetry test]\label{thm:symmetry-test}
Suppose $\Psi$ is not an $\epsilon$-tilted state.
Then the symmetry test passes with probability at most $\exp(-\Theta(\epsilon^{2}k)).$
On the contrary, for $0$-tilted state $\Psi$, the symmetry test accepts with probability $1$.
\end{theorem}

Let $\cN(i):=\{ a_j: \TD( a_{i},a_{j} ) \le \sqrt\epsilon /2\}$ be the set of vectors that are close to $a_i$, 
and $\cB:=\{i:|\cN(i)|\le k/2\}$ the set of vector $a_i$ who is far from at least half of the other vectors. Finally for a random matching define
\[
\ell(\pi)=|\{i:\pi(i)\not\in\cN(i)|,
\]
twice the number of distant pairs in the matching.

\begin{claim}
\label{claim:ell-pi}Suppose $|\cB| = \gamma k$ for any constant $ \gamma \in (0, 1 ]$.
Then
\[
\Pr_\pi \left[\ell(\pi)\ge\frac{\gamma k}{18}\right]\ge1-\exp(-\Theta(\gamma k)).
\]
\end{claim}

\begin{proof}
Without loss of generality, let $\cB=\{a_{1},a_{2},\ldots,a_{m}\}$. Assume $m \le 2k/3$, in another word, $\gamma \le 2/3$. Consider the matching procedure: For $i$ from $1$ to $k/2$, find one vector 
in $\cB$ if there is one that hasn't been matched yet, and pair it with a random unmatched
vector; if all vectors in $\cB$ have been matched, pair two random unmatched vectors.
Let $X_{i}$ be the indicator function that at time $i$, the paired
vectors are $\sqrt\epsilon /2$ far away in trace distance. Then, 
\begin{align*}
\sum_{i=1}^{\lceil m/2 \rceil}\Exp[X_{i}] 
 & \ge\frac{1}{2}+\left(\frac{k/2-2}{k}\right)+\cdots+\left(\frac{k/2-2 \lceil m/2\rceil +2}{k}\right)\\
 & =\frac{1}{2k}\left( k - 2\lceil m/2\rceil +2 \right) \lceil m / 2 \rceil \\
 & \ge \frac{1}{4k}(k-m)m   \\
 & \ge \frac{1}{4} \gamma (1-\gamma)k.
\end{align*}
Since 
$
S_{j}=\sum_{i=1}^{j}X_{i}-\Exp[X_{i}]
$
is a martingale, we have
\begin{align*}
    \Pr\left[\sum_{i=1}^{\lceil m/2 \rceil } (X_{i} - \Exp[X_{i}])\le-t\right]\le\exp\left(-\frac{t^{2}}{m+1}\right).
\end{align*}
Set $t=\gamma(1-\gamma)k/12$, our claim holds. When $\gamma> 2/3$, the claim can be verified by comparing it with the case
of $\gamma=2/3$.
\end{proof}

\begin{lemma}\label{lem:bad-points}
Suppose $|\cB|\ge\gamma k$, for any constant
$\gamma\in(0,1]$. Then the probability that the symmetry test passes with probability
at most $\exp(-\Theta(\epsilon\gamma k))$.
\end{lemma}

\begin{proof}
Fix any permutation $\pi,$ the symmetry test passes with probability
at most $(1-\epsilon/8)^{\ell(\pi)}.$ Therefore using Claim~\ref{claim:ell-pi},
we have
\begin{align*}
 & \Pr[\text{Symmetry test passes}]\\
 & \qquad\le\Pr\left[\ell(\pi)<\frac{\gamma k }{18}\right]+(1-\epsilon/8)^{\gamma k/18}\\
 & \qquad\le\exp(-\Theta(\gamma k))+\exp(-\Theta(\epsilon\gamma k)).\qedhere
\end{align*}
\end{proof}
At the point, Theorem~\ref{thm:symmetry-test} is a straightforward
corollary of the above lemma.
\begin{proof}[Proof of Theorem~\ref{thm:symmetry-test}]
Let $\cG = [k]\setminus \cB$. Note that for any $i,j\in \cG $, 
\[
\cN(i)\cap\cN(j)\not=\varnothing.
\]
Thus $\TD(a_i, a_j) \le \sqrt\epsilon $ by triangle inequality. Thus, $|\cG|\ge(1-\epsilon)k$ implies that $\Psi$ is an $\epsilon$-tilted state. By contraposition, if $\Psi$ is not an $\epsilon$-tilted state, then $|\cB| > \epsilon k$. It follows that, by Lemma~\ref{lem:bad-points}, the symmetry test passes with probability at most $\exp(-\Theta(\epsilon^2 k))$.
\end{proof}

\subsection{Sparsity Test}
Now we move on to the sparsity test, where the non-negative assumption is used crucially. In the sparsity
test, aside from the state that we want to test whether it's close to some subset state, the prover will provide an auxiliary proof to assist the verifier. 

In what follows, we provide two versions of the sparsity tests. In the first version, we want to know if a given state $\qpsi{}$ is close to some subset state without prior knowledge of the sparsity $\gamma$. In the second version, there is a target sparsity $\gamma$, and we want to know if $\qpsi{}$ is close to $\SSS_\gamma$. We describe the first version below.

\noindent\fbox{\begin{minipage}[t]{1\columnwidth - 2\fboxsep - 2\fboxrule}%
\uline{Sparsity test \rom{1} (with precision \mbox{$\epsilon$})}

\textbf{Input:} $\Psi = \{u_1, \ldots, u_{2k}\}\subseteq \unitR,\Phi=\{v_1, \ldots, v_{2k} \}\subseteq \unitR.$

Partition $\Psi$ into $\Psi_0$ and $\Psi_1$ of equal size, and partition $\Phi$ into $\Phi_0$ and $\Phi_1$ of equal size.
\begin{enumerate}
    \item \label{enu:swap-u-1} SwapTest on ($\Psi_0, \one_{[n]}/\sqrt{n}$); 
    \item \label{enu:swap-v-1} SwapTest on ($\Phi_0, \one_{[n]}/\sqrt{n}$);
    \item \label{enu:swap-u-v} SwapTest on ($\Psi_1, \Phi_1$) . 
\end{enumerate}
\emph{Accept} if and only if $\alpha + \beta \in [3/2-\sqrt{\epsilon}, 3/2+\sqrt\epsilon]$ and $\lambda \le 1/2 + \sqrt\epsilon$, where
$\alpha, \beta$ and $\lambda$ are the fractions of accepted SwapTests in~\ref{enu:swap-u-1}, \ref{enu:swap-v-1}, and~\ref{enu:swap-u-v}, respectively.

\textbf{Output: }$\alpha, \beta, \lambda$. 
\end{minipage}}

\vspace{1mm}
\begin{theorem}[Sparsity test]\label{thm:sparsity-test}
Given $\Psi=\{u_{i}\in\unitR_n\}_{i\in[2k]},\Phi=\{v_{i}\in\unitR_n\}_{i\in[2k]}$ two $\epsilon$-tilted
states for $\epsilon<1/2$. 
Let $\alpha, \beta$, and $\lambda$ be the outputs. 

(Completeness) For any $0$-tilted states $\Psi$ and $\Phi$, such that $\Psi\in \SSS_{\delta}$, $\Phi\in \SSS_{1-\delta}$, and $\Psi \perp \Phi$. Then with probability at least $1-\exp(-\Theta(\epsilon k))$ the sparsity test accepts, furthermore, 
\begin{align*}
    &|2\alpha -1 - \delta|\le \sqrt\epsilon, \\
    &|2\beta -1 - (1-\delta)| \le \sqrt\epsilon.
\end{align*}

(Soundness) The sparsity test accepts with probability at most $\exp(-\epsilon k)$, if either of the following fails to hold: 
\begin{enumerate}
\item \label{enu:sparsity-supp-upper}There is $S\subseteq[n]$, such that
for any $\gamma>0,$
\[
|S|\le(2\alpha-1)n+ 9\epsilon^{1/4}n/\gamma,
\]
and for any representative $u\in\Psi$,
\[
\|u|_{S}\|^2\ge1-\gamma-2\sqrt\epsilon.
\]
\item \label{enu:sparsity-supp-estimate}There is $S\subseteq[n]$, such
that
\[
\left||S|-(2\alpha-1)n\right|\le O(\epsilon^{1/12}(2\alpha-1)^{1/3})n,
\]
and for any representative $u\in\Psi$,
\[
 \TD \left(u,\one_{S}/\sqrt{|S|}\right) = O\left(\frac{\epsilon^{1/24}}{(2\alpha-1)^{1/3}}\right).
\]
\end{enumerate}
\end{theorem}

We first prove the following lemma useful in the soundness part.
\begin{lemma}
\label{lem:sparsity}Let $u,v\in\unitR_n$ for an arbitrary natural number $n$. Let $\delta\in(0,1)$
be some constant. If for some small constant $\epsilon>0$, the following items are true:
\begin{enumerate}
\item \label{enu:correlation-uv}$\langle u,v\rangle^{2}\le\epsilon,$
\item \label{enu:correlation-u-1}$|\langle u,\one_{[n]}/\sqrt{n}\rangle^{2}-\delta|\le\epsilon,$ 
\item \label{enu:correlation-v-1}$|\langle v,\one_{[n]}/\sqrt{n}\rangle^{2}-(1-\delta)|\le\epsilon$.
\end{enumerate}
Then, for any $0<\gamma<1/2$,
and some $|S|\le(\delta + 2\sqrt{\epsilon}/\gamma)n$,
\begin{equation}
\|u|_{S}\|^2\ge1-\gamma.\label{eq:supp-upper}
\end{equation}
Furthermore, for some $S\subseteq[n]$ with 
\[
(\delta-O(\epsilon))n\le|S|\le(\delta+O(\epsilon^{1/6}\delta^{1/3}))n
\]
we have 
\[
\langle u,\one_{S}/\sqrt{|S|}\rangle\ge1-O\left(\frac{\epsilon^{1/6}}{\delta^{2/3}}\right).
\]
\end{lemma}

\begin{proof}
Let 
\[
U=\left\{ i:u_{i}\ge\sqrt{\frac{\gamma}{n}}\right\} ,V=\left\{ i:v_{i}\ge\sqrt{\frac{\gamma}{n}}\right\} ,
\]
for some $\gamma$ to be determined later. $U$ will be the set $S$ in the statement. Note that by our definition
of $U,V$, 
\begin{align}
 & \|u|_{\bar{U}}\|^{2},\|v|_{\bar{V}}\|^{2}\le\gamma, \label{eq:uv-nonUV-mass}\\
 & \|u|_{U}\|^{2},\|v|_{V}\|^{2}\ge1-\gamma.            \label{eq:uv-UV-mass}
\end{align}
We claim that
\begin{align}
    & |U|\ge (\delta-\epsilon )n,\label{eq:supp-U-lower}\\
    & |V|\ge (1-\delta-\epsilon) n,\label{eq:supp-V}\\
    & |U\cap V|\le\frac{\sqrt{\epsilon}}{\gamma}n.\label{eq:intersection-U-V}
\end{align}
We verify (\ref{eq:supp-U-lower}), and (\ref{eq:supp-V}) will follow
the same reasoning. Note
\begin{align*}
\delta-\epsilon & \le\left\langle u|_{U},\frac{\one_{[n]}}{\sqrt{n}}\right\rangle ^{2}\le\|u|_{U}\|^{2}\frac{|U|}{n},
\end{align*}
where the first inequality is given; the second step uses Cauchy-Schwarz.
Rearranging the terms, we get (\ref{eq:supp-U-lower}). Next, we obtain (\ref{eq:intersection-U-V}),
\begin{align*}
\sqrt{\epsilon} & \ge\sum_{i\in U\cap V}u_{i}v_{i}\ge|U\cap V|\frac{\gamma}{n},
\end{align*}
where the first step uses \ref{enu:correlation-uv}, and second step follows the definition of $U$ and $V$. In view of (\ref{eq:intersection-U-V}), we are done by rearranging
the terms. By (\ref{eq:supp-V})-(\ref{eq:intersection-U-V}), we
can conclude
\begin{align}
|U| & \le|U\cup V|-|V|+|U\cap V| \nonumber \\
    & \le n - (1-\delta -\epsilon)n + \frac{\sqrt \epsilon }{\gamma}n \nonumber \\
    & \le \left(\delta + \frac{2\sqrt{\epsilon}}{\gamma}\right)n.
 \label{eq:supp-U-upper-secondary}
\end{align}
This finishes the proof of the first part of the lemma. 
For the furthermore part, 
calculate: 
\begin{align*}
\left\langle u,\frac{\one_{U}}{\sqrt{|U|}}\right\rangle  & =\frac{1}{\sqrt{|U|}}\langle u|_{U,}\one_{[n]}\rangle\\
 & = \frac{1}{\sqrt{|U|}}(\langle u,\one_{[n]}\rangle-\langle u|_{\bar{U}},\one_{[n]}\rangle)\\
 & \ge\sqrt{\frac{n}{|U|}}(\sqrt{\delta-\epsilon}-\sqrt\gamma)\\
 & \ge\sqrt{\frac{\delta-\epsilon}{\delta+2\sqrt{\epsilon}/\gamma}}-\sqrt{\frac{\gamma}{\delta+2\sqrt{\epsilon}/\gamma}}\\
 & \ge\sqrt{1-\frac{2\sqrt{\epsilon}/\gamma+\epsilon}{\delta+2\sqrt{\epsilon}/\gamma}}-{\sqrt\frac{\gamma}{\delta}},
\end{align*}
where the third step uses \ref{enu:correlation-u-1} given in the lemma statement, and (\ref{eq:uv-nonUV-mass}) with Cauchy-Schwarz inequality; the fourth step uses (\ref{eq:supp-U-upper-secondary}).
Set $\kappa^{6}=\epsilon/\delta^{4},\gamma=\kappa^2\delta,$
then
\begin{align*}
\left\langle u,\frac{\one|_{U}}{\sqrt{|U|}}\right\rangle  & \ge1-O(\kappa).\qedhere
\end{align*}
\end{proof}

Equipped with the above lemma, we move on to prove Theorem~\ref{thm:sparsity-test}.
\begin{proof}[Proof of Theorem~\ref{thm:sparsity-test}]
The completeness part is a straightforward application of Chernoff bound.
So we focus on the soundness part. 
Let $R$ and $T$ be the representative set of $\Psi$ and $\Phi$, respectively. 
When $\Psi,\Phi$ are $\epsilon$-tilted states, then $\Psi_0,\Psi_1,\Phi_0, \Phi_1$ are $2\epsilon$-tilted states, and $\Psi_1\otimes \Phi_1$ is a $4\epsilon$-tilted state by \cref{prop:tensor-tilted-states}. By \cref{prop:tilted-state-approx}, we have for any $i\in R$, and $j\in T$,
\begin{align}
    &\Pr \left[ \left|\langle u_{i},\one_{[n]}/\sqrt{n}\rangle^{2}+ 1- 2\alpha \right|> 12\sqrt{\epsilon} \right]\le\exp(-\epsilon k), \label{eq:sparsity-u-1-approx}\\
    &\Pr \left[ \left|\langle v_{j},\one_{[n]}/\sqrt{n}\rangle^{2} + 1 -2\beta \right|> 12\sqrt{\epsilon} \right]\le\exp(-\epsilon k), \label{eq:sparsity-v-1-approx} \\
    &\Pr \left[ \left| \langle u_{i},v_j \rangle^{2}+1-2\lambda \right|> 16\sqrt{\epsilon} \right]\le\exp(-2\epsilon k).\label{eq:sparsity-u-v-approx}
\end{align}

Set $\delta=2\alpha-1$. Note that the test passes only if $|(2\alpha-1)+(2\beta-1)-1|\le 2\sqrt{\epsilon}.$
Together with (\ref{eq:sparsity-u-1-approx}) and (\ref{eq:sparsity-v-1-approx}), it implies that
\begin{align}
 & |\langle u_{i},\one_{[n]}/\sqrt{n}\rangle^{2}-\delta|\le 12\sqrt{\epsilon},\label{eq:sparsity-u}\\
 & |\langle v_{j},\one_{[n]}/\sqrt{n}\rangle^{2}-(1-\delta)|\le 14\sqrt{\epsilon}.\label{eq:sparsity-v}
\end{align}
Therefore, if either (\ref{eq:sparsity-u}) or (\ref{eq:sparsity-v}) fails, the protocol accepts with probability at most $\exp(-\epsilon k)$.

Moreover, the test passes only if $2\lambda - 1 \le 2\sqrt\epsilon$. Thus when the following does not hold the test fails with probability at least $1-\exp(-2\epsilon k)$.
\begin{align}
    \langle u_i, v_j \rangle ^2 \le 18\sqrt\epsilon.\label{eq:sparsity-uv}
\end{align}

Now suppose (\ref{eq:sparsity-u}), (\ref{eq:sparsity-v}) and (\ref{eq:sparsity-uv}) are true for some $i\in R$ and $j\in T$. By \cref{lem:sparsity}, we have: 
\begin{enumerate}
    \item For any $\gamma$, there is subset $S\subseteq [n]$ such that $|S|\le ((2\alpha-1) + 9\epsilon ^{1/4} / \gamma)n$, and $\|u_i|_S\|^2 \ge 1-\gamma$.
    \item There is subset $S\subseteq[n]$ such that
    \[
    \left||S|-(2\alpha-1)n\right|\le O(\epsilon^{1/12}(2\alpha-1)^{1/3})n,
    \]
    and 
    \[
    \langle u_{i},\one_{S}/\sqrt{|S|}\rangle\ge1-O\left(\frac{\epsilon^{1/12}}{(2\alpha-1)^{2/3}}\right).
    \]
\end{enumerate}
Since for any representative state $u\in \Psi$, $\TD(u, u_i) \le \sqrt \epsilon$, the above two items implies
\ref{enu:sparsity-supp-upper} and \ref{enu:sparsity-supp-estimate} in the theorem statements. Therefore,
if either \ref{enu:sparsity-supp-upper} or \ref{enu:sparsity-supp-estimate} in the theorem statements
does not hold, then one of (\ref{eq:sparsity-u}), (\ref{eq:sparsity-v}) and (\ref{eq:sparsity-uv}) is
not true, failing the sparsity test with probability at least $1-\exp(-\epsilon k)$.
\end{proof}

Suppose that we have a target sparsity $\gamma$, a constant number in $(0,1)$. We adapt the previous sparsity test slightly to test whether some given state is close to $\SSS_\gamma$.

\noindent\fbox{\begin{minipage}[t]{1\columnwidth - 2\fboxsep - 2\fboxrule}%
\uline{Sparsity test \rom{2} (with target sparsity $\gamma$ and precision $\epsilon$)}

\textbf{Input:} $\Psi = \{u_1, \ldots, u_{2k}\},\Phi=\{v_1, \ldots, v_{2k}\}$
\begin{enumerate}
    \item Sparsity test I on $(\Psi,\Phi)$ with precision $\epsilon$. 
\end{enumerate}
\emph{Accept} if the sparsity test \rom{1} accepts and its output satisfies: $2\alpha-1 \in [\gamma  -\sqrt\epsilon, \gamma+\sqrt\epsilon]$.
\end{minipage}}

\begin{theorem}[Sparsity test with target sparsity $\gamma$]
\label{thm:sparsity-test-target}
Let $\epsilon >0$ be such that $\epsilon<\gamma^{4/5}$. Suppose that $\Psi$ and $\Phi$ are $\epsilon$-tilted states. 
Then the sparsity test accepts with probability at most $\exp(-\epsilon k)$ if the following fails to hold:
\begin{equation}    \label{eq:sparsity-test-target-gamma}
    \TD(\Psi, \SSS_\gamma) \le O\left(\frac{\epsilon^{1/24}}{\gamma^{1/3}}\right).
\end{equation}
If $\Psi$ is the $0$-tilted states from $\SSS_\gamma$, then there is $\Phi$ such that the sparsity test accepts with probability $1-\exp(-\Theta(\epsilon k))$
\end{theorem}
\begin{proof}
To prove the first part, it suffices to show that assuming Theorem~\ref{thm:sparsity-test}~\ref{enu:sparsity-supp-estimate} holds then (\ref{eq:sparsity-test-target-gamma})  holds. Suppose $2\alpha -1= (1+\epsilon')\gamma$. Then we assume that $|\epsilon'|\le \sqrt\epsilon / \gamma$, since otherwise the sparsity test rejects immediately. Note that $\epsilon'$ is a very tiny number in absolute value. By Theorem~\ref{thm:sparsity-test}~\ref{enu:sparsity-supp-estimate}, there is constant $c,C$ such that
for any representative state $\ket{\psi}\in\Psi$, 
\begin{equation} \label{eq:dist-psi-sss-gamma}
    \TD(\ket{\psi}, \SSS_{\gamma'}) \le 
    \frac{C\epsilon^{1/24}}{((1 + \epsilon')\gamma)^{1/3} },
\end{equation}
where
\[
|\gamma' - (2\alpha-1)|
    \le c\epsilon^{1/12}(1+\epsilon')^{1/3} \gamma^{1/3}.
\]
Therefore,
\begin{align}
    |\gamma - \gamma' | 
    &\le |\gamma - (2\alpha-1)| + |\gamma'-(2\alpha-1)| \nonumber\\
    &\le \epsilon' \gamma + c\epsilon^{1/12}(1+\epsilon')^{1/3} \gamma^{1/3} \nonumber \\
    &\le c'\epsilon^{1/12} \gamma^{1/3}, \label{eq:gamma-diff}
\end{align}
where the last step holds due to that we set $\epsilon < \gamma ^{4/5}$, and $c'$ is some constant.
Note that for any $S\subseteq T\subseteq[n],$  we have,
\begin{align}\label{eq:dist-sss}
\TD\left(\frac{\one_{S}}{\sqrt{|S|}},\frac{\one_{T}}{\sqrt{|T|}}\right) & =\sqrt{1-\left(\frac{|S|}{\sqrt{|S||T|}}\right)^{2}}=\sqrt{\frac{|T|-|S|}{|T|}}.
\end{align}
By (\ref{eq:dist-psi-sss-gamma})-(\ref{eq:dist-sss}) and triangle inequality, for some absolute constant $C'$,
\begin{align}
    \TD(\ket\psi, \SSS_\gamma) 
        &\le \frac{C\epsilon^{1/24}}{(1 + \epsilon')^{1/3}\gamma^{1/3}} + \sqrt{|\gamma-\gamma'|/\gamma}, \nonumber \\
        & \le C'\epsilon^{1/24}\gamma^{-1/3}.
\end{align}

The second part of the theorem is simply the completeness case from \cref{thm:sparsity-test}.
\end{proof}

\subsection{Validity Test}

Consider the variable set $X=\{1,2,\ldots, n\}$, and domain $\Sigma = \{1,2,\ldots, q\}$. Recall that the valid set is the following
\[
    \cV =\left\{ \frac{1}{\sqrt n}\sum_{i\in [n]} \ket i \ket {x_i}:  \; \forall i\in[n], x_i \in \Sigma
    \right\}.
\]
The goal is to test whether a state is close to $\cV$.

\noindent\fbox{\begin{minipage}[t]{1\columnwidth - 2\fboxsep - 2\fboxrule}%
\uline{Validity test (with precision $d$)}

\textbf{Input: }$\Psi=\{|\psi_{1}\rangle,|\psi_{2}\rangle,\ldots,|\psi_{k}\rangle\}\subseteq\unitR.$
\begin{enumerate}
\item Apply discrete Fourier transform to the second register of $\Psi$. 
\item Measure the second register. 
\end{enumerate}
\emph{Accept} if $\alpha \le 1/q + d$, where $\alpha$ is the fraction of $|0\rangle$ observed after measuring.
\end{minipage}}

\begin{theorem}[Validity test]\label{thm:validity-test}
Suppose that $\Psi$ is an $\epsilon$-tilted state for some small
$\epsilon>0$. Further suppose that for any representative state $\ket{\psi}\in\Psi,$
$\TD(\ket{\psi},\SSS_{1/q})\le d$ for $2\epsilon \le d<1/q$. Then the probability that in the validity
test the fraction of measured $|0\rangle$ is less than
$(1+qd)/q$ is at most $\exp(-\Theta(qd^{2}k)),$ if 
\[
\TD(\ket{\psi},\cV)\ge \sqrt{6qd}+d.
\]
If $\Psi$ is a $0$-tilted state from $\cV$, then the validity test accepts with probability at least $1-\exp(-\Theta(q d^2 k))$.
\end{theorem}

\begin{proof}
Fix an arbitrary representative state $\ket{\psi},$ let $\ket{\phi}\in \SSS_{1/q}$
be such that
\[
\TD(\ket{\psi},\ket{\phi})=\TD(\ket{\psi},\SSS_{1/q})\le d.
\]
 If $\TD(\ket{\psi},\cV)\ge \sqrt{2qd}+d$, by triangle inequality
\begin{equation}
\TD(\ket\phi,\cV)\ge\TD(\ket{\psi},\cV)-\TD(\ket{\psi},\ket{\phi})\ge\sqrt{6qd}.\label{eq:dist-phi-valid-upper}
\end{equation}
Say $S\subseteq[n]\times[q]$ of size $n$ is such that
\[
    \ket{\phi}=\frac{1}{\sqrt{n}}\sum_{(i,v)\in S}|i\rangle|v\rangle.
\]
For each $i\in[n]$, let $c_{i}:=|\{(i,v)\in[n]\times[q]:(i,v)\in S\}|.$
Let $Z:=\{i:c_{i}=0\}$. 
Then, 
\[
\TD(\ket\phi,\cV)=\sqrt{1-\left(\frac{n-|Z|}{n}\right)^{2}}\;\Longrightarrow\;\frac{|Z|}{n}\ge\frac{1}{2}\TD(\ket\phi,\cV)^{2}.
\]
When measuring the second register of $\ket\phi$
after the discrete Fourier transform, the probability $\tilde{p}$
that we observe $|0\rangle$ can be calculated as below,
\begin{align}
\tilde{p} &=\frac{\sum_{i\in[n]}c_{i}^{2}}{nq} \ge \frac{n-|Z|}{nq} \left( \frac{n}{n-|Z|}  \right)^2 \nonumber \\
               &= \frac{1}{q}\cdot\frac{n}{n-|Z|} \ge \frac{1}{q}\left(1+\frac{|Z|}{n}\right) \nonumber \\
               &\ge \frac{1}{q}(1+\TD(\ket\phi, \cV)^2/2),
               \label{eq:est-prob-observing-0}
\end{align}
where the second step follows by convexity. 
By (\ref{eq:est-prob-observing-0}) and (\ref{eq:dist-phi-valid-upper}),
\[
\tilde{p}\ge \left(1+3qd\right)\frac{1}{q}.
\]
Now let $p$ be the probability that we observe $0$ measuring the second register
of $\ket{\psi}$ after applying Fourier transform, then by \cref{fact:trace_norm_acc},
\[
p\ge (1+3qd)\frac{1}{q} - d \ge (1+2qd)\frac{1}{q}.
\]
The first part of our lemma holds by Chernoff bound.

Now suppose that $\Psi$ is a $0$-tilted state from $\cV$. Let $\qpsi{}$ be the representative state of $\Psi$ and let $\mathop{\Pi} \qpsi{}$ denote the projection of $|\psi\rangle$ onto the subspace 
\[
\C^n\otimes \left(\frac{1}{\sqrt q}\sum_{v\in \Sigma} |v\rangle\right).
\]
Thus $\|\mathop{\Pi}|\psi\rangle\|^2$ is the probability of observing $|0\rangle$, after applying the Fourier transform to and measuring the second register of $|\psi\rangle$.
For any $|\psi\rangle \in \cV$, 
\[
\|\mathop{\Pi} |\psi\rangle\|^2 = \frac{1}{q}.
\]
It thus follows that in the validity test, we observe less than $1/q+d$ fraction of $|0\rangle$ with probability at least $1-\exp(-\Theta(q d^2 k))$.
\end{proof}

% SSE protocol
\section{\texorpdfstring{$\textup{SSE} \in \QMA^+_{\log}(2)$}
         {SEE in QMA+log(2)}}\label{sec:SSE-protocol}
The small-set expansion problem arises in the context of the unique games conjecture~\cite{RS10, BarakBHKSZ12}. The formal definition is given below.\footnote{This definition is slightly stronger than the original definition in~\cite{RS10}.}
\begin{definition}[$(\eta,\delta)$-SSE graph]
  Let $\eta,\delta \in (0,1)$.
  We say that $G$ is a $(\eta,\delta)$-small-set expander, or simply
  $(\eta,\delta)$-SSE for short, if for every $\emptyset \ne S \subseteq V$ of size $\abs{S}
  \le \delta \abs{V}$ we have $\Phi_G(S) \ge 1 - \eta$.
\end{definition}
\begin{definition}[$(\eta,\delta)$-SSE]
  Let $\eta, \delta \in (0,1)$. An instance of $(\eta,\delta)$-small-set expansion (SSE)
  problem is a graph $G$ on the vertex set $V$ such that
  \begin{description}
    \item[(Yes)] There exists $S \subseteq V$ with measure at most $\delta$ and $\Phi_G(S)\le \eta$;
    \item[(No)] Every  set $S \subseteq V$ of measure at most $\delta$ has expansion $\Phi_G(S)\ge 1-\eta$.
  \end{description}
\end{definition}

We now show that SSE can be verified with constant copies of unentangled proofs of
non-negative amplitudes and a logarithmic number of qubits with
\emph{constant} completeness-soundness gap. More precisely, we prove
the following theorem.

\begin{restatable}{theorem}{SSEMain}\label{theo:sse_qma2_plus}
  The $(\eta,\delta)$-SSE problem is in $\QMA_{O_\delta(\log(n))}^+(k,c,s)$ for some constant $k$, completeness $c \ge 1-\eta$
  and soundness $s \le 5/6 + O(\sqrt\eta\log(1/\eta))$.
\end{restatable}

We will prove the theorem by showing that the $\QMA_{\log}(k)$
protocol described in Algorithm~\ref{algo:sse_protocol} is complete and
sound for ($\eta,\delta$)-SSE. More precisely, the theorem follows immediately
from the following lemmas proven
in~\cref{sec:comp_anal,sec:sound_anal}, respectively.

\begin{restatable*}[Completeness]{lemma}{ProtComp}
  The protocol in Algorithm~\ref{algo:sse_protocol} accepts any yes instance with probability at least $1-\eta$.
\end{restatable*}

\begin{restatable*}[Soundness]{lemma}{ProtSound}\label{lemma:soundness}
  The protocol in Algorithm~\ref{algo:sse_protocol} accepts any no instance with probability at most $5/6+O(\sqrt\eta\log(1/\eta))$.
\end{restatable*}

\noindent\fbox{
\begin{myalg}[$(\eta, \delta)$-SSE Protocol]\label{algo:sse_protocol}\ignorespacesafterend
Let $\epsilon=\eta^{8}\delta^{4}/C$, and $k=C\log(1/\eta)/\epsilon^{2}$ for
some large enough constant $C$.

Let $S$ be the vertex set such that $|S|\le \delta n$ and $\Phi_G(S) \le \eta$.

\textbf{Provers:} Send
\begin{enumerate}
    \item $2k$ copies of the superpositions of the non-expanding set $S$, i.e.,
    \[
        |\psi_{1}\rangle,|\psi_{2}\rangle,\ldots,|\psi_{2k}\rangle=\frac{1}{\sqrt{\delta n}}\sum_{i\in S}|i\rangle.
    \]
    \item $2k$ copies of the superpositions of the complement of $S$, i.e.,
    \[
        |\phi_{1}\rangle,|\phi_{2}\rangle,\ldots,|\phi_{2k}\rangle=\frac{1}{\sqrt{(1-\delta)n}}\sum_{i\not\in S}|i\rangle.
    \]
\end{enumerate}

\textbf{Verifier: }Choose uniformly at random one of the following tests.
\begin{enumerate}
\item Symmetry test on $\{|\psi_{i}\rangle\}$ and symmetry test on $\{|\phi_{i}\rangle\}$.
\item Sparsity test I on ($\{|\psi_{i}\rangle\}, \{|\phi_{i}\rangle\}$)
with precision $\epsilon$. If the output $\alpha$ is such that $2\alpha-1> (1 + \eta)\delta$,
\emph{reject.}
\item Expansion test on $|\psi_{i}\rangle$ and $|\psi_j\rangle$ for two distinct random $i,j \in \{1,2,\ldots,2k\}$.
\end{enumerate}
\end{myalg}}

Since $G$ is a $d$ regular graph, its adjacency matrix $A$ can be
written as a sum of $d$ permutation matrices $P_1,\ldots,P_d$.
This representation as a sum of unitary matrices will be important to
view these matrices as valid quantum operations.
To test the lack of expansion of the support of $\ket{\psi_1}$, we
apply to this state a permutation $P_i$, chosen uniformly at random.
Then, we test if the resulting state (mostly) overlaps with
$\ket{\psi_2}$ (which is supposed to encode the same set in its
support).
This test is described in Algorithm~\ref{algo:test_expansion}.

\noindent\fbox{
\begin{myalg}[Expansion Test]\label{algo:test_expansion}\textbf{Input}: {$\ket{\psi_1},\ket{\psi_2} \in \unitR$}
   \begin{enumerate}
      \item Choose $r \in [d]$ uniformly at random;
      \item Compute $P_r \ket{\psi_1}$;
      \item SwapTest$\left(P_r \ket{\psi_1},\ket{\psi_2}\right)$.
   \end{enumerate}
\emph{Accept} if the swap test accepts.
\end{myalg}
}

\subsection{Completeness Analysis}\label{sec:comp_anal}

We now analyze the completeness of the protocol by proving the following lemma.

\ProtComp

\begin{proof}
  Suppose that $G$ is the input graph of a yes instance where $S$ is a
  non-expanding set of measure at most
  $\delta$. We expect $4k$ unentangled quantum proofs of the form
  \begin{align*}
    &\ket{\psi_j} = \frac{1}{\sqrt{\abs{S}}} \sum_{i \in S} \ket{i}, & \forall j\in\{1,2,\ldots,2k\},\\
    &\ket{\phi_j} = \frac{1}{\sqrt{n-\abs{S}}} \sum_{i\not\in S} \ket{i}, & \forall j\in \{1,2,\ldots,2k\}.
  \end{align*}

  The two symmetry tests accept with probability $1$ since they are running on sets of equal states. The sparsity test accepts with probability at least $1-\eta$ by~\cref{thm:sparsity-test}. 
  It only remains to analyze the expansion test.
  Recall that $A$ is the adjacency matrix of the graph instance, the assumption that $\Phi_G(S) \le \eta$ can be expressed as
  \begin{align*}
    \frac{1}{d} \braket{A\psi_1}{\psi_1} ~\ge~ 1-\eta\mper
  \end{align*}  
  Then, using Jensen's inequality we have
  \begin{align}
    \Ex{r \in [d]}{ \abs{\braket{P_r\psi_1}{\psi_1}}^2} &~\ge~ \left(\Ex{r \in [d]}{ \abs{\braket{P_r\psi_1}{\psi_1}}}\right)^2 \nonumber\\ 
                                                      &~=~ \braket{\Ex{r \in [d]}{P_r}\psi_1}{\psi_1}^2 \nonumber\\ 
                                                      &~=~  \braket{\frac{1}{d} A\psi_1}{\psi_1}^2 ~=~ (1-\eta)^2\mper \nonumber
  \end{align}
  In this case, the swap test on $\left(P_r \ket{\psi_1},\ket{\psi_2}\right)$ accepts with probability at least $1/2 + (1-\eta)^2/2 \ge 1 - \eta$.  
  Therefore, the entire protocol accepts with probability at least $1-\eta$ as claimed.    
\end{proof}

\subsection{Soundness Analysis}\label{sec:sound_anal}

We will establish the soundness of the protocol by showing the following lemma.

\ProtSound

First, we record a simple fact about expander graphs.

\begin{fact}\label{fact:expansion-eta-extra}
Suppose that the graph $G$ is $(\eta,\delta)$-SSE. Then $G$ is also $( (c+1)\eta, (1+c\eta)\delta)$-SSE, for any $c\ge0$.
\end{fact}

\begin{proof}
For any $\delta n<|S|\le(1+c\eta)\delta n$, let $T\subseteq S$ be
such that $|T|=\delta n$. Then
\[
|E(S,S)| \le |E(T,T)|+d|S\setminus T|\le \eta d|T| + d|S|\frac{|S\setminus T|} {|S|}\le (1+c\eta) d|S|.\qedhere
\]
\end{proof}

We will also need the following analytic version of the SSE property.

\begin{definition}[Analytic SSE]
  Let $\eta,\delta \in (0,1)$.
  We say that a graph $G=(V,E)$ with normalized adjacency matrix $A$ is $(\eta,\delta)$-analytic SSE if for every $v \in \R^V$
  of $\ell_2$-norm $1$ and support of measure at most $\delta$ it holds that
  \begin{align*}
    \abs{\ip{Av}{v}} \le \eta\mper
  \end{align*}
\end{definition}

This analytic property is implied by the SSE property as we show in
the following proposition (proved
in~\cref{subsubsec:analytic_sse_proof}).

\begin{restatable}{proposition}{Prop}\label{prop:sse_to_analytic_sse}
  If $G$ is $(\eta,\delta)$-SSE, then $G$ is $(O(\sqrt\eta(\log(1/\eta)+1)),\delta)$-analytic SSE.
\end{restatable}

Assuming~\cref{prop:sse_to_analytic_sse}, we now proceed to the proof
of~\cref{lemma:soundness}.

\begin{proof}[Proof of~\cref{lemma:soundness}]
Assume that $|\Psi\rangle= \{|\psi_{1}\rangle,\ldots,|\psi_{2k}\rangle\}$ 
and $|\Phi\rangle= \{|\phi_{1}\rangle,\ldots,|\phi_{2k}\rangle \}$
are $\epsilon$-tilted states.
We call this event $\cE_1$. If $\cE_1$ does not
hold, the symmetry test accepts with probability at most $\sqrt\eta$ for
$k=\Omega(\epsilon^{-2}\log(1/\eta)).$ 

Now we further assume that there is $S\subseteq[n],$
such that
\[
|S|\le(1+6\eta)\delta n,
\]
and let $\Pi_S$ be the projection into the subspace corresponding to $S$, then for any representative state $\ket{\psi_i}$,
\[
\|\Pi_S \ket{\psi_{i}}\|^2 \ge 1- 1.1\eta.
\]
This is the second event $\cE_2$. By our choice of parameters and the fact that if $2\alpha -1 >(1+\eta)\delta$
the sparsity test fails immediately,  we can assume that
\begin{align*}
    & (2\alpha - 1) + 9\epsilon^{1/4} / \eta \le (1+6\eta)\delta,   \\
    & 20\sqrt{\epsilon} \le \eta.
\end{align*}
Therefore, by~\cref{thm:sparsity-test}~\ref{enu:sparsity-supp-upper},
the sparsity test accepts with probability at most $\sqrt\eta$ by our choice of parameters if $\cE_2$ does not hold. 

Conditioning on $\cE_1$ and $\cE_2$, we analyze the probability that the expansion
test passes. Let's say the two proofs we get for the expansion test are $\ket{\psi_1}, \ket{\psi_2}$. With
probability at least $(1-2\epsilon)^2\ge 1-4\epsilon$, both are representative, thus satisfying
that their mass projected on
to the coordinates of $S$ is at least $ 1 -\eta.$ We call this event $\cE_3$. Let 
\[
|\pi_1\rangle = \frac{ \Pi_S |\psi_1  \rangle}{\| \Pi_S |\psi_1  \rangle \|},
\; 
|\pi_2\rangle = \frac{ \Pi_S |\psi_2  \rangle}{\| \Pi_S|\psi_2  \rangle \|}.
\]
It follows that
\begin{equation}\label{eq:psi-sparse-mass}
\langle \pi_1 \mid \psi_1 \rangle ^2 = \| \Pi|\psi_1 \rangle\| \ge 1-1.1\eta.
\end{equation}

Let $\delta_0 = (1+6\eta)\delta$. By~\cref{prop:sse_to_analytic_sse}, the \emph{analytic}
$(O(\sqrt\eta(\log(1/\eta)+1)),\delta_0)$-SSE property follows from the
$(\eta,\delta)$-SSE assumption and~\cref{fact:expansion-eta-extra}.
To determine the expected acceptance probability of the swap test, we
first bound the average value of $\abs{\braket{P_r\pi_1}{\pi_1}}$
over the random choice of $r$ obtaining
\begin{align}
  \Ex{r \in [d]}{ \abs{\braket{P_r\pi_1}{\pi_1}}}  
                                                     &= \braket{\Ex{r \in [d]}{P_r}\pi_1}{\pi_1}\nonumber \\
                                                     &= \frac{1}{d} \braket{A\pi_1}{\pi_1} \nonumber \\
                                                     &\le O(\sqrt\eta(\log(1/\eta)+1))\nonumber \mcom
\end{align}
where the first step 
holds because the entries of $\ket{\psi_1}$, $\ket{\psi_2}$ and $P_r$, for every
$r$, are non-negative real numbers. Now, it follows that
\begin{align*}
    \Ex{r\in[d]} {\abs{\braket{P_r \psi_1}{\psi_2}}^2} 
            &\le \Ex{r\in[d]}{|\braket{P_r \psi_1}{\psi_1}|^2} + 3\sqrt{\epsilon} \\
            &\le \Ex{r\in[d]}{|\braket{P_r \pi_1}{\pi_1}|^2} + 3\sqrt{\epsilon} + \sqrt{1.1\eta} \\
            &\le \Ex{r\in[d]}{\braket{P_r \pi_1}{\pi_1}} + 3\sqrt\epsilon+\sqrt{1.1\eta}\\
            &= O(\sqrt{\eta}\log(1/\eta)),
\end{align*}
where the first  step follows \cref{claim:correlation-dist}~\ref{enu:correlation-closeness}; the second step is due to \cref{fact:trace_norm_acc} and the bound $\TD(\ket{\pi_1}, \ket{\psi_1})\le \sqrt{1.1\eta}$ that follows (\ref{eq:psi-sparse-mass}).
Hence, the swap test on $\left(P_r \ket{\psi_1},\ket{\psi_2}\right)$ accepts
with probability at most $1/2 + O(\sqrt\eta\log(1/\eta))$.

To conclude, if $\cE_1$ (or $\cE_2$) does not hold with probability at least $1/3 \times (1-\sqrt\eta)$, the protocol chooses the symmetry test (or sparsity test) and rejects. If both $\cE_1$ and $\cE_2$ hold, then
the protocol chooses the expansion test with probability $1/3$ and rejects with probability at least
$(1/2 - O(\sqrt\eta\log(1/\eta)))$ conditioning on $\cE_3$, which happens with probability at least $1-4\epsilon$. Hence the protocol accepts with probability at most $5/6+ O(\sqrt\eta\log(1/\eta)).$
\end{proof}

\subsection{The Analytic SSE Property}\label{subsubsec:analytic_sse_proof}

In this section, we will establish the \emph{analytic} SSE property
from the usual SSE property.

\Prop*

In a seminal work on $2$-lifts of graphs~\cite{BL06}, Bilu and Linial
found conditions under which bounding the quadratic form $\ip{A u}{u}$
of a matrix $A$ for \emph{arbitrary} vector $u$ follows from bounds on
much simpler ``flat'' indicator vectors $\ip{A \one_S}{\one_T}$. Our
goal is to use a version of their result adapted for vectors of small
support as arising in our application. More precisely, we will need an
inequality of the form $\abs{\ip{A \one_S}{\one_T}} \le \eta
(\abs{S} + \abs{T})$ for every disjoint $S,T \subseteq V$ of support
at most $\delta$. We first show that this inequality is indeed
satisfied by the adjacency matrix of SSE graph.

\begin{lemma}\label{claim:sse_between_small_sets}
  Suppose $G=(V,E)$ is a $d$-regular $(\eta,\delta)$-SSE with adjacency matrix $A$ (not normalized).
  If $S,T \subseteq V$ are disjoint sets with $\abs{S} +\abs{T} \le \delta \abs{V}$, then
  \begin{align*}
    \ip{A \one_S}{\one_T} \le 2\sqrt\eta d \sqrt{\abs{S}\abs{T}}.
  \end{align*}
\end{lemma}

\begin{proof}
  Let $S, T$ be as in the assumption of the claim. Without loss of generality, assume that $\abs{S} \le \abs{T}$. If $\abs{S} \le \eta \abs{T}$, then we can use the trivial bound by the fact that $A$ is the adjacency matrix of a $d$-regular graph,
  \[
    \ip{A \one_S}{\one_T} \le d|S| \le \sqrt\eta\sqrt{\abs{S}\abs{T}}.
  \]
  
  Now consider the case $\eta\abs{T} < |S| \le |T|$. Set $S' = S \sqcup T$. Towards a contradiction, suppose that $\ip{A \one_S}{\one_T} > 2 \sqrt{\eta} d \sqrt{\abs{S}\abs{T}}$. In turn, this assumption implies that
  \begin{align}\label{eq:claim_ineq_disj_sets}
     \ip{A \one_S}{\one_T} > 2\sqrt{\eta} d \sqrt{\abs{S}\abs{T}} \ge 2 \eta d \abs{T} \ge \eta d \left(\abs{S} + \abs{T}\right) = \eta d \abs{S'}\mper
   \end{align}
    Using the above  bound on the number of edges between $S$ and $T$ together with the SSE assumption on $G$, we obtain
   \begin{align*}
     (1- \eta) d \abs{S'} &\le \ip{A \one_{S'}}{\one_{\overline{S'}}} \\
                          &\le d\abs{S'} -  \ip{A \one_S}{\one_T}\\
                          & < d\abs{S'} -  \eta d \abs{S'} && \text{(By~(\ref{eq:claim_ineq_disj_sets}))}\\
                          &\le (1- \eta) d \abs{S'}\mcom
   \end{align*}
   contradicting the $(\eta,\delta)$-SSE property.
\end{proof}

We now show a ``sparse support'' analogue of a lemma in Bilu and
Linial~\cite[Lemma~3.3]{BL06} bounding the quadratic form of the
adjacency matrix for arbitrary sparse vectors assuming that the
quadratic form is bounded for ``flat'' sparse indicator vectors.
This sparse analogue follows by checking that their proof suitably
``respects'' the sparse support conditions we need.

\begin{lemma}[Sparse Analogue of~{\cite[Lemma~3.3]{BL06}}]\label{lemma:quadratic_form_bound}
  Let $A \in \R^{V\times V}$ be a real symmetric matrix with non-negative entries, $\ell_1$-norm of each row at most $d$ and diagonal entries zero.
  Let $\delta \in (0,1)$.
  If there exists $\alpha \in (0,1)$ such that for every disjoint sets $S,T \subseteq V$ with $\abs{S\sqcup T} \le \delta \abs{V}$
  we have
  \begin{align}\label{eq:bl_bound_assumption}
    \ip{A \one_{S}}{\one_{T}} \le \alpha d \sqrt{\abs{S} \abs{T}},
  \end{align}
  then for every $u \in \R^V$  with $\abs{\supp(u)} \le \delta \abs{V}$
  we have
  \begin{align}\label{eq:bl_bound_conclusion}
    \ip{A u}{u} \le O(\alpha (\log(1/\alpha)+1)) d \norm{u}^2.
  \end{align}  
\end{lemma}

\begin{proof}
  The assumption of~(\ref{eq:bl_bound_assumption}) on disjoint sets $S$
  and $T$ is strong enough to imply a similar bound with an additional
  factor of $2$ when $S=T$ as follows.
  \begin{claim}\label{claim:bound_on_equal}
    Suppose that $A$ is a symmetric matrix with diagonal entries equal to zero.
    The assumption from~(\ref{eq:bl_bound_assumption}) implies that for every $R \subseteq V$ with $\abs{R} \le \delta \abs{V}$
    \begin{align*}
      \ip{A \one_{R}}{\one_{R}} \le 2 \alpha d \abs{R}\mper
    \end{align*}    
  \end{claim}
  
  \begin{proof}
    Let $r = \abs{R}$. 
    If $r = 1$, we have $\ip{A \one_{R}}{\one_{R}} = 0$ since $A$ has diagonal entries equal to zero.
    Now assume $r \ge 2$. On one hand, we have
    \begin{align*}
      \sum_{\substack{R' \subseteq R \\ \abs{R'} = \lceil r/2\rceil}} \abs{\ip{A \one_{R'}}{\one_{R\setminus R'}}}  
      &\le  \binom{r}{\lceil r/2\rceil}  \alpha d \sqrt{\abs{R'}  \abs{R\setminus R'}} \le \binom{r}{\lceil r/2\rceil}  \alpha d \abs{R} /2.
    \end{align*}
    On the other hand, for distinct $x,y \in R$, the value $A_{x,y}$ appears $\binom{r-2}{\lceil r/2\rceil-1}$ in the LHS above.
    Since $A$ has diagonal entries equal to zero, this gives
    \begin{align*}
      \binom{r-2}{\lceil r/2-1\rceil} \abs{\ip{A \one_{R}}{\one_{R}}} = \sum_{\substack{R' \subseteq R \\ \abs{R'} = \lceil r/2\rceil}} \abs{\ip{A \one_{R'}}{\one_{R\setminus R'}}}.
    \end{align*}
    From the two previous displays and the bound on the following binomial ratio
    \begin{align*}
      \binom{r}{\lceil r/2\rceil}/\binom{r-2}{\lceil r/2\rceil-1} = \frac{r(r-1)}{\lceil r/2 \rceil \lfloor r/2 \rfloor} ~\le~ 4,
    \end{align*}
    we conclude the proof.
  \end{proof}

  Arbitrary vectors can be approximated to have entries that are powers of two
  with the following nice properties.
  \begin{claim}\label{claim:diadic_approx}
    Suppose $A \in \R^{V\times V}$ has diagonal entries equal to zero.
    Let $u \in \R^V$ with $\norm{u}_\infty \le 1/2$.
    Then, there exists $u' \in \set{\pm 1/2^i \mid i \in \mathbb{N}^+}^V$ such
    that
    \begin{enumerate}
      \item $\abs{\ip{A u}{u}} \le \abs{\ip{A u'}{u'}}$,
      \item $\norm{u'} \le 2 \norm{u}$,
      \item $\supp(u') \subseteq \supp(u)$.
    \end{enumerate}
  \end{claim}

  \begin{proof}
    For every $i \in V$,  we define $\eta_i \in [0,1/4]$ such that $u_i = (1- 2\eta_i) \sgn(u_i) 2^{\lceil \log(\abs{u_i}) \rceil}$.
    Note that $(1- 2\eta_i) \in [1/2,1]$.
    We define a random vector $\rv Z' \in \R^V$ by setting $\rv Z'_i = 0$ if $i \not\in \supp(u)$,
    and otherwise by setting
    \begin{align*}
      \rv Z'_i = \begin{cases}
        +\sgn(u_i) 2^{\lceil \log(\abs{u_i}) \rceil} & \text{w.p. } 1-\eta_i,\\
        -\sgn(u_i) 2^{\lceil \log(\abs{u_i}) \rceil} & \text{w.p. } \eta_i.
      \end{cases}
    \end{align*}
    Note that by construction, we have $\Ex{\rv Z'_i} = u_i$.
    Using linearity of expectation and the assumption that $A$ has diagonal  entries equal to zero, we have
    \begin{align*}
      \ip{A u}{u} = \sum_{i\ne j} A_{i,j} u_i u_j = \sum_{i\ne j} A_{i,j} \E[\rv Z_i']\E[ \rv Z_j'] = \sum_{i\ne j} A_{i,j} \E[\rv Z_i' \rv Z_j'] = \Ex{\ip{A \rv Z'}{\rv Z'}}\mper
    \end{align*}
    This implies that there is a choice of $u'$ satisfying
    \begin{align*}
      \abs{\ip{A u}{u}} = \abs{\Ex{\ip{A \rv Z'}{\rv Z'}}} \le \Ex{\abs{\ip{A \rv Z'}{\rv Z'}}} \le \abs{\ip{A u'}{u'}}\mper
    \end{align*}
    To conclude note that a term-by-term inequality gives
    \begin{align*}
      \norm{u'}_2^2 = \sum_i (u_i')^2 \le 4 \sum_i u_i^2 = 4 \norm{u}_2^2\mcom
    \end{align*}
    concluding the proof.
  \end{proof}

  Let $u \in \R^V$ be an arbitrary vector with $\abs{\supp(u)} \le \delta \abs{V}$.
  We want to give an upper bound on $\abs{\ip{A u}{u}}$ as in~(\ref{eq:bl_bound_conclusion}).
  To prove this bound, we can assume $\norm{u}_\infty \le 1/2$ without loss of generality.
  Using~\cref{claim:diadic_approx}, we obtain $u' \in \set{\pm 1/2^i \mid i \in \mathbb{N}^+}^V$.
  Let $S_i \coloneqq \{j \in V \mid |u'_j| = 2^{-i} \}$.
  Set $t = \log(1/\alpha)$.
  Since the entries of $A$ are non-negative, we have
  \begin{align*}
    \abs{\ip{A u'}{u'}} 
                        &\le \sum_{x,y \in V} A_{x,y} \abs{u'_x} \abs{u'_y}\\
                        &= \sum_{i,j \in \mathbb{N}^+} \frac{1}{2^{i+j}} \ip{A \one_{S_i}}{\one_{S_j}} \\  
                        &= \underbrace{\sum_{i} \frac{1}{2^{2i}} \ip{A \one_{S_i}}{\one_{S_i}}}_{(a)} + \underbrace{\sum_{i} \sum_{i < j \le i + t} \frac{1}{2^{i+j}} \ip{A \one_{S_i}}{\one_{S_j}}}_{(b)} \\
                        &{\hspace{6cm}+
                        \underbrace{\sum_{i} \sum_{j  > i + t} \frac{1}{2^{i+j}} \ip{A \one_{S_i}}{\one_{S_j}}}_{(c)}\mper}
  \end{align*}
Using the assumption in~(\ref{eq:bl_bound_assumption}), by~\cref{claim:bound_on_equal} term $(a)$ becomes
  \begin{align*}
    \sum_{i} \frac{1}{2^{2i}} \ip{A \one_{S_i}}{\one_{S_i}} \le 4\alpha d \sum_{i} \frac{1}{2^{2i}}\abs{S_i} = 4\alpha d \norm{u'}_2^2\mper
  \end{align*}  
Note that $S_i \cap S_j = \emptyset$ when $i \ne j$.
Using the assumption in~(\ref{eq:bl_bound_assumption}), term $(b)$ becomes
  \begin{align*}
    \sum_{i} \sum_{i < j \le i + t} \frac{1}{2^{i+j}} \ip{A \one_{S_i}}{\one_{S_j}} 
            &\le \sum_{i} \sum_{i < j \le i + t} \frac{1}{2^{i+j}} \alpha d \sqrt{\abs{S_i} \abs{S_j}}\\
            &\le \sum_{i}  \sum_{i < j \le i + t}  \alpha d \left(\frac{1}{2^{2i}}\abs{S_i} +\frac{1}{2^{2j}}\abs{S_j}\right)\\
            &\le 2\alpha \log(1/\alpha) d \sum_{i} \frac{1}{2^{2i}} \abs{S_i}\\
            &= 2\alpha \log(1/\alpha) d \norm{u'}_2^2,
  \end{align*}
  where the second step applies the Cauchy-Schwartz inequality.

  Note that the $\ell_1$ bound of $d$ on the row and column sums of $A$ trivially implies that $\ip{A \one_{S_i}}{\one_{S_j}} \le d \abs{S_i}$.  
  By the choice of $t$ and this trivial bound, term $(c)$ becomes
  \begin{align*}
    \sum_{i} \sum_{j  > i + t} \frac{1}{2^{i+j}} \ip{A \one_{S_i}}{\one_{S_j}} 
    &\le \sum_{i} \sum_{j  > i + t} \frac{1}{2^{i+j}} d \abs{S_i}\\
    &\le \alpha d \sum_{i} \sum_{j  > i} \frac{1}{2^{i+j}} \abs{S_i} \\
    &\le 2 \alpha d \sum_{i}  \frac{1}{2^{2i}} \abs{S_i} = 2 \alpha d \norm{u'}_2^2\mper
  \end{align*}
Putting the bounds on $(a)$, $(b)$ and $(c)$ together, we obtain
  \begin{align*}
    \abs{\ip{A u}{u}} \le \abs{\ip{A u'}{u'}} \le 6 \alpha (\log(1/\alpha)+1) d \norm{u'}_2^2 \le 12 \alpha (\log(1/\alpha)+1) d \norm{u}_2^2\mcom
  \end{align*}
  concluding the proof.
\end{proof}

As a consequence of the above lemma and
\cref{claim:sse_between_small_sets}, we obtain our main result of this
section, namely, that SSE graphs are analytic SSE as follows.

\Prop*

\begin{proof}[Proof of~\cref{prop:sse_to_analytic_sse}]
  Since $G$ is a $(\eta,\delta)$-SSE, using~\cref{claim:sse_between_small_sets} we have
  for every disjoint sets $S,T \subseteq V(G)$ with $\abs{S \sqcup T} \le \delta \abs{V}$
  \begin{align*}
    \ip{A \one_{S}}{\one_{T}} \le \alpha d \sqrt{\abs{S} \abs{T} }\mcom
  \end{align*}
  where $\alpha = 2\sqrt\eta$. By~\cref{lemma:quadratic_form_bound}, this implies that $G$ is $(O(\sqrt\eta\log(1/\eta)),\delta)$-analytic SSE
  concluding the proof.
\end{proof}

% UG protocol

\section{\texorpdfstring{$\textup{GapUG} \in \QMA^+_{\log}(2)$ and $\NP \subseteq \QMA^+_{\log}(2)$}{GapUC in QMA+log(2) and NP subset of QMA+log(2)} }\label{sec:ug-protocol}

\begin{definition}[Unique Games]
A unique game instance $\fI$ consists of a $d$-regular graph
$G=(V,E)$. Each edge $e=(a,b)\in E$ is associated with a bijective
constraint $f_{e}:\Sigma\to\Sigma$, where $\Sigma=\{1,2,\ldots,q\}$
for some constant $q$. 

For any labeling $\ell: [n]\to \Sigma$, the value associated
with the labeling is the fraction of edge constraints satisfied by the
labeling, i.e.,
\[
\frac{1}{nd}|\{ (a,b)\in E: f_{(a,b)}(\ell(a)) = \ell(b) \}|.\footnotemark
\]
The value of $\fI$, denoted $\val(\fI)$, is the max value over all
possible labelings.
\end{definition}
\footnotetext{Though we can think of the graph in the definition being undirected, when we describe an edge constraint for $e=(a,b)$ using a bijection, we need labels of one vertex as the domain and labels of the other as the range of $f$. So when we say $f_e$, we always have an implicit orientation of the edge. So the set here counts each edge twice, that is $\val$ can take the value up to 1.}

\begin{definition}[($1-\delta, \eta$)-GapUG problem]
Given any unique games instance $\fI$. Determine which of the following two cases is true:
\begin{description}
\item [(Yes)] $\val(\fI)\ge 1-\delta$.
\item [(No)] $\val(\fI) \le \eta$.
\end{description}
\end{definition}

The purpose of this section is to establish the following theorem.
\begin{theorem}\label{thm:ug-in-qma2}
For any $\delta, \eta \in(0,1)$ such that $ (1-\delta)^2 > \eta$, then 
\begin{align*}
(1-\delta, \eta)\textup{-{GapUG}} \in \QMA^+_{\log}(2).
\end{align*}
\end{theorem}

It suffices to present a $\QMA^+_{\log}(k)$ protocol (see Algorithm~\ref{algo:ug-protocol}) for some constant $k$ for the
$(1-\delta, \eta)$-GapUG problem.
For the given graph $G=(V, E)$, say $V=\{1,2,\ldots,n\}.$ Since $G$ is a regular graph,
we can partition $E$ into $d$ permutations $\pi_{1},\pi_{2},\ldots,\pi_{d}:\{n\}\to\{n\}.$ The permutation can also
be thought of as a perfect matching between two vertex sets $L$ and $R$ with $L=R=V$.
We find the matching view more convenient, so we often call $\pi$ a matching.
For any labeling $\ell:[n]\to\Sigma$, we represent it by the following quantum state
\[
    |\psi\rangle = \frac{1}{\sqrt n} \sum_{i\in [n]} |i\rangle|\ell(i)\rangle.
\]
Recall that $\cV\subseteq \SSS_{1/q}$ denote the set of all valid labelings, i.e., 
\[
\cV:=\left\{ \frac{1}{\sqrt{n}}\sum_{i=1}^{n}|i\rangle|v_{i}\rangle:v_{i}\in \Sigma \right\} .
\]
Let $\Pi_{r}$ be the unitary map associated with the matching $\pi_r$, such that for any $r\in[d],i\in[n],$
and $v\in\Sigma:$
\[
\mathop{\Pi_{r}}|i\rangle|v\rangle\mapsto|\pi_{r}(i)\rangle|f_{(i,\pi_{r}(i))}(v)\rangle.
\]
In words, when we pick a matching $\pi_r$ and a labeling $\ket{\psi}$ on $L$, then $\Pi_r \ket{\psi}$ represents the unique labeling on $R$ that satisfies all the edge constraints for the edges in $\pi_r$. In reality, $L$ and $R$ are the same vertex set, they have the same labeling.
Let 
\begin{align*}
    &\theta = \frac{1}{2}\left(
        \frac{1+(1-\delta)^2}{2} + \frac{1+\eta}{2}
    \right),\\
    &\lambda = 
        \frac{(1-\delta)^2}{2} - \frac{\eta}{2}.
\end{align*}
We prove \cref{thm:ug-in-qma2} by establishing the following two lemmas in the next subsection.
\begin{lemma}[Completeness of UG protocol]\label{lem:ug-complete}
For any unique games instance $\fI$, if $\val(\fI)\ge 1-\delta$. Then there is a proof with $k=O_{\delta,\eta}(1)$ unentangled states, each of size $O_{\delta,\eta}(\log n)$, such that Algorithm~\ref{algo:ug-protocol} accepts with probability at least $0.99$.
\end{lemma}
\begin{lemma}[Soundness of UG protocol]\label{lem:ug-sound}
For any unique games instance $\fI$, if $\val(\fI)\le \eta$. Then for any proof with $k=O_{\delta,\eta}(1)$ unentangled states, each of size $O_{\delta,\eta}(\log n)$, such that Algorithm~\ref{algo:ug-protocol} accepts with probability at most $7/8$.
\end{lemma}

\noindent\fbox{\noindent
\begin{myalg}[$(1-\delta, \eta)$-GapUG Protocol]%
\label{algo:ug-protocol}\ignorespacesafterend
Let $\epsilon = \lambda^{48}  / (Cq^{32})$, 
and $k = C / \epsilon^2$ for some large enough constant $C$.

\textbf{Provers:} send 
\begin{enumerate}
\item $2k$ copies of labelings that realize $\val(\fI)$, i.e.,
\begin{align*}
 & |\psi_{1}\rangle,|\psi_{2}\rangle,\ldots,|\psi_{2k}\rangle=\frac{1}{\sqrt{n}}\sum_{i\in[n]}|i\rangle|\ell(i)\rangle.
\end{align*}
\item $2k$ copies of the labelings but complemented, i.e.,
\begin{align*}
 & |\gamma_{1}\rangle,|\gamma_{2}\rangle,\ldots,|\gamma_{2k}\rangle=\frac{1}{\sqrt{n}}\sum_{i\in[n]}|i\rangle\frac{1}{\sqrt{q-1}}\sum_{v\not=\ell(i)}|v\rangle.
\end{align*}
\end{enumerate}

\textbf{Verifier:} Let $\Psi=\{|\psi_{1}\rangle,\ldots,|\psi_{2k}\rangle\}$,
and similarly for $\Gamma$. Run a uniformly
random test of the following
\begin{enumerate}
    \item Two symmetry tests on $\Psi$ and $\Gamma$.
    \item Sparsity test on $(\Psi,\Gamma)$ with target sparsity $1/q$ and precision $\epsilon$.
    \item Validity test on $\Psi$ with precision $\nu=\epsilon^{1/24}q^{1/3}$.
    \item Labeling test on $\Psi_0$, $\Psi_1$, where $\Psi_0$ and $\Psi_1$ are partition of $\Psi$ into two subsets with equal size.
\end{enumerate}
\end{myalg}
}

The labeling test is described below.

\noindent\fbox{\begin{minipage}[t]{1\columnwidth - 2\fboxsep - 2\fboxrule}%
\uline{Labeling Test}

\textbf{Input: }$\Psi=\{|\psi_{1}\rangle,|\psi_{2}\rangle,\ldots,|\psi_{k}\rangle\},\Phi =\{|\phi_{1}\rangle,|\phi_{2}\rangle,\ldots,|\phi_{k}\rangle\}.$
\begin{enumerate}

\item For $i$ from $1$ to $k$, SwapTest on $(\Pi_{r}|\psi_{i}\rangle,|\phi_{i}\rangle)$
for uniformly random $r\in[d]$ (each iteration with a fresh random
choice).

\iffalse
\item Take a uniformly random $r\in[d]$.
\item SwapTest on $\Pi_r \psi_i$ and $\phi_j$ for uniformly random $i,j\in [k].$
\fi
\end{enumerate}
\emph{Accept} if more than a $\theta$ fraction the SwapTests accept. 
\end{minipage}}

\subsection{Analysis}

We first prove Lemma~\ref{lem:ug-complete}, the completeness. In particular, we show that for whichever test the protocol chooses, it accepts with probability at least $0.99$ when $\val(\fI)\ge 1-\delta$.

For faithful proofs, the symmetry test passes with probability 1, and the sparsity test accepts with probability at least, by Theorem~\ref{thm:sparsity-test-target}, $1-\exp(-\Theta(\epsilon k))$. The validity test accepts with probability at least $1-\exp(-\Theta(q \nu^2 k))$ by Theorem~\ref{thm:validity-test}. The way we choose our parameters guarantees that the accept probability is at least $0.99$.

Finally, when the UG instance has a value of at least $1-\delta,$ then there is
valid labeling $|\psi\rangle \in \cV,$ such that
\begin{align*}
\Exp_{r\in[d]}\braket{\psi }{ \Pi_{r}\psi }  & \ge1-\delta.
\end{align*}
Analogous to our analysis in \cref{sec:SSE-protocol}, we have
\begin{align*}
\Exp_{r\in[d]}[\braket{\psi}{\Pi_{r}\psi}^{2}] & \ge\left(\Exp_{r\in[d]}[\braket{\psi}{\Pi_{r}\psi}]\right)^{2}
  \ge(1-\delta)^{2}.
\end{align*}
Therefore, each swap test in the labeling test accepts with probability at least $1/2+(1-\delta)^2/2\ge 1-\delta$. By Chernoff bound, with probability at least $1-\exp ( -\Theta(\lambda^2 k) ) \ge 0.99$ for our choice of parameters, the labeling test accepts.

Now, we have proved the completeness. Next, we prove Lemma~\ref{lem:ug-sound}, the soundness, for which the following analysis on the labeling test will complete the last missing piece.

\begin{lemma}[Labeling test]\label{lem:labeling-test}
Suppose $\val(\fI)\le\eta.$  Given $\epsilon$-tilted states $\Psi$ such that any
representative state $\ket{\psi}$ satisfies $\TD(\ket{\psi},\cV)$ and $\epsilon$ being sufficiently small (for example, $\TD(\ket{\psi},\cV) \le \lambda / 8$ and $\epsilon \le \lambda^2 / 256$). 
Then the labeling test accepts $\Psi$ with probability at most $\exp(-\Theta(\lambda^2k))$.
\end{lemma}

\begin{proof}
For any valid labelings $\ket{\tilde{\psi}}\in\cV$,
\begin{align*}
\val(\fI) & \ge\Exp_{r\in[d]}\langle\tilde{\psi},\Pi_{r}\tilde{\psi}\rangle
 \ge\Exp_{r\in[d]} [\langle\tilde{\psi},\Pi_{r}\tilde{\psi}\rangle^{2}].
\end{align*}
Therefore the probability that  SwapTest accepts $\ket{\tilde\psi}$ is at most $1/2+\eta/2$.  Let $\ket\psi, \ket\phi$ be two representative states from $\Psi$. Suppose that for some $\ket{\tilde\psi}\in\cV$,
$\TD(\ket\psi, \ket {\tilde\psi})\le D.$
By Fact~\ref{fact:trace-tensor},
\[
\TD(\ket\psi\otimes\ket\phi, \ket {\tilde\psi}\otimes \ket {\tilde\psi}) \le \sqrt{D^2 + (D+\sqrt{\epsilon})^2}\le 2(D+\sqrt\epsilon).
\]
It then follows by Fact~\ref{fact:trace_norm_acc} that
the labeling test accepts $\ket{\psi_1}\otimes \ket{\psi_2} $ for two representative states in $\Psi$ with probability at most $1/2 + \eta/2 + 2(D+\sqrt\epsilon)$. When we partition $\Psi$ into two subsets $\Psi_1$ and $\Psi_2$, then with probability at least $1-2\epsilon$ the states we pick from $\Psi_1$ and $\Psi_2$ are both representative states of $\Psi$. 
By Chernoff bound, with probability at most $\exp(-\Theta((\lambda-2D-2\sqrt\epsilon-3\epsilon)^2k))=\exp(-\Theta(\lambda^2 k))$, the SwapTests accept more than $\theta-3\epsilon$ fraction within the $1-2\epsilon$ good pairs. Since $2\epsilon \le 3\epsilon(1-2\epsilon)$ for sufficiently small $\epsilon$, in total, the swap tests accept more than $\theta$ fraction of the pairs with probability at most $\exp(-\Theta(\lambda^2 k))$.
\end{proof}

With all the above preparations, we are now ready prove the soundness lemma.
\begin{proof}[Proof of \cref{lem:ug-sound}]

Consider the following events. 
\begin{description}
    \item [{$\cE_{1}:$}] $\Psi$ and $\Gamma$ are $\epsilon$-tilted
    states;
    \item [{$\cE_{2}:$}] $\TD(\Psi,\SSS_{1/q}) \le O(\epsilon^{1/24}q^{1/3})$;
    \item [$\cE_3:$]  $\TD(\Psi,\cV) \le O(\epsilon^{1/48}q^{2/3})$.
\end{description}

If $\cE_1$ is not true, then the symmetry test accepts with probability at most $\exp(-\Theta(\epsilon^2 k )) < 0.01$ by~\cref{thm:symmetry-test} for $k=\Omega(1/\epsilon^2)$. Thus the probability that the protocol accepts is at most $3/4+0.01 < 7/8$.

Conditioning on $\cE_1$, if $\cE_2$ does not hold, then the sparsity test accepts with probability at most $\exp(-\Theta(\epsilon k)) < 0.01$ for $k=\Omega(1/\epsilon)$ by \cref{thm:sparsity-test-target}. In total, the protocol accepts with a probability less than 7/8.

Conditioning on $\cE_1$ and $\cE_2$, by \cref{thm:validity-test}, if $\cE_3$ does not hold, then the validity test accepts with probability at most $\exp(-\Theta(q^{5/3}\epsilon^{1/12}k))<0.01$. Therefore, the protocol accepts with probability less than $7/8$ again.

Finally, conditioning on $\cE_1$ and $\cE_3$, by Lemma~\ref{lem:labeling-test}, the labeling test accepts with probability at most $\exp(-\Theta(\lambda^2 k))$ if $\epsilon^{1/48}q^{2/3} = O( \lambda) $ and $\epsilon = O(\lambda^2)$. By our choice of parameters, the protocol accepts with probability at most $7/8$.
\end{proof}

\subsection{Regularization---\texorpdfstring{$\NP\subseteq \QMA^+_{\log}(2)$}{NP subseteq QMA+log(2)}}\label{sec:regularization}
Due to the works~\cite{KMS17,KMS18,DKKMS18towards2to1,DKKMS18}, it is known that the $(1/2, \eta)$-GapUG problem is $\NP$-hard. An optimistic reader would happily conclude that $\NP\subseteq \QMA_{\log}^+(2)$. This is indeed the case, with a small caveat though: In our previous discussion, we assumed the graph instance to be regular. However, when we convert a general graph into a regular one, the value of the game will change. We address this issue here. 

\begin{theorem}[Regularization~\cite{Dinur07pcp}]\label{thm:regularization}
For any general unique games instance $\fI$, there is a new unique games instance $\fI'$ that is polynomial time constructible such that
\begin{align}
     \val(\fI) \ge \frac{1}{2} &\Longrightarrow \val(\fI') \ge  1-\frac{1}{2(d+1)},  \label{eq:regular-completeness}\\
     \val(\fI)\le \eta &\Longrightarrow \val(\fI')\le 1 - \frac{1-\eta}{d+1}. \label{eq:regular-soundness}
\end{align}

\end{theorem}

The regularization process follows closely that of Dinur's treatment~\cite{Dinur07pcp}. Define a new graph $G'=(V', E')$, such that  
\begin{align*}
    & V' = \{(v,e)\in V\times E: v\text{ is incident to }e\} \\
    & E' = E'' \cup \bigcup_{v\in V} E_v, 
\end{align*}
where $E'' =\{((v,e),(u,e)):(v,u)=e\in E\}$ and $E_v$ is the set of edges in the $d$-regular expander graph $G_v=( V_v = \{(v,e)\in V'\}, E_v)$, for some constant $d$, whose Cheeger constant is at least 2.\footnote{A random graph $G_v$ would be good, and various explicit constructions are known. We refer interested readers to the wonderful survey on this topic~\cite{HooryLW06}.}  In words, we replace every vertex $v$ with a cluster of vertices of size equal to the number of edges that $v$ is incident to in $G$. Within each cluster, the vertices are connected based on expander graphs. For every edge, $e=(u,v)$ in the original graph, connect the vertex $(u,e)$ with vertex $(v,e)$ in the new graph.  
The constraints $f'$ on $E''$ will be like that of $f_e$ on $E$. In particular, $f'_{((u,e),(v,e))} = f_{(u,v)}$. Further, the constraints on edges $E_v$ will be the equality constraints, which can be represented as a bijective map.
This new UG instance $\fI'$ satisfies that described in \cref{thm:regularization}.
Therefore, for the regular graph, $(1-\frac{1}{2(d+1)}, 1-\frac{1-\eta}{d+1})$-GapUG problem is NP-hard.

We verify the above claim. First note that in the new graph $G'$, the number of edges blows up by a factor of $d+1$. This is because
\[
    |E'|= |V'|(d+1)/2 = |E|(d+1).
\]
Now for (\ref{eq:regular-completeness}), a faithful prover will assign the label of a vertex $v$ in $G$ to the vertices of the form $(v,e)$. Then the number of unsatisfied constraints is unchanged, but the fraction decreases by a factor of $d+1$.

For (\ref{eq:regular-soundness}), let $\ell'$  be the labeling that the adversarial prover chooses. Let $\ell$ be the labeling on $V$ induced by $\ell'$ such that for any $v\in G$, $\ell(v)$ is chosen to be the majority labeling of $\{(v,e): e\sim v\}$ (break ties arbitrarily). For any $e = (u,v)\in E$ that is not satisfied by $\ell$, either both $\ell'((u,e))=\ell(u)$ and $\ell'((v,e))=\ell(v)$, then the edge $((u,e), (v,e))$ is not satisfied. Or, one of the vertices $(u,e), (v,e)$ is not labeled by the majority label. The following lemma proves that within any cluster, the number of unsatisfied constraints is at least the number of vertices with minority labels. Therefore, the total number of unsatisfied constraints in $G'$ with $\ell'$ is at least that of $G$ with labeling $\ell$.

\newcommand{\unsat}{\mathrm{uneq}}
\begin{lemma}
Suppose the $d$-regular graph $G=(V,E)$ has Cheeger constant at least 2 and $\ell$ be some labeling $\ell:V\to\Sigma$. Let $q$ denote the majority label on $V$, and let $\mathrm{uneq}(G)$ denote the number of edges $(u,v)$ in $G$ such that $\ell(v)\not=\ell(u)$. Then
\[
    \unsat(G) \ge | \{v\in V: \ell(v)\not=q\} |.
\]
\end{lemma}
\begin{proof}
The vertex set is partitioned by the labeling $\ell$ into, say, $m$ subsets $V_1, V_2, \ldots, V_m$. Let $n_1 \ge n_2 \ge \ldots \ge n_m$ be the number of vertices in each subset. 

If $n_1 \ge n/2$, then statement holds by the expansion property of $G$: 
\[
    \unsat(G) \ge E(V_1, \bar V_1) \ge 2 |\bar V_1|.
\]
If $n_1 < n/2$, we bound the number of edges within each subset:
\begin{align*}
   \frac{1}{2} \sum_{i\in \{1,\ldots, m\} } E(V_i, V_i) 
        &\le \sum_{i\in \{1,\ldots, m\} } (d|V_i| - 2|V_i|)/2  \\
        &= \frac{d n}{2} -n,
\end{align*}
where the first inequality uses the expansion property of $G$ as $|V_i| < n/2$. Therefore $\unsat(G) \ge n \ge |\bar V_1|$.
\end{proof}

We verify that for any $\eta < 1/4(d+1)$,
\[
    \left( 1 - \frac{1}{2(d+1)} \right)^2 > 1 - \frac{1-\eta}{d+1}.
\]
Therefore, by~\cref{thm:ug-in-qma2}, we have
\begin{theorem} With constant completeness and soundness gap,
 $\textup{NP} \subseteq \QMA^+_{\log}(2).$
\end{theorem}

One can work with various other approaches to prove the above theorem. For example, one can work with the 3COLOR problem, or work with the $\PCP$ characterization of $\NP$. Looking ahead, to take advantage of the $\PCP$ characterization will be the approach we take to show $\NEXP=\QMA^+(2)$.

% nexp
\section{\texorpdfstring{$\NEXP=\QMA^{+}(2)$}{NEXP=QMA+(2)}}\label{sec:nexp}

In this section, we scale up our previous result to $\NEXP = \QMA^{+}(2).$
The direction that $\QMA^+(2)\subseteq \NEXP$ follows the same trivial argument 
that $\QMA(2)\subseteq \NEXP$---guess the quantum proofs. Our focus will be on the other direction. The starting point would be a $\PCP$ for $\NEXP$.
For the moment, we abstract things out and focus on the constraints satisfaction problem (CSP) with the understanding that the CSP system will come from the corresponding $\PCP$.

\begin{definition}
An $(N,R,q,\Sigma)$-CSP system $\cons$ on $N$ variables with values
in $\Sigma$ consists of a set (possibly a multi-set) of $R$ constraints
$\{\pred_{1},\pred_{2},\ldots,\pred_{R}\}$, and the arity of each
constraint is exactly $q$. The value of $\cons$, denoted $\val(\cons)$, is the maximum
fraction of the satisfiable constraints over all possible assignment
$\sigma:[N]\to\Sigma$. The $(1,\delta)$-GapCSP problem is to distinguish
whether a given system $\cons$ is such that 
(\textbf{Yes}) $\val(\cons)=1$ or 
(\textbf{No}) $\val(\cons)\le\delta.$
\end{definition}

For any CSP system $\cons$, we think of a bipartite graph $G_{\cons}$
where the left vertices are the constraints and the right vertices
are the variables. Whenever a constraint queries a variable there
is an edge in the graph between the corresponding vertices. For any
$j\in[R],$ let $\AdjC(j)$ denote the list of variables that $\pred_{j}$
queries; and for any $i\in[N]$, let $\AdjV(i)$ denote the list of
constraints that query variable $i$. An efficient CSP system
$\cons$ should satisfy that for any $j\in[R],$ there is an algorithm
that compute
$\pred_{j}$ in time $\poly\log (NR)$. That includes deciding which variables
are queried by $\pred_{j}$, and based on the values of the relevant variables
compute $\pred_{j}$. For our purpose, we
require stronger properties, which we refer to as \emph{double explicitness}. Informally,
we require that given any variable $i$, we can also ``list'' the constraints that query $i$ efficiently.
\begin{definition}[Doubly explicit CSP]
For any (family of) $(N,R,q,\Sigma)$-CSP system $\cons,$ we say that  $\cons$ is doubly explicit if the following are
computable in time $\poly\log(NR)$:
\begin{enumerate}
\item The cardinality of $\AdjC(j)$ for any $j\in[R]$ and the cardinality
of $\AdjV(i)$ for any $i\in[N]$.\vspace{2mm}
\item $\AdjGloC:[R]\times[N]\to[q]$, such that $ \AdjGloC(j,i) = \iota$ if $i$ is $\iota$th variable that $\pred_{j}$ queries.\footnote{If $\pred_j$ does not query $i$, we don't care about the value of $\AdjGloC$. Similarly for $\AdjGloV$.}\vspace{2mm}
\item $\AdjLocC:[R]\times[q]\to[N]$, such that $\AdjLocC(j,\iota)$ is the $\iota$th variable that $\pred_{j}$ queries.\vspace{2mm}
\item $\AdjGloV:[N]\times[R]\to[R],$ such that $\AdjGloV(i,j)=\iota$ if $\iota$ is the index of constraint $j$ in
$\AdjV(i)$.\vspace{2mm}
\item $\AdjLocV:[N]\times[R]\to[R]$ such that for any $i\in[N]$ and $\iota\in[|\AdjV(i)|],$
let $j=\AdjLocV(i,\iota)$, then $\iota$th constraints in $\AdjV(i)$
is $\pred_{j}$. 
\end{enumerate}
\end{definition}

In words, in the bipartite graph $G_\cons$. For each vertex, say $i\in [N]$, there is a local view of its neighborhood $\AdjV(i)$. We should be able to efficiently switch from this local representation to a global representation, by $\AdjLocV(i, \cdot)$, and vice versa. 

Another property we require is the \emph{uniformity}, defined below.
\begin{definition}[$T$-Strongly uniform CSP]
For any $(N,R,q,\Sigma)$-CSP system $\cons$ and $T\in\ZZ$, we
say that $\cons$ is $T$-strongly uniform if the variable set $[N]$
can be partitioned into at most $T$ subsets $V_{1}\cup V_{2}\cup\cdots\cup V_{T}$,
such that the cardinality of $\AdjV(i)$ for any variable $i$ only
depends on which subset it belongs to. Furthermore, let $\tau:[N]\to[T]$,
such that $\tau(i)=j$ if $i\in V_{j}.$ Then $\tau(i)$ can be computed
in time $\poly\log (NR).$ 
\end{definition}

Given some $(N,R,q,\{0,1\})$-CSP system $\cons$ that is $T$-strongly uniform for some constant $T$ and is strongly explicit. Then  it is  $\NEXP$-hard to decide whether $\val(\cons)=1$ or $\val(\cons)<\delta$ for some absolute constant $\delta$. This CSP $\cons$ comes from the efficient
$\PCP$ for $\NEXP$. Although not all $\PCP$ satisfies doubly explicitness or uniformity, there is some $\PCP$ construction that enjoys these properties. We discuss such $\PCP$ in more detail and prove the related properties in \cref{sec:pcp-nexp}. 
\begin{theorem}[$\PCP$ for $\NEXP$]
\label{thm:explicit-NEXP-PCP} There is a $\PCP$ system for a $\NEXP$-complete problem, in which the verifier tosses $\poly(n)$ random bits and makes a constant number of queries to the proof $\Pi$ such that if the input is a ``Yes'' instance, then the verifier always accept; if the input is a ``no'' instance, then the verifier accepts with probability at most $\delta$ for some constant $\delta$.
Furthermore, this $\PCP$ is doubly explicit and $T$-strongly uniform for some constant $T$.
\end{theorem}

This $\PCP$ gives rise to a $(1,\delta)$-GapCSP instances for some
$(N=2^{\poly(n)},R=2^{\poly(n)},q=O(1),\{0,1\})$-CSP system that are $T$-strongly
uniform for some constant $T$ and doubly explicit. 
In the remainder of the section, our goal is to prove the following theorem:
\begin{theorem}\label{thm:nexp-in-qma2}
For any constant strongly uniform and doubly explicit $(N,R,q,\Sigma)$-CSP
system $\cons$, there is a $\QMA^+(2)$ protocol that solves the $(1,\delta)$-GapCSP
problem for $\cons$ with constant completeness and soundness gap.
\end{theorem}

\cref{thm:explicit-NEXP-PCP} together with~\cref{thm:nexp-in-qma2} imply that
\begin{theorem}
$\NEXP \subseteq \QMA^+(2)$ with constant completeness and soundness gap.
\end{theorem}

In the next three subsections, we prove \cref{thm:nexp-in-qma2}.

\subsection{Explicit Regularization}
\label{sec:explicit-regularization}
The first step towards proving~\cref{thm:nexp-in-qma2} is regularization for the CSP $\cons$, very much like that in~\cref{thm:regularization}.
The main technical issue is that everything happening in the previous case needs to be efficient for the exponentially large expander graphs. Fortunately, explicit constructions of expander graphs are very well-studied.

\begin{theorem}[Explicit regular expander graphs~\cite{Lubotzky11,Alon21}]
\label{thm:explicit-expander}There is some constant
$d$, for which we have the following explicit constructions on expander graphs with
Cheeger constant at least $2$:
\begin{enumerate}
\item \label{enu:alon-expander}For any $n$, there is a $d$-regular expander
graph on $n$ vertices. 
\item \label{enu:lubotzky-expander}For any prime $p>17$, there exists
a $d$-regular expander graph on $n=p(p^{2}-1)$ vertices.
Furthermore, the graph $G$ can be decomposed into $d$ matchings
$\pi_{1},\pi_{2},\ldots,\pi_{d}$, such that given $i\in[n]$ and
$j\in [d]$, there is a $\poly\log(n)$-time algorithm
$\Pi_{G}:[n]\times[d]\to[n]$, such that 
\[
\Pi_{G}(i,j)=\pi_{j}(i).
\]
\end{enumerate}
For both constructions, given $i\in[n]$, the neighbors of $i$ can be listed in time $\poly\log(n)$.
\end{theorem}

Since the second construction of expander graphs from the above theorem does not work for any number of vertices, we also need the following theorem about primes in short intervals.
\begin{theorem}[Primes in short intervals~\cite{Cheng10primes}]
\label{thm:primes} There is some absolute constant $n_{0}$, such
that for any integer $n>n_{0},$ there is a prime between the interval
$[n-4n^{2/3},n]$.
\end{theorem}

With the above tools at our disposal, we discuss the explicit regularization for this exponentially large CSP $\cons$. Replace the variable $i$ with a cluster of variables labeled $(i,\iota)$
for $\iota\in[n_{i}]$, where $n_i = |\AdjV(i)|$. If $n_i < n_0 $ for some absolute constant $n_0$ (this can be a larger constant than that in~\cref{thm:primes}), then we can simply use the expander graph provided by~\cref{thm:explicit-expander}~\ref{enu:alon-expander}. For $n_i \ge n_0$, we use the expander graph provided by~\cref{thm:explicit-expander}~\ref{enu:lubotzky-expander}. In particular, let $p_i$ be some prime such that
\begin{align*}p_{i}\in[\lfloor n_{i}^{1/3}\rfloor-4\lfloor n_{i}^{1/3}\rfloor^{2/3},\lfloor n_{i}^{1/3}\rfloor].
\end{align*}
The existence of $p_i$ is guaranteed by~\cref{thm:primes}. Let $n'_i:= p_i (p_i^2-1) \in [n-O(n^{8/9}), n],$ and let 
\begin{align*}
 & V'_{i}=\{(i,j):j\le n_{i}'\},\\
 & V_{i}''=\{(i,j):n_{i}'<j\le n_{i}\}.
\end{align*}
Depending on $n_0$, $|V_i''| \le \eta n_i$ for $\eta=\eta(n_0)$. As we set $n_0$ to be a large enough constant, $\eta$ is arbitrarily small.
Connect the vertices in $V'_{i}$ by a $d$-regular expander graph
$G_{i}$, whose existence is guaranteed by~\cref{thm:explicit-expander}~\ref{enu:lubotzky-expander}. For all vertices in $V_{i}''$, add $d$ self-loops. Associate
these edges with equality constraints. Let $\cons'$ denote the new
CSP instance. Recall that $q$ is the number of variables queried by each constraint in $\cons$

\begin{claim}
\label{claim:explicit-regularization}
If $\val(\cons)=1$, then $\val(\cons')=1$. If $\val(\cons)=\delta<1$, then the total number of unsatisfied constraints in $\cons'$ is at least $(1-\delta-q\eta)R$.
\end{claim}
\begin{proof}
The analysis is similar to that of \cref{thm:regularization}. If $\val(\cons)=1$, then just assign the same label to all variables in $V_i'$ and $V_i''$ based on the correct label for $\cons$.
If $\val(\cons)<1$, whenever some constraints $\pred_{i}\in\cons$ is not satisfied by
the majority labeling for the queried variables, then either (1) $\pred_{i}$
is still not satisfied in $\cons'$ for the corresponding constraint or (2) at least one of the queried variables is not
colored by the majority label. The difference in the current case from that of \cref{thm:regularization} is that all the variables from  $V''_i$ can have arbitrary values without hurting any equality constraints. Since  $|V''_i|\le \eta n_i$ for any $i\in[N]$, in total
\[
    \left |\bigcup _{i\in [N]} V''_i \right| \le  \eta \sum_{i\in[N]} n_i = \eta q R.
\]
Therefore the total number of unsatisfied constraints is at least 
$(1-\delta - q\eta) R$.
\iffalse
Let $s_{1}$ be the number of unsatisfied constraints in $\cons'$ that correspond to some constraints in $\cons$, and let $s_{2}$ be the number
of unsatisfied equality constraints introduced by the expander graphs,
then we have:
$s_{1}+s_{2}\ge(1-\delta-q\eta)R.$ Finally, note that the total number of constraints in $\cons'$ is 
\[
\sum_{i\in[N]} \left(\sum_{(i,j)\in V'_i} \frac{d+1}{2} + \sum_{(i,j)\in V''_i}\left(\frac{1}{2} + d \right)\right) \le R(d+1+ 
\epsilon q /2),
\]
where in the second term on the LHS, the self-loops are not double counted and thus not divided by 2. 
The proof is finished.
\fi
\end{proof}

\subsection{The Protocol}
In the protocol, the provers are supposed to provide the following state:
\begin{equation}
    \ket\psi = \sum _ {j\in [R]} \ket j \ket {v_j}, \label{eq:csp-encoding}
\end{equation}
where $v_j \in \C ^{|\Sigma^q|}$, which should indicate that the $q$ variables with order listed in
$\AdjC(j)$ queried by $\pred_{j}$ have value $v_{j,1},v_{j,2},\ldots,v_{j,q}$,
respectively. This way of encoding is very convenient for evaluating whether
each constraint is satisfied or not. But requires some work to make sure that
the values $v_j$ are consistent: Different constraints will share variables and the 
value of any variable across different constraints should be the same. Recall that,
in the previous section when we discuss the regularization step for our CSP $\cons$ with 
variable set $V=[N]$ and constraints $\pred_1, \ldots, \pred_R$, from which we obtain a 
new CSP $\cons'$ such that each variable appears in exactly $d$ number of the new 
constraints. Furthermore, a new variable in $\cons'$ will be a tuple composed of a 
variable $i\in V$ and a constraint 
$\pred_j$ that queries $i$. 
Therefore, our way of encoding in (\ref{eq:csp-encoding}), in a sense, is to write the 
superpositions of the new variables along with their values in the regularized CSP. 

Let $n_{1},n_{2},\ldots,n_{T}$ be the cardinalities of $\AdjV(i_{1}),\AdjV(i_{2}),\ldots\AdjV(i_{T})$
where $i_{1},i_{2},\ldots,i_{T}$ are arbitrary variables from $V_{1},V_{2},\ldots,V_{T}$,
respectively. 
Next, we describe our protocol for the CSP instance that we have.

\noindent\fbox
{\begin{myalg}
[Protocol for strongly uniform and doubly explicit CSP]
\ignorespacesafterend
Let $\epsilon$ be some small enough constant, and $k$ some large enough constant.

\textbf{Prover} provides:
\begin{enumerate}
\item $T$ primes $p_{1},p_{2},\ldots,p_{T}$, such that $p_{i}\in[\lfloor n_{i}^{1/3}\rfloor-4\lfloor n_{i}^{1/3}\rfloor^{2/3},\lfloor n_{i}^{1/3}\rfloor].$ 
\item $\Psi:=2k$ copies of the state
\[
\sum_{j\in [R]}|j\rangle|v_j \rangle, \qquad \forall j\in [R],\, v_j\in \Sigma^q.
\]
\item $\Phi:=2k$ copies of the state
\[
\sum_{j\in[R]}|j\rangle\sum_{v\in\Sigma^{q}:v\not=v_j}\frac{|v\rangle}{\sqrt{|\Sigma|^{q}-1}}.
\]
\end{enumerate}
\textbf{Verifier}:
\begin{enumerate}
\item Test if $p_{1},p_{2},\ldots,p_{T}$ are primes satisfying the size
constraints, \emph{reject} if not.
\item Symmetry test on $\Psi$ and $\Phi.$
\item Sparsity test II on $(\Psi,\Phi)$ with target sparsity $|\Sigma|^{-q}$ and precision $\epsilon$ 
\item Validity test on $\Psi$.
\item Constraints test $\Psi$.
\end{enumerate}
\end{myalg}}

To remove any ambiguity, when taking the validity test, the valid set will be
\[
    \cV := \left\{
        \sum_{j\in[R]} \ket j \ket {v_j}: v_j \in \Sigma^q
    \right\}.
\]
Since $\Sigma$ is of constant size, and $q$ is a constant, $\Sigma^q$ is still of constant size.

The constraints test will be used to check the new constraints of our CSP after the regularization. But before we formally describe the constraints test, we make some preparations. Let $ H = \C^R \otimes \C^{|\Sigma|^q} \otimes \C^N \otimes \C^{|\Sigma|}$. 
The first register is the constraint register. The second register is used to encode the 
values of the $q$ variables queried by the constraint stored in the first register. The 
third register is the variable register to store the variable name. The last register is 
used to store the value of the variable in the third register.
Now we define three quantum channels that will be used to manipulate our state in the constraints test.
\begin{itemize}[label={$\bullet$}]
    \item $\cA$, the operator that converts a given state from (\ref{eq:csp-encoding}) to an actual superposition of the new variables from $\cons'$ together with their values.
    \item $\cM_k$ for $k\in [d]$, the operator that ``implements'' the $k$th one after we decompose the d-regular expander graphs into matchings.
    \item $\cB$, the operator that given $\ket j\ket {v_j}$, evaluates if $\pred_j$ outputs $1$ if the values of the variables it queries are given by the string $v_j$.
\end{itemize}

Precisely, let $\cB$ acting on $\C^R\otimes \C^{q|\Sigma|} \otimes \C ^2$ be such that
\[
    \cB: \ket j \ket v \ket 0 \mapsto \ket j \ket v \ket {\pred_j(v)}.
\]
Recall that the constraints of $\cons'$ consist of that from $\cons$ and the consistency constraints induced by the expander graphs and self-loops we add. As $\cB$ checks if the value $v$ satisfies the constraints $\pred_j$, it takes care of the first kind of constraints of $\cons'$.

Define the operator $\cA$ acting on $ H = \C^R \otimes \C^{|\Sigma|^q} \otimes \C^N \otimes \C^{|\Sigma|}$ such that
\[
    \cA: |j\rangle|v\rangle|0\rangle\ket{0} \mapsto \frac{1}{\sqrt q} \sum_{\iota=1}^{q} |j\rangle\ket v  \ket{i_\iota} \ket{v_\iota},
\]
where $x_{i_1}, x_{i_2},\ldots, x_{i_q}$ are the variables listed in $\AdjC(j)$. In words, given the constraints $j$, and the values $v$ to the variables that $j$ queries, we put the third and fourth register (the variable register) into the superposition of the variables in $\AdjC(j)$ together with their value based on $v$. 

Next, we define $\cM$ formally. Recall that for any variable $i\in[N]$, after regularization, the set of variables constructed from $i$ includes
\begin{align*}
 & V'_{i}=\{(i,j):j\le n_{i}'\},\\
 & V_{i}''=\{(i,j):n_{i}'<j\le n_{i}\}.
\end{align*}
The new constraints include an expander $G_i$ on $V'$ and self-loops on $V_i ''$. We can decompose $G_i$ into $d$ matchings, and for variables in $V''_i$, they are matched with themselves. 
For any $k\in [d]$, let $\cM_k$ be the operator such that:
\begin{align*}
    \cM_k: \ket j \ket v \ket i \ket {v'} \mapsto 
        \ket {j'} \ket v \ket i \ket {v'},
\end{align*}
where 
\begin{align} \label{eq:j-p}
 j'=\begin{cases}
    \AdjLocV(i,\Pi_{G_{i}}(\iota,k)), & \iota\le n_{i}',\\
    j, & \text{otherwise},
\end{cases}
\end{align}
\[
\iota=\AdjGloV(i,j).
\]
That is, suppose we take the $k$th matching to permute the variables in $\cons'$, then $j'$ in (\ref{eq:j-p}) determines that $(i,j)\in\cons'$ should be switched to $(i,j')\in \cons'$. But the expander graphs are labeled by $\{1,2,\ldots,n_i'\}$, corresponding to indices of $\AdjV(i)$, to obtain the actual constraint $\pred_{j'}$, we need to convert from local index to global index, and later convert it back. 

$\cA$ together with $\cM_k$ takes care of the consistency constraints just like how we do it for
UG games. Take a pair of equal states $\ket \psi$ and $\ket\phi$ supposed to be valid. Apply $\cA$ to both states. But apply $\cM_k$ only to $\ket\phi$. Now the two states are equal if the original states encode a consistent value for all constraints, except we should ignore the second register. To get rid of the second register, we make a measurement. In particular, let
\[
    \ket \mu = \frac{1}{|\Sigma|^{-q}}\sum_{v\in \Sigma^q} \ket v. 
\]
Consider the measurement $M=\{\Pi_{\ket\mu \langle \mu|}, 1- \Pi_{\ket\mu \langle \mu|} \}$.
It's easy to see that after the measurement, with probability $p=|\Sigma|^{-q}$, the 
second register is set to $\ket\mu$ and thus disentangled from the other registers.
Since we have a larger number of provers, with $p$ fraction of proofs left is enough.

The next claim certifies that $\cA, \cM, \cB$ are all valid quantum operations.
\begin{claim}
$\cA, \cB, \cM_k$ can be implemented by $\BQP$ circuits.
\end{claim}
\begin{proof}
First, consider the implementation of $\cA$. Let $H' = H \otimes \C ^q$, where the new register will be some working space. Take the following sequence of manipulations:
\begin{enumerate}
    \item Get a superposition on the last register:
    \[
        \ket j \ket v \ket 0 \ket 0 \ket 0
        \mapsto \frac{1}{\sqrt q}\sum_{\iota=1}^q \ket j \ket v \ket 0 \ket 0 \ket \iota.
    \]
    \item From the second and the last register, compute $v_\iota$ and set the fourth register accordingly:
    \[
        \ket j \ket v \ket 0 \ket 0 \ket \iota
        \mapsto \ket j \ket v \ket 0 \ket {v_\iota} \ket \iota.
    \]
    \item Compute $\AdjLocC(j,\iota)$, and put it in the third register:
    \[
        \ket j \ket v \ket 0 \ket {v_\iota} \ket \iota
        \mapsto
        \ket j \ket v \ket {\AdjLocC(j,\iota)} \ket {v_\iota} \ket \iota.
    \]
    \item Set the last register to $0$:
    \[
        \ket j \ket v \ket i \ket {v_\iota} \ket \iota
        \mapsto
        \ket j \ket v \ket i \ket {v_\iota} \ket 0.
    \]
\end{enumerate}
The final step is valid because it is the inverse of the following operation
\[
    \ket j \ket v \ket i \ket {v_\iota} \ket 0
    \mapsto
    \ket j \ket v \ket i \ket {v_\iota} \ket {\AdjGloC(j,i)}.
\]
Since $\AdjGloC$ and $\AdjLocC$ can be computed efficiently classically due to the explicitness of $\cons$, the above steps are efficient.

The situation for $\cM_k$ is similar. Consider $H'= H \otimes \C^R$. Do the following:
\begin{enumerate}
    \item Based on constraint $j$ and variable $i$, and $k$, compute $j'$ as in (\ref{eq:j-p}), put $j'$ in the working space.
    \[
        \ket j \ket v \ket i \ket {v_i} \ket 0
        \mapsto
        \ket j \ket v \ket i \ket {v_i} \ket {j'}.
    \]
    \item Set the first register to be 0.
    \[
        \ket j \ket v \ket i \ket {v_i} \ket {j'}
        \mapsto
        \ket 0 \ket v \ket i \ket {v_i} \ket {j'}.
    \]
    \item Swap the contents of the first and the last registers.
\end{enumerate}
The second step is a valid step because it is the inverse of the first operation (acting on a different order of the registers). Since $\AdjLocV$, $\AdjGloV$ and $\Pi_{G_i}$ are efficient classically due to the explicitness of our CSP system and expander graphs provided in~\cref{thm:explicit-expander}, $\cM_k$ is efficient.

$\cB$ can be implemented efficiently because each constraint can be verified in polynomial time classically.
\end{proof}

With the above preparation, we now describe the constraints test.
\iffalse
for two copies of a valid assignment from $\cV$. Once we understand this test, then the more technical case when supplied with general $\epsilon$-tilted states that is close to $\cV$ guaranteed by the symmetry, sparsity, and validity tests can be analyzed analogously to the previous sections. 
\fi

\noindent\fbox{\begin{minipage}[t]{1\columnwidth - 2\fboxsep - 2\fboxrule}%
\uline{Constraints test}

\textbf{Input: $\Psi_0, \Psi_1$}, each is a set of $k$ states for some large constant $k$. 

Pair the states in $\Psi$ and $\Phi$. 

For each pair $\qpsi{}$ and $\qphi$, with probability $2d/(2d+1)$ take the consistency check, with the
remaining probability take the inner constraints test
\begin{enumerate}
\item Consistency check
\begin{itemize}
    \item Apply $\cA$ to $\qphi$ and $\qpsi{}$.
    \item Apply $\cM_k$ to $\qphi$ for a uniformly random $k\in [d]$.
    \item Measure the second register of $\ket\psi, \ket\phi$ based on $M$, if either measurement does not output $\ket\mu$, ignore this pair.
    \item SwapTest on $\qpsi{}$ and $\qphi.$
\end{itemize}

\item Inner constraints test
\begin{itemize}
    \item With probability $1-|\Sigma|^{-2q}$, ignore this pair.
    \item Apply $\cB$ to $\qpsi{}$
    \item Measure the third register, \emph{Accept }if $1$ is observed.
\end{itemize}
\end{enumerate}

\emph{Accepts} if more than $\theta$ fraction of the pairs (that are not ignored) get accepted, 
where
    \[
        \theta = 1- \frac{1-\delta}{4(2d+1)}.
    \]
\end{minipage}}

\subsection{Analysis}
\begin{lemma}[constraints test]\label{lem:constraints-test}
Suppose $\val(\cons)=1$, then a faithful prover passes the constraints
test with probability $1$. On the other hand, if $\val(\cons)\le\delta,$
then on any two valid pair of states $\qpsi{}$ and $\qphi$, the constraints test  rejects
with probability at least $(1-\delta)/(2(2d+1)).$
\end{lemma}

\begin{proof}
Let $s_1$ be the fraction of $\ket j \ket v$ such that $\pred_j(v)=0$, and let $s_2$ be the fraction of unsatisfied edges coming from the expander graphs implicitly used in the consistency test, 
then by~\cref{claim:explicit-regularization}, we have:
\[
s_1 + s_2 \ge (1-\delta-q\eta) R.
\]

The consistency test partitions all pairs of the same variable in
different constraints into $d$ matching. Let $\lambda_{1},\lambda_{2},\ldots,\lambda_{d}$
denotes the fraction of inconsistency pairs in matching $1,2,\ldots,d$,
respectively. Analogous to the previous analysis, the probability the
SwapTest accepts is 
\[
\Exp_{i\in[d]}\left[\frac{1+(1-\lambda_{i})^{2}}{2}\right]\le1-\frac{1}{2}\Exp_{i\in[d]}\lambda_{i}=1-\frac{s_{2}}{2dR}.
\]
In the inner constraints test, $1$ is observed with probability $1-s_1/R$.
Therefore, in total the reject probability is at least
\[
    \frac{1}{2d+1}\cdot\frac{s_1}{R}+\frac{2d}{2d+1}\cdot\frac{s_2}{2dR}\ge\frac{1-\delta-q\eta}{2d+1}.
\]
    By picking the suitable $n_0$, we make sure $q\eta < (1-\delta) /2$, thus the reject probability is at least $(1-\delta)/(4d+2)$.
\end{proof}

\begin{proof}[Proof of \cref{thm:nexp-in-qma2}]
The completeness in \cref{thm:nexp-in-qma2} is completely analoguous to the analysis in \cref{thm:ug-in-qma2}.
The soundness in \cref{thm:nexp-in-qma2} is also similar to the previous analysis. If $\Psi$ supplied by the prover is not an $\epsilon$-tilted state or is far from $\cV$, then the symmetry test, sparsity test, and validity test will catch it. Therefore, we can assume that essentially all states in $\Psi$ are close to some state $\qpsi{}\in\cV$. By choosing $\epsilon$ small enough and the size of $\Psi$ sufficiently large, by \cref{lem:constraints-test} and Chernoff bound, the fraction of accepted states in the constraints test will be less than $\theta$ with high probability in the constraints test. 
\end{proof}

% gap amp
\section{$\QMA(2)$ v.s. $\QMA^+(2)$}
\label{sec:qma2-vs-qma2+}

Our discussion so far has focused on the somewhat ``artificial'' complexity class $\QMA^+(2)$. In section, we aim to demystify this complexity class. In particular, we show the relationship between $\QMA(2)$ and $\QMA^+(2)$, and provide a concrete reason why the $\QMA^+(2)=\NEXP$ is an interesting result and what it reveals about the open problem $\QMA(2)$ v.s. $\NEXP$.

\subsection{\texorpdfstring{Simulations between $\QMA(2)$ and $\QMA^+(2)$}{Simulations between QMA(2) and QMA+(2)}}
\label{sec:simulation}

We start by showing that $\QMA(2)$ and $\QMA^+(2)$ can actually simulate each other to some extent. In particular, we will prove the following two lemmas.

\begin{lemma}\label{lem:qmaplus-in-qma}
For any constant $c,s\in [0,1] $,
\[
  \QMA^+(2, c,s)\subseteq \QMA(2, c, 4s).
\]
\end{lemma}
\begin{lemma}\label{lem:qma-in-qmaplus}
For any constant $c,s\in [0,1] $ with $c-s > 1/\mathrm{poly}(n)$,
\[
  \QMA(2, c,s)\subseteq \QMA^+(2, 1-e^{-O(\mathrm{poly}(n))}, e^{-O(\mathrm{poly}(n))}).
\]
\end{lemma}

For the purpose of settling the problem ``$\QMA(2)$ v.s. $\NEXP$'', \cref{lem:qmaplus-in-qma} is probably more interesting. In fact, an immediate consequence of~\cref{lem:qmaplus-in-qma} is the following.
\begin{corollary}\label{cor:qmaplus-in-qma}
$\QMA^+(2,c,s) \subseteq \QMA(2)$ for any $c-s \ge \gapLarge$.
In particular, suppose $\NEXP\subseteq \QMA^+(2)$ with a gap at least $3/4+1/\poly(n)$, then 
\[
    \QMA(2) = \NEXP.
\]
\end{corollary}
\begin{proof}
    Suppose a given $\QMA^+(2,c,s)$ protocol has completeness $c$ and soundness $s$. By our assumption,
    \begin{align}
        c &\ge s + 3/4 + 1/\poly(n),    \label{eq:cs-gap-3}\\
        s &\le 1/4 - 1/ \poly(n).        \label{eq:s-bound-1}
    \end{align}
    By \cref{lem:qmaplus-in-qma}, there is a $\QMA(2)$ protocol that solves the same problem with completeness $c'=c$ and soundness $s'=4s$. Then
    \begin{align}
        c'-s' = c - 4s \ge 3/4 + 1/\poly(n) - 3s \ge 1/\poly(n),
    \end{align}
    where the first and second inequalities use (\ref{eq:cs-gap-3}) and (\ref{eq:s-bound-1}), respectively. Since $\QMA(2)$ admits strong gap amplification~\cite{HM13}, it shows that $\QMA^+(2,c,s)\subseteq\QMA(2)$. This finishes the proof.
\end{proof}

Next, we show a straightforward proof of \cref{lem:qmaplus-in-qma}. The fact that such proof works may be somewhat surprising.
\begin{proof}[Proof of \cref{lem:qmaplus-in-qma}]
Given any $\QMA^+(2)$ protocol $\calP$ with completeness $c$ and soundness $s$. We simply reuse the verfication procedure of $\calP$. In the completeness case, the benign provers will supply proofs with nonnegative amplitudes. Thus the completeness remains the same. 

In the soundness case. On input $x$, let $M$ be the operator representing the verification measurement of $\calP$ for accepting (i.e. $I-M$ corresponds to rejecting). Let $\ket\psi$ be an arbitrary quantum state, we compute the probability that $\calP$ accept, $\bra\psi M \ket\psi$. First, consider a decomposition of $\ket\psi$ as follows:
\[
    \ket\psi = \sqrt{\alpha_1}\ket {\psi_1} + \sqrt{\alpha_2}    \ket{\psi_2} + \sqrt{\alpha_3}\ket{\psi_3} + \sqrt{\alpha_4}\ket{\psi_4},
\]
where $\ket{\psi_1}$ is the normalized vector  consisting exactly the components of positive real amplitudes of $\ket\psi$; $\ket{\psi_2}$ is the normalized vector consisting exactly the components of positive imaginary amplitudes of $\ket\psi$; and so on. For example, if $\ket\psi = 1/\sqrt{3}(\ket{1} + \ket{2} + i \ket{3})$, then $\ket{\psi_1} = 1/\sqrt{2} (\ket{1} + \ket{2})$ and $\ket{\psi_2} = i \ket{3}$. Calculating
\begin{align}
    \bra\psi M \ket\psi 
    &
    = \sum_{i,j} \sqrt{\alpha_i\alpha_j} \bra{\psi_i} M \ket{\psi_j} 
    \nonumber \\
    &
    \le \sum_{i,j} \sqrt{\alpha_i\alpha_j}s
    = \left(\sum_i \sqrt{\alpha_i}\right)^2 s
    \nonumber \\
    &
    \le 4s,\nonumber
\end{align}
where the final step is due to Cauchy-Schwarz and that $\sum \alpha_i = 1.$ The reason why the second inequality holds is as follows:
In the case $i=j$, $\ket{\psi_i}$ is a quantum state without relative phase, equivalent to a quantum state with nonnegative amplitudes. 
Therefore $\bra{\psi_i}M\ket{\psi_i} \le s$ by the soundness of $\calP$. In the case when $i\not=j$, then because $M$ is PSD, 
\begin{align*}
     & (\bra{\psi_i}-\bra{\psi_j}) M 
     (\ket{\psi_i}-\ket{\psi_j}) \ge 0 \\
     & \quad \implies 
     \bra{\psi_i} M \ket{\psi_j} + \bra{\psi_j} M \ket{\psi_i}
     \le \bra{\psi_i} M \ket{\psi_i} + 
     \bra{\psi_j} M \ket{\psi_j} \le 2s.\qedhere
\end{align*}
\end{proof}

Finally, we settle \cref{lem:qma-in-qmaplus}, which is much more intuitive as a general quantum state can always be re-encoded into a quantum state with nonnegative amplitudes.
\begin{proof}[Proof of \cref{lem:qma-in-qmaplus}]
Given some $\QMA(2,c,s)$ protocol, which decides the measurement
$\{M_{x}\}.$ Since strong gap amplification for $\QMA(2)$ is known~\cite{HM13}, assume, without loss of generality, that
$c>3/4$ and $s<1/2$.
The $\QMA^{+}(2)$ protocol will reuse the same
measurement $\{M_{x}\}$. On input $x$, suppose that a faithful prover for $\QMA(2)$ sends
\[
|\psi_{x}\rangle=\sum_{r}(\alpha_{r}+\beta_{r}\cdot i)\ket r.
\]
Now the correponding $\QMA^{+}(2)$ prover will be asked to send the following:
\[
|\phi_{x}\rangle=\sum|\alpha_{r}|\ket r\ket{\sgn\alpha_{i}}\ket\R+|\beta_{r}|\ket r\ket{\sgn\beta_{i}}\ket\C.
\]
Arthur, on receiving a proof $|\phi\rangle,$ applies the following
transforms to the second and third registers, respectively:
\begin{align*}
 & \ket+\mapsto\frac{|0\rangle+|1\rangle}{\sqrt{2}},\qquad\quad\ket-\mapsto\frac{-\ket0+\ket1}{\sqrt{2}}.\\
 & \ket\R\mapsto\frac{|0\rangle+i\cdot|1\rangle}{\sqrt{2}},\qquad\ket\C\mapsto\frac{i\cdot\ket0-\ket1}{\sqrt{2}}.
\end{align*}
Arthur then measures the second and third registers. If after the measurement,
00 is observed, apply $M_{x}$ to $|\phi\rangle.$ Otherwise, Arthur can
output a random value. 

In the completeness case, the $00$ is observed with probability
$1/4$ and in this case the correct answer is outputted with probability
$c$ by the completeness of the $\QMA(2)$ protocol. In total,
Arthur outputs a correct value with a probability of at least
\[
\frac{c}{4}+\frac{3}{4}\frac{1}{2}>\frac{9}{16}.
\]
In the soundness case, if $00$ is observed, then Arthur is
correct with probability at most $s$ by the soundness of the $\QMA(2)$
protocol. Otherwise, Arthur is correct with probability $1/2$. In any
case, Arthur is correct with probability at most $1/2.$ In conclusion, 
\[
    \QMA(2,c,s)\subseteq \QMA^+\left(2, \frac{1}{2}+\frac{1}{16}, \frac{1}{2}\right). 
\]

Note that in the above argument, the measurement used in $\QMA^+(2)$ is adapted 
from $\QMA(2)$, whose soundness makes no assumption of the proofs that Arthur receives. In particular, the gap amplification for general $\QMA(2)$ can be applied. 
\end{proof}

Note that similar arguments for \cref{lem:qmaplus-in-qma,lem:qma-in-qmaplus} can also be adapted to $\QMA$ or $\QMA(k)$ v.s. $\QMA^+$ or $\QMA^+(k)$.

\subsection{Product Test and Gap Amplification for $\QMA(2)$} \label{sec:prod-test}
The previous discussion begs for a better understanding of gap amplification for $\QMA^+(2)$. Since many randomized/quantum computation models admit strong gap amplification, e.g. $\QMA$ and $\QMA(2)$. This gives hope that $\QMA^+(2)$ may also admit strong gap amplification. 
The story, however, will be much more complicated in view of the recent work of Bassirian, Fefferman and Marwaha~\cite{bassirian2023quantum}, building on the STOC version of this paper~\cite{JW23},  which curiously showed that $\QMA^+(1)=\NEXP$ also
with a constant gap.\footnote{It is not clear that their gap can be made as large as the one for $\QMA^+(2)=\NEXP$.}
Since in the large constant gap regime of $\QMA^+(1)$, we have $\QMA^+(1)=\QMA \subseteq \textup{PP}$ for the same reason as discussed in \cref{sec:simulation}, their result
rules out the strong gap amplification for $\QMA^+(1)$ unless $\NEXP\subseteq\textup{PP}$. Moreover, it also suggests that
techniques to amplify the gap for $\QMA^+(2)$ should crucially use the unentanglement assumption.

This subsection is focused on product test and its applications to gap amplification for $\QMA(2)$, which can actually be applied to $\QMA^+(2)$ as to be studied in the next part. In this subsection, though, our discussion will be based on Harrow and Montanaro's work~\cite{HM13} with a bit more careful treatment of the completeness and soundness gap. We include this discussion because it's central to sections to come.
Experts should feel free to skip.

Given some Hilbert space $H = H_1\otimes H_2 \otimes \ldots \otimes H_k$, each $H_i$ is of dimension $d_i$. A fundamental task is to test whether a given state $\ket\psi\in H$ is unentangled between the $k$ subsystems. In another word, the task is to test whether $\ket\psi$ is a product state, i.e., for some $\ket{\psi_i}\in H_i$,
$    \ket\psi = \ket{\psi_1} \otimes \ket {\psi_2} \otimes \cdots \otimes \ket{\psi_k}.$ The following test has been proposed for this purpose when there are many copies of $\ket\psi$.

\noindent\fbox{
\begin{myalg}[Prodcut Test]\label{algo:prod-test}\ignorespacesafterend
\textbf{Input}: $\ket\psi, \ket\phi \in H_1 \otimes H_2 \otimes \cdots \otimes H_k$.\footnotemark
\\
For $i=1$ to $k$:
Apply Swap Test to the $i$th subsystem of $\ket\psi$ and $\ket\phi$.
 \\
\textit{Accept} if all the swap test accepts.
\end{myalg}
}
\footnotetext{Ideally, the verifier should receive two copies of $\ket\psi$. However, in a $\QMA(2)$ protocol, the verifier can receive anything.}

Let $\pt(\ket\psi, \ket\phi)$ be the probability that $\ket\psi$ and $\ket\phi$ pass the product test.  Harrow and Montanaro first gave a formal analysis of the product test, and recently Soleimanifar and Wright improved the analysis.

\iffalse
\begin{theorem}[Harrow, Montanaro~\cite{HM13}]
    \begin{equation}
        \pt(\ket\psi, \ket\psi) = 2^{-k} \sum_{S\subseteq [k]} \trace(\psi_S\cdot\phi_S).
    \end{equation}
\end{theorem}
\fi
\begin{lemma}[Harrow, Montanaro~\cite{HM13}]\label{cor:product-test}
    \begin{equation*}
        \pt(\ket\psi, \ket\phi) \le \frac{1}{2} (\pt(\ket\psi,\ket\psi)+\pt(\ket\phi,\ket\phi)).
    \end{equation*}
\end{lemma}

\begin{theorem}[Soleimanifar, Wright~\cite{SW22matrix-product}]\label{thm:prod-test}
For any state $\ket\psi\in H_1 \otimes H_2\otimes \cdots \otimes H_k$, let
\begin{equation*}
  \max_{\phi_1\in H_1, \phi_2\in H_2, \ldots, \phi_k\in H_k}
    |\langle \psi \mid \phi_1\otimes\phi_2\otimes\cdots\otimes\phi_k \rangle|^2 = \omega.
\end{equation*}
Then, the following are two upper bounds on the probability that the product test passes. The first only works for  $\omega$ at least $1/2$ (and is optimal in that regime); the second is general:
\begin{align*}
    &\pt(\ket\psi, \ket\psi) \le 
        1 - \omega + \omega^2, & \omega \ge \frac{1}{2}; \\
    &\pt(\ket\psi, \ket\psi) \le \frac{1}{3} \omega^2 + \frac{2}{3}, & 0\le \omega \le 1.
\end{align*}
\end{theorem}

With the product test, one can do some gap amplification for $\QMA(2)$. \begin{lemma}[Gap amplification for $\QMA(2)$, c.f.~\cite{HM13}]
\label{lem:qma-gap-amp-naive}
For any $c-s>1/\poly(n)$,
    \[\QMA(2, c,s) \subseteq \QMA\left(2,~ 1-e^{-\poly(n)},~ \gapNaive \right).\]
\end{lemma}
\begin{proof}
Repeat the protocol by asking for more provers. Based on the standard Chernoff argument, for any $c-s>1/\poly(n)$, there is $k\in\poly(n)$ such that
\begin{align}
    \QMA(2, c, s) \subseteq \QMA(k, c', s'),
\end{align}
where $c'=1-e^{-\poly(n)}$ and $s'=e^{-\poly(n)}$.

Let the two provers simulate $k$ provers by requiring each prover to provide the $k$ proofs that are supposed to be sent in the $\QMA(k)$ protocol. Now with probability $p$, apply the product test; otherwise, apply the $\QMA(k)$ verification to the proofs provided by a random prover. 

In the completeness case, the product test passes with probability 1. Therefore, the verification passes with probability still $1-\exp(-\poly(n)).$ 

In the soundness case, let $\ket{\psi_1}, \ket{\psi_2}$ be the proofs provided by the first prover and second prover, respectively. Suppose that $\ket{\phi_1}, \ket{\phi_2}$ are the product states with the max overlap with $\ket{\psi_1}$ and $\ket{\psi_2}$, respectively,
\begin{equation}
    |\langle \psi_1, \phi_1\rangle|^2=\omega_1, \quad |\langle \psi_2, \phi_2\rangle|^2=\omega_2. \label{eq:psi-overlap-omega}
\end{equation}
By the soundness of the $\QMA(k)$ protocol, \cref{fact:trace_norm_acc}, 
and (\ref{eq:psi-overlap-omega}), in the case that the verification is 
applied to a random proof $\ket{\psi_i}$ for $i=1,2$, it passes with 
probability at most $s'+\sqrt{1-\omega_i}$.
Therefore, the total probability of acceptance, i.e., the soundness of the new $\QMA(2)$ protocol, can be bounded as below
\begin{align}
    p \pt(\ket{\psi_1}, &\ket{\psi_2}) +
        (1-p) \left( s' + \frac{1}{2}(\sqrt{1-\omega_1} + \sqrt{1-\omega_2})\right) 
        \nonumber \\
    &\le \Exp_{i \in \{1,2\}} [p \pt(\ket{\psi_i}, \ket{\psi_i}) + (1-p)\sqrt{1-\omega_i}] + (1-p)s'
        \nonumber \\
    &\le \Exp_{i\in \{1,2\}} [f(p, \omega_i)] + (1-p)s',\label{eq:naive-gap}
\end{align}
where the first step uses \cref{cor:product-test}; in the second step,
by \cref{thm:prod-test}, $f$ can be chosen as below
\begin{align*}
    f(p, \omega_i) 
        = p\cdot \frac{\omega_i^2 + 2}{3} 
     + (1-p)  \sqrt{1-\omega_i} .
\end{align*}
Set $p = 2/3$. It can be computed that
\[
    \frac{\d}{\d x} f(p, x) = \frac{4x}{9} - \frac{1}{6\sqrt{ 1-x }},
\]
and the critical points are 
\[
    x_0 = \frac{3}{4},\quad x_1 = \frac{1+\sqrt{13}}{8}.
\]
Turns out that the global maximum of $f$ is obtained in the boundary, \[\max_{x\in[0,1]} f(2/3, x) = f(2/3, 0) = 7/9.\]
That finishes the proof.
\end{proof}

The above amplification can be further boosted by a ``sequential'' repetition. The crucial observation is that the product test is a \emph{separable} measurement.
\begin{definition}[Separable measurement]
A measurement $M = (M_0, M_1)$ is separable if in the yes case, the corresponding Hermitian matrix $M_1$ can be represented as a conical combination of two operators acting on the first and second parts, i.e.,
for some distribution $\mu$ over the tensor product of PSD matrices $\alpha$ and $\beta$ on the corresponding space,
\[
M_1 = \int  \alpha \otimes \beta ~\mathrm{d}\mu.
\]
\end{definition}
\begin{fact}[Folklore] \label{fact:product-test-sep}
The product test is separable.
\end{fact}
\begin{proof}
It suffices to show that the swap test is separable. The swap test on $\C^d\otimes \C^d$ is the projection onto the symmetric subspace, i.e., the space consisting of states that are invariant under permutations. This projection can be expressed as (for some constant $C$)
\begin{equation}
\label{eq:projection-symm-subspace}
    \Exp_{\theta\in \C^d } [\ket{\theta,\theta}\bra{\theta,\theta}] = C\int_{\theta\in \C^d} ~\ket{\theta,\theta}\bra{\theta,\theta} ~\d \theta.
\end{equation}
\end{proof}
A detailed discussion about the above fact can be found in~\cite{harrow2013church}.

\begin{definition}
$\QMA^\sep(2, c, s)$ is like the standard $\QMA(2,c,s)$ with the restriction that the measurement $M = (M_0, M_1)$ used by Arthur needs to be separable between the provers. 
\end{definition}

So \cref{lem:qma-gap-amp-naive} can be rephrased.
\begin{lemma}[Gap amplification for $\QMA(2)$ rephrased, c.f.~\cite{HM13}]
\label{lem:qma-gap-amp-naive-rephrase}
For any $c-s>1/\poly(n)$,
    \[\QMA(2, c,s) \subseteq \QMA^\sep\left(2,~ 1-e^{-\poly(n)},~ \gapNaive\right).\]
\end{lemma}
\begin{proof}
In the proof of \cref{lem:qma-gap-amp-naive}, the measurement is a linear combination of the product test and the verification procedure applied to the proof provided by either the first prover or the second prover. The product test is a separable measurement and since the verification is applied to either one part, it's also separable.
\end{proof}

The nice property of $\QMA^\sep(2)$ is that it admits sequential repetition without going through the path that increases the number of provers and then reduces the number of provers again by the product test.
\begin{lemma}[Gap amplification for $\QMA^\sep(2)$~\cite{HM13}]
\label{lem:qma-sep-gap-amp}
For any $c,s\in[0,1]$, and any $\ell = \poly(n),$
\[\QMA^\sep(2, c, s)\subseteq\QMA^\sep(2, c^\ell, s^\ell).\]
\end{lemma}

\begin{proof}
Do sequential repetition without increasing the number of provers, i.e., each prover will provide many copies of the original proof. Apply the measurement $M$ from the $\QMA^\sep(2,c,s)$ to each copy. Accept only if all the measurements accept.

In the completeness case, it's clear that each copy is accepted independently with a probability at least $c$. Therefore, with probability $c^\ell$ all measurements accept.

In the soundness case, use the fact that the measurement $M$ is separable. The reduced state after applying $M$ by tracing out the registers being measured, pure or mixed, is separable between the two provers given that the $M$ accepts. Therefore, by the soundness, each measurement is accepted with probability at most $s$, and the probability that all measurements accept is at most $s^\ell$.
\end{proof}

Combining all the above pieces, Harrow and Montanaro obtained a stronger gap amplification for $\QMA(2)$.
\begin{theorem}[Sequential repetition for $\QMA(2)$~\cite{HM13}]
\label{thm:qma2-strong-gap-amp}
For any $c-s>\poly(n)$,
\[
    \QMA(2,c,s)  \subseteq \QMA(2,~ 1-e^{-\poly(n)}, ~e^{-\poly(n)}).
\]
\end{theorem}
\begin{proof}
Apply \cref{lem:qma-gap-amp-naive-rephrase} and then \cref{lem:qma-sep-gap-amp} with suitable parameters.
\end{proof}

\subsection{A Mild Gap Amplification for $\QMA^+(2)$}
In the previous sections, we discussed the connection between $\QMA^+(2)$ and $\QMA(2)$. The point is that, suppose we can do a strong gap amplification for $\QMA^+(2)$, then it would imply that
\[
    \NEXP \subseteq \QMA^+(2, c, s) \stackrel{?}{\subseteq} \QMA^+(2, 0.01, 0.99) \subseteq \QMA(2)!
\]
This motivates us to review the gap amplification for $\QMA(2)$.

In this section, we discuss a little bit about the gap amplification for
$\QMA^+(2)$. Although we don't know how to do gap amplification to 
achieve arbitrary completeness soundness gap (c-s gap). But as already somewhat alluded, gap amplification for $\QMA^+(2)$ in several regimes do exist.
In particular, whenever the c-s gap is smaller than \gapSmall, it can be amplified to \gapSmall; and whenever the c-s gap is larger than $\gapLarge$, then it can be further amplified to $1-\exp(-\poly(n))$.

\iffalse
\begin{figure}[H]
\caption{Gap amplifiable regime}
\begin{center}
    \includegraphics[scale=0.4]{gap-amp}
\end{center}
\end{figure}
\fi

\begin{lemma}[Gap amplification for $\QMA^+(2)$ in the small gap regime]
For any $c-s \ge 1/\poly(n)$, 
\[
    \QMA^+(2,c,s) \subseteq \QMA^+\left(2, ~1-e^{-\poly(n)}, ~\gapNaive\right).
\]
\end{lemma}
\begin{proof}
This is very much analogous to \cref{lem:qma-gap-amp-naive}. Just note that for any state $\ket\psi$, the closest product state to $\ket\psi$ also has  nonnegative amplitudes. Because, for any state $\ket\phi = \sum \alpha_i \ket i$, consider $\ket{\phi'} = \sum |\alpha_i| \ket i$. Then $|\langle \phi' \mid \psi \rangle |^2 \ge |\langle \phi \mid \psi \rangle | ^2.$
\end{proof}

Gap amplification also works when $c-s>\gapLarge.$
\begin{lemma}[Gap amplification for $\QMA^+(2)$ in the large gap regime]
For any $c-s \ge 7/8$, 
\begin{align*}
    \QMA^+(2,c,s) \subseteq \QMA^+(2, ~1-e^{-\poly(n)}, ~e^{-\poly(n)}).
\end{align*}
\end{lemma}
\begin{proof}
Apply \cref{cor:qmaplus-in-qma} and \cref{lem:qma-in-qmaplus},
\begin{align*}
    \QMA^+(2,c,s) &\subseteq \QMA\left(2, c, s+\gapLarge\right) \\
        &\quad\subseteq \QMA^+(2, ~1-e^{-\poly(n)}, ~e^{-\poly(n)}). \qedhere
\end{align*}
\end{proof}

\section*{Acknowledgement}
We thank Vijay Bhattiprolu for the discussions during the initial stages
of this project. We thank STOC reviewers, Scott Aaronson, Penghui Yao, 
and Avi Wigderson for their valuable feedback on our early draft.

\bibliographystyle{alpha}
\bibliography{macros,references}

\appendix

\appendix
\renewcommand{\thesection}{A} 

\section{Doubly Explicit \texorpdfstring{$\PCP$}{PCP} for \texorpdfstring{$\NEXP$}{NEXP}}
\label{sec:pcp-nexp}

In this section, we describe a $\PCP$ for $\NEXP$, which is doubly 
explicit and satisfies the strong uniformity property. This will imply the
following theorem immediately.
\begin{theorem}[Doubly explicit $\PCP$ for $\NEXP$]
 There is some absolute constant $\kappa<1$ and  natural number $q$, such that it is $\NEXP$-hard
    to decide $(1,\kappa)$-GapCSP for $(N=2^{\poly(n)},R=2^{\poly(n)},q,\{0,1\})$-CSP
systems that are doubly explicit.
\end{theorem}

The above $\PCP$ is obtained by $\PCP$ composition. 
The outer $\PCP$ follows closely that of~\cite[Chapter 5]{harsha2004robust},
the inner $\PCP$ (of proximity) is the Hadamard code based $\PCP$~\cite{ALMSS92}.
Our focus is the double explicitness, therefore the analysis on correctness
will be omitted. The interested readers are referred to Harsha's thesis~\cite{harsha2004robust}.

\subsection{A \texorpdfstring{$\NEXP$}{NEXP}-Complete Problem---Succinct SAT}

The starting point is a $\NEXP$-complete problem---the succinct
SAT problem~\cite{PY86}. A succinct SAT instance is an encoding of some circuit $M:\{0,1\}^{3n}\times\{0,1\}^{3}\to\{0,1\},$
$\succSAT(M)=1$ if and only if 
\begin{align*}
 & \exists x\in\{0,1\}^{2^{n}}, \, \forall(i_{1},i_{2},i_{3},\sigma)\in\{0,1\}^{3n}\times\{0,1\}^{3},\text{ s.t.}\\
 & \hspace{8em}\neg M(i_{1},i_{2},i_{3},\sigma)\lor\left(\bigvee_{j=1}^{3}(\sigma_{j}\oplus x_{i_{j}})\right).
\end{align*}
$M(i_1,i_2,i_3,\sigma)$ determines whether there is a clause consists of variables $x_{i_1}, x_{i_2}, x_{i_3}$, and $\sigma$ indicates in the clause whether the variable is negated. For example, $\sigma_1=1$ would indicate the corresponding literal being $\neg x_{i_1}$, while $\sigma_1=0$ would indicate the literal being $x_{i_1}$. The size of the circuit $M$ should be at most $\poly(n).$ 

Integrating Cook-Levin's reduction, one can conclude that there is a polynomial-size 3-CNF formula
$\Phi:\{0,1\}^{3n}\times\{0,1\}^{3}\times\{0,1\}^{t}$ for $t=n^{O(1)}$ constructing from $M$ in polynomial time,
such that
\begin{align*}
 & \succSAT(M)=1\\
 & \iff\exists x\in\{0,1\}^{2^{n}},\,\forall(i_{1},i_{2},i_{3},\sigma,w)\in\{0,1\}^{3n}\times\{0,1\}^{3}\times\{0,1\}^{t},\text{ s.t.}\\
 & \hspace{10em}\neg\Phi(i_{1},i_{2},i_{3},\sigma,w)\lor\left(\bigvee_{j=1}^{3}(\sigma_{j}\oplus x_{i_{j}})\right).
\end{align*}
Abbreviate $(i_{1},i_{2},i_{3},\sigma,w)$ by $y\in\{0,1\}^{3n+3+t},$
and let $A:\{0,1\}^{n}\to\{0,1\}$ be the polynomial of degree at most
$n$ that encodes input $x$, i.e., $A(i)=x_{i}$ for $i\in \{0,1\}^n$. Using standard arithmetization,
there is a polynomial $P:\{0,1\}^{3n+3+t+3}\to\{0,1\}$ with $\deg P=O(\size|\Phi|)=\poly(n),$
such that
\[
\neg\Phi(i_{1},i_{2},i_{3},\sigma,w)\lor\left(\bigvee_{j=1}^{3}(\sigma_{j}\oplus A(i_{j}))\right)
\iff P(y,A(i_{1}),A(i_{2}),A(i_{3}))=0.
\]
Given $M$, this polynomial $P$ can be evaluated on any input in polynomial time.

\subsection{A Robust Outer PCP for \texorpdfstring{$\NEXP$}{NEXP} with \texorpdfstring{$\poly(n)$}{poly(n)} Queries}

Based on the above discussion, a prover needs to provide $A:\{0,1\}^{n}\to\{0,1\}$
which is supposedly a polynomial of degree at most $n$, representing a satisfying assignment. To assist
the verifier, the prover will in reality provide the extended version
of $A:\FF^{n}\to\FF$ for some large finite field $\FF$ with $|\FF|=\poly(n)$.
The verifier will carry a low-degree test on $A$ to make sure that
$A$ is close to some polynomial of degree at most $n$. The low-degree
test is described below.

\noindent\fbox{\begin{minipage}[t]{1\columnwidth - 2\fboxsep - 2\fboxrule}%
\uline{Low-degree test}

\textbf{Input}: Oracle $A:\FF^{n}\to\FF$
\begin{enumerate}
\item Sample a random line by sampling random $a,b\in\FF^{n}$. 
\item Query $A(at+b)$ for all $t\in\FF.$
 \end{enumerate}
\emph{Accept} if $A(at+b)$ is a polynomial of $t$ with degree at most $n$.%
\end{minipage}}

Conditioning on $A$ being close to a low-degree polynomial, $P(y,A(i_{1}),A(i_{2}),\allowbreak A(i_{3}))$
is close to a polynomial $P_{0}:\FF^{3n+3+t}\to\FF$ of degree at
most $d=O(\deg A\cdot\deg P)=\poly(n).$ Let $m=3n+3+t$. The goal
is to test if $P_{0}$ vanishes on $\{0,1\}^{m}.$ To accomplish this
goal, the prover should provide the following auxiliary polynomials
\[
Q_{1},Q_{2},\ldots,Q_{m},P_{1},P_{2},\ldots,P_{m}:\FF^{m}\to\FF
\]
satisfying that for $i\in[m]$ 
\begin{align*}
 & P_{i-1}\equiv Z_{i}Q_{i}+P{}_{i},\\
 & P_{m}\equiv0,
\end{align*}
where $Z_{i}$ is a polynomial such that $Z_{i}(x)=0$
if and only if $x_{i}\in\{0,1\}$, for example,
\[
Z_{i}(x)=(x_{i}-1)x_{i}.
\]
These auxiliary polynomials will be bundled together in the oracle
$\Pi:\FF^{m}\to\FF^{2m},$ such that for any $x\in\FF^{m},$ $\Pi(x)$
is supposed to equal $(P_{1}(x),P_{2}(x),\ldots,P_{m}(x),\allowbreak Q_{1}(x),\ldots,Q_{m}(x)).$
Once the prover provide the auxiliary proof $P_{0}$ and $\Pi,$ the
verifier will take the following test that check whether $P_{0}$
vanishes on $\{0,1\}^{n}.$

\noindent\fbox{\begin{minipage}[t]{1\columnwidth - 2\fboxsep - 2\fboxrule}%
\uline{Zero subcube test}

\textbf{Input}: Oracle $P_{0}:\FF^{m}\to\FF,\Pi:\FF^{m}\to\FF^{2m}$
\begin{enumerate}
\item Sample a random line by sampling $a,b\in\FF^m$
\item Query all points in the line $L_{a,b}=\{t\in\FF:at+b\}$ on $\Pi$
and $P_{0}$.

\end{enumerate}
 \emph{Reject} if $P_{i-1}\not=Z_{i}Q_{i}+P_{i}$ for any $i\in[m]$ or $P_{m}\not=0$ on any point in $L_{a,b}$.\\
 \emph{Reject} if $P_{i}(at+b)$ is not a polynomial on $t$ with degree
at most $d$, $Q_{i}(at+b)$ is not a polynomial of degree at most
$d-2$.\\
\emph{Accept, }otherwise.%
\end{minipage}}

The combined $\PCP$ will be the following

\noindent\fbox{\begin{myalg}\label{alg:robust-pcp}
[Robust \mbox{$\PCP$} for \mbox{$\succSAT$}\label{algo:robust-PCP}]\ignorespacesafterend
\textbf{Input}: \textbf{$A:\FF^{n}\to\FF,\Pi:\FF^{m}\to\FF^{2m},P_{0}:\FF^{m}\to\FF$}

Take one of the following tests uniformly at random.
\begin{enumerate}
\item Low-degree test on $A$.
\item Zero subcube test on $P_{0}$ and $\Pi$.
\item Consistency test: Sample a random line $L$ by sampling random $a,b\in\FF^{m}$. \emph{Reject} if $P_{0}(y)\not=P(y,A(i_{1}),A(i_{2}),A(i_{3}))$
for any point $y\in L$.
\end{enumerate}
\emph{Accept }if all tests accept.%
\end{myalg}}
\begin{theorem}[{Robust $\PCP$~\cite[Lemma 5.4.4]{harsha2004robust}}]
 For some large enough field $\FF$ with size $\poly(n)$. If the
succinct SAT instance $M$ is satisfiable, then the test accepts with
probability $1$. Otherwise, the test satisfies the robust soundness:
If $\succSAT(M)=0$, then for some constant $\delta\in(0,1]$, with
probability at least $\delta$, the test rejects; Furthermore, the
variables queried have values $\delta/C$ far away from any satisfying
assignment for some absolute constant $C$.
\end{theorem}

We establish the uniformity and the double explicitness property for
the outer $\PCP$. The uniformity is very straightforward from the specifications
of the $\PCP$ protocol.
\begin{claim}[Uniformity of the outer $\PCP$]
 For any variable $v$ in the proof $A\circ\Pi\circ P_{0}$, the
size of the $\AdjV(v)$ depends only on which of the following parts
$v$ lies in: $A$, $P_{0},$ or $\Pi$. 
\end{claim}
To clarify, variables in the above claim have large and different alphabets. For example, a variable in $A$ has alphabet $\FF$, a variable in $\Pi$ would have alphabet $\FF^{2m}$. Toward the end, we will switch to the binary representation. But this is not an issue since the size of each variable is known and at most polynomially large (since the alphabet is at most exponentially large). The index of variables using a large alphabet and the index of the bit variables can be computed efficiently. 

Given the randomness used in Algorithm~\ref{alg:robust-pcp},
\[
r=(r_{0},a,b)\in(\{0\}\times\FF^{n}\times\FF^{n})\cup(\{1,2\}\times\FF^{m}\times\FF^{m}),
\] 
it is very efficient to compute the variables to query since only some elementary operations are required to compute the points on the line determined by $a,b$.
Moreover, given any variable, we can also compute the randomness with which the test queries the corresponding variable. To see this, we first record
a related simple fact.
\begin{claim}
\label{claim:reverse-point-lines}Given some $n\in\ZZ$ and finite
field $\FF$ with size polynomial in $n$. For any $p\in\FF^{n},$
let 
\[
\cL_{n,p}=\{(a,b)\in\FF^n\times\FF^n:at+b=p\text{ for some }t\in\FF\}
\]
be the set of lines that pass point $p$. There is a natural order
on $\cL_{n,p}$ (i.e., the alphabetical order), such that the following can be computed in time $\poly(n)$:
\begin{enumerate}
\item Given any $(a,b)\in\cL_{n,p},$ output the index of $(a,b)$
in $\cL_{n,p}$; 
\item Given any index $\iota\in[|\cL_{n,p}|]$, output the line
$(a,b)$ with index $\iota$ in $\cL_{n,p}$.
\end{enumerate}
\end{claim}

\begin{proof}
(i) For $(a,b)=(0,p)$, this is the line with the first index in $\cL_{n,p}$.
For any $a\not=0\in\FF^{n}$ and $t\in\FF$, there is a unique $b\in\FF^{n}$
such that $at+b=p.$ Therefore, given $(a,b)$, it is easy to compute
the number of lines $(a',b)$ containing $p$ with $a'<a$.
Now for all $t\in\FF,$ we can list all the $b'$ such that $at+b'=p$.
This gives us the exact index of the given pair $(a,b)$.

(ii) Given an index $\iota\in[|\cL_{n,p}|]$. If $\iota=1$, we can determine that $(a,b)=(0,p)$. Otherwise, determine $a$ by setting $a$ to be
$\lfloor(\iota-2)/|\FF|\rfloor+2$ in $\FF^n$. Then run over $t\in\FF,$ we find
the correct $b$.
\end{proof}
\begin{claim}[Double explicitness of the outer $\PCP$]
For any variable $v$ in $A,$ or in $P_{0}$ or in $\Pi$. There
is a list $\AdjV(v)$ of randomness $r$ with which the outer $\PCP$
queries the variable $v$. The following are computable in time polynomial
in $n$. 
\begin{enumerate}
\item Given any $r\in\AdjV(v)$, output the index $\iota$ of $r$ in $\AdjV(v)$. 
\item Given any $\iota\in[|\AdjV(v)|]$, output the $\iota$th randomness
in $\AdjV(v)$. 
\end{enumerate}
\end{claim}

\begin{proof}
We check all the variables and tests.

\textbf{Case 1}: For any point $p\in\FF^{m},$ consider the variable
$P_{0}(p)$. $P_0(p)$ can be queried in the zero subcube tests and the consistency
test. In either case, $P_{0}(p)$ is queried when the sampled line
passes $p$. Therefore, double explicitness holds by Claim~\ref{claim:reverse-point-lines}. 

\textbf{Case 2}: For any point $p\in\FF^{m}$, consider the variable
$\Pi(p)$, which is queried only in the zero subcube test. Again, $\Pi(p)$
is queried only when the sampled line passes $p$, so double explicitness
follows from Claim~\ref{claim:reverse-point-lines}. 

\textbf{Case 3}: For any point $p\in\FF^{n}$ corresponding to the
variable $A(p)$. In this case, $A(p)$ can be queried in the low-degree test and the consistency test.  In the low-degree test, the situation is completely covered by Claim~\ref{claim:reverse-point-lines}.
Furthermore, we know that there are exactly $m_{0}=|\FF|^{2n}$
possible $(a,b)$ that queries $p.$ So we ignore the low-degree test,
this offsets the index $\iota$ in $\AdjV(v)$ for $r=(r_0,a,b)$ with  $r_{0}\not=0$
by $m_{0}$.  So in the remainder
of the proof, we handle the consistency test.

(i) Given $r=(2, a,b)$ that queries $A(p)$, we want to determine the index of $r$. For string $y\in\FF^{m}$, focus on the coordinates $I_{1},I_{2},I_{3}$
that correspond to variables $i_{1},i_{2}$ and $i_{3},$ respectively. Fix some arbitrary $a\in\FF^{m}$, count the $b$ that satisfies
any of the following 
\begin{align*}
 &  (a|_{I_{1}},b|_{I_{1}}) \in \cL_{n,p}, & (1)\\
 &  (a|_{I_{2}},b|_{I_{2}}) \in \cL_{n,p}, & (2)\\
 &  (a|_{I_{3}},b|_{I_{3}}) \in \cL_{n,p}. & (3)
\end{align*}
Recall that $\cL_{n,p}$ is the set of the lines that pass the point $p$ in $\FF^n$. For $I\in\{I_1,I_2,I_3\}$, if $a|_{I}\not=0$, the number of $b|_I$ in $\cL_{n,p}$ is $|\FF|$; if $a|_I = 0$, then there is only one $b|_I$. In any case, there are at most polynomially many different assignments to $b|_{I_1}, b|_{I_2}, b|_{I_3}$ to satisfy $(1)$, $(2)$ or $(3)$ depending only on how many $0$s in $a|_{I_1}, a|_{I_2}, a|_{I_3}$. The other coordinates can be set arbitrarily. Therefore, even if we don't know $a$ exactly, but only the number of $0$s in $I_1, I_2, I_3$, we can still compute the number of $b$s such that $(a,b)$ queries $A(p)$.
Let 
\begin{align*}
    C_k (a) = \left\{{a'<a : \sum_{i\in[3]} \1{ a'|_{I_i} = 0}}=k\right\}.
\end{align*}
For any fixed $a$, note that $C_k(a)$ can be computed efficiently. Since $k$
decides  
\[|\{b'\in \FF^m : (a',b') \text{ queries } A(p)\}|,
\]
for any $a'\in C_k(a)$, we can compute the total number of $(a',b')$ in consistency check that queries $A(p)$ for $a'<a$.

Now fix some $b$, such that $p$ lies in line $(a,b)$ restricted to $I_{1},I_{2}$
or $I_{3}$. Count $b'<b$ such that $(a,b')$ queries $A(p)$. We
can do this because we can count the following efficiently
\begin{align*}
 & B_{k}=\{b'<b:(a|_{I_{k}},b'|_{I_{k}})\in\cL_{n,p}\},\qquad\qquad\qquad\qquad\qquad\qquad k=1,2,3.\\
 & B_{jk}=\{b'<b:(a|_{I_{j}},b'|_{I_{j}}),(a|_{I_{k}},b'|_{I_{k}})\in\cL_{n,p}\},\qquad\qquad\quad1\le j<k\le3.\\
 & B_{123}=\{b'<b:(a|_{I_{1}},b'|_{I_{1}}),(a|_{I_{2}},b'|_{I_{2}}),(a|_{I_{3}},b'|_{I_{3}})\in\cL_{n,p}\}.
\end{align*}
The reason is that for a fixed $a$, the arbitrary combination
of $(1),(2)$ and $(3)$ restricts $b'$ on the corresponding locations
(e.g. $b'|_{I_{1}}$, or $b'|_{I_{1}\cup I_{2}},$...) with at most polynomially
many assignments (in particular at most $|\FF|^3$). For each of the assignments, it is easy to count
the number of assignments on the unrestricted coordinates such that
$b'<b$. Finally, using the inclusion-exclusion principle, we know exactly
the number of $(a,b')$ that queries $A(p)$ for $b'<b$. This tells
us the index $\iota$ of $(a,b)$ for variable $A(p)$. 

(ii) Now given the index $\iota$, we want to determine $(a,b)$ in the randomness. 
We first fix the value of $a$. To do so, we start by fixing the coordinates before $I_1, I_2, I_3$. 
Then we decide if $a_{I_1}$ is 0. If not we can decide the value of $a_{I_1}$, and so 
on. After we fix the value of $a$, we decide the
value of $b|_{I_{1}}$. For any assignment $\sigma$ to $b|_{I_{1}}$,
we can count the total number of assignments of $b$ such that $(a,b)$ queries $A(p)$. 
This number only depends on whether
$(a|_{I_{1}},b|_{I_{1}})\in\cL_{n,p}$ under the given assignment
$\sigma$. Since there are only polynomially many assignments $\sigma$
making $(a|_{I_{1}},b|_{I_{1}})\in\cL_{n,p}$, we can decide the value
of $b|_{I_{1}}.$ Analogously, we can decide the value of $b|_{I_{2}}$
and $b|_{I_{3}},$ and finally all the other coordinates.

\end{proof}

\subsection{The Hadamard Inner \texorpdfstring{$\PCP$}{PCP}}

The standard approach to reduce the number of queries in a $\PCP$
system is to compose the outer $\PCP$ with a query-efficient inner
$\PCP$. In the case of $\NEXP$, the task is much easier. Simply
note that once the randomness $r$ is fixed, then there is a polynomial-time
Turing machine $M_{r}$ that verifies if the variables to query, again
depending on $r$, satisfies the corresponding test. This verification
can also be ``verified'' using the following well-known Hadamard code based
$\PCP$.
\begin{theorem}[cf. \cite{ALMSS92}]
For any constant $\delta>0,$ there is a $\PCP$ of proximity for
any $\NP$ problem with $\poly(n)$ number of random
bits, query complexity $O(1),$ perfect completeness, and robust soundness
$\delta$: for any input $\delta$-far from satisfying the circuit,
the test rejects with probability at least $O(\delta).$
\end{theorem}

For the purpose of showing the double explicitness property, we briefly
go over the construction of this $\PCP.$ For any Turing machine $M$
runs in time $\poly(|x|)$ on input $x$, whether $M(x)=1$ can be
reduced to the problem of deciding the existence of a solution to a system
of polynomially many quadratic equations in $\FF_{2}$. 
\begin{theorem}[$\NP$-completeness of quadratic equations]
\label{thm:quad-eq}Given any Turing machine $M$ that runs in time
$t=\poly(m)$ on input $x$ of length $m$. There is a polynomial
time reduction $\cA$ that runs in time $\poly(m)$ on $x$, and outputs
$A\in\{0,1\}^{\ell\times n^{2}},b\in\{0,1\}^{\ell}$ for $n,\ell=\poly(m)$,
such that
\begin{enumerate}
\item If $M(x)=1$, then for some $x'\in\{0,1\}^{n}$ that $x'\succ x$
and $A(x'\otimes x')=b,$ 
\item If $M(x)=0$, for any $x'\in\{0,1\}^{n}$ that $x'\succ x,$ $A(x'\otimes x')\not=b.$
\end{enumerate}
Furthermore, the rows of $A$ are linearly independent. 
\end{theorem}
Here $x'\succ x$ means that $x$ is a prefix of $x'$.
\begin{proof}
The correctness is standard.
The focus is to show that $A$ has linearly independent rows. Start
from the Cook-Levin's reduction. Consider the computational tableau
$T\in\{0,1,\dot{0},\dot{1},\bot \}^{t\times t}$, where $\dot{0}$ and $\dot{1}$
denote that the header is pointing to the current cell and $\bot$ denotes the empty cell. We can encode
$0,1,\dot{0},\dot{1},\bot$ using the binary alphabet by for example, $000,001,010,011,111$, respectively.
We interpret $\{0,1\}$ as elements in $\FF_{2}$. Therefore, each symbol
is encoded using three variables. By Cook-Levin's reduction, there is a 3SAT formula $\Psi$ on variables
associated with $T$. The way we encode the symbols guarantees that the input $x$ to $M$ is a substring of the
input $x'$ to $\Psi$. By rearranging, we can make sure $x$ is a prefix of $x'$. We make $\Psi$ to have fan-in 2 by adding
intermediate gates. In particular, for every internal gate, associate
a new variable. Then for every gate $z$ that takes two variables
$x$ and $y$ as its input (when the input, say $x$, is negated,
simply replace $x$ with $1-x$), add the equation based on the operation
of $z$, as below:
\begin{enumerate}
\item If $z=x\land y,$ add the equation: $xy+z=0,$
\item If $z=x\lor y,$ add the equation: $z+x+y+xy=0.$
\end{enumerate}
For the top gate $z$, add equation $z=1.$ For variables associated with the first row in the tableau $T$, add the corresponding equation to ensure things like the header is pointing to the first cell; the cells after the input $x$ is empty; etc. These equations are only enforced on the ``inputs'' to the formula $\Psi$, and for each such variable, there is only one such equation. 
Note that for any internal
gate $z$ in the formula, they only show up in two equations. One that verifies the
inputs variables are consistent with $z$. The other verifies that when
$z$ is fed into an upper gate, the values are consistent. In the first
case, there is always the term $z$. In the second case, there is
always the term $zy$ for some other variable $y$. 

We show that the equations introduced above result in a matrix $A$ with linearly independent rows. Take an arbitrary
equation that corresponds to some gate $z$ in the formula, where we introduced a term
$z$. If $z$ is not the top gate, then to eliminate the term $z$
we must include the equation corresponding to the gate $z'$ that
takes $z$ as an input, which will introduce the term $zy$ for some
$y$. This term $zy$ is not removable. If $z$ is the top gate, then
the equation itself already introduces a term $xy$ for some gate $x$
and $y$, which is not removable. 
One remaining case is when the equation is $z=1$ or $z=0$ for some $z$, an input to the formula. 
In this case, to remove the variable $z$, the only way is to look for any internal gate that takes $z$ as an input. 
However, once we take variables associated with internal gates, we are back to the first case.
\end{proof}
\noindent\fbox{\begin{myalg}
[Hadamard \mbox{$\PCP$} for some polynomial-time Turing
machine \mbox{$M$}]\ignorespacesafterend
Convert $M$ into a system of quadratic equation $A:\{0,1\}^{\ell\times n^{2}},b\in\{0,1\}^{\ell}.$
Let $x\in\{0,1\}^{m}$ be the input to $M$.

\textbf{Prover} provides the proof consists of $Y\in\FF_{2}^{n},Z\in\FF_{2}^{n^{2}}$
such that for some solution $x'\in\{0,1\}^{n}$ to $A(x'\otimes x')=b$
that extends $x$, i.e., $x'\succ x$, and satisfy
\begin{align*}
Y(y) & =\langle y,x'\rangle,\\
Z(z) & =\langle z,x'\otimes x'\rangle.
\end{align*}
\textbf{Verifier} checks the following
\begin{enumerate}
\item (Linearity test for $Y$) Sample random $y,y'\in\{0,1\}^{n},$ test
if $\langle y,x'\rangle+\langle y',x'\rangle=\langle y+y',x'\rangle$.
\item (Linearity test for $Z$) Sample random $z\in\{0,1\}^{n\times n},z'\in\{0,1\}^{n\times n}$,
test if $\langle z,x'\otimes x'\rangle+\langle z',x'\otimes x'\rangle=\langle z+z',x'\otimes x'\rangle.$
\item (Consistency test on $Y$ and $Z$) Sample $w,w'\in\{0,1\}^{n},$
test if $\langle w,x'\rangle\langle w',x'\rangle=\langle w\otimes w',x'\otimes x'\rangle.$ 
\item (Equation test) Sample $u\in\{0,1\}^{\ell}$, test if $\langle A^{T}u,x'\otimes x'\rangle=\langle A^{T}u,b\rangle.$ 
\item (Proximity test) Sample $i\in[m]$ and $v\in\{0,1\}^{n},$ test if
$\langle v+e_{i},x'\rangle+\langle v,x'\rangle=\langle e_{i},x'\rangle.$
\end{enumerate}
\emph{Accept }only if all tests pass.%
\end{myalg}}

\vspace{3mm}We make a few remarks here. First, for our purpose, we
don't really worry about the optimal query complexity as long as the
total number of queries is a constant number. Second, note that by
repeating the test multiple of times, we can detect any proximity
parameter. If the above $\PCP$ is doubly explicit, it remains doubly
explicit repeating constant number of times. Finally, actually the
prover only provides $x'_{m+1}x'_{m+2}\cdots x'_{n},$ since the first
$m$ bits are part of the input. 

Next, we establish the uniformity and double explicitness of the inner
$\PCP$. For uniformity, we can classify the variables in the proof
of the Hadamard $\PCP$ described above into constantly different
types based on:
\begin{enumerate}
\item The input $x$ to $M$, in another word, $Y(e_{i})$ for $i\in[m]$ form
one type of variables.
\item For $Y(a)$ for $a\not\in\{e_{i}:i\in[m]\}$, form another type of
variables.
\item For $Z(a)$, depending on whether $a\in\{0,1\}^{n}\otimes\{0,1\}^{n}$,
and whether $\exists u\in\{0,1\}^{\ell}$ such that $A^{T}u=a$, $\{Z(a):a\in\{0,1\}^{n^{2}}\}$
are decomposed into four different types of variables.
\end{enumerate}
\begin{claim}[Uniformity of inner $\PCP$]
There are six types of variables for the inner $\PCP$ as listed
above. For any two variable $v_{1},v_{2}$ that belong to the same
type, $|\AdjV(v_{1})|=|\AdjV(v_{2})|$. Furthermore, $|\AdjV(v_{1})|$
can be computed efficiently.
\end{claim}

\begin{proof}
By inspection. 
\end{proof}
\begin{claim}[Double explicitness of inner $\PCP$]
Fix any variable $a$ which can be either some $Y(y)$ for $y\in\FF^{n}$,
or $Z(z)$ for $z\in\FF^{n^{2}}$. Let $\AdjV(a)$ be the list of
randomness $r=(y,y',z,z',w,w',u,i,v)$ that queries $a$. The following
are computable in time $\poly(n)$:
\begin{enumerate}
\item \label{enu:inner-global-to-local}Given any $r\in\AdjV(a)$, output
the index $\iota$ of $r$ in $\AdjV(a)$. 
\item \label{enu:inner-local-to-global}Given an index $\iota,$ output
the $\iota$th random string $r$ in $\AdjV(a)$.
\end{enumerate}
\end{claim}

\begin{proof}
We carefully examine all the cases. Let $\cU=\{0,1\}^{2n+2n^{2}+2n+\ell}\times[m]\times\{0,1\}^{n},$
given any $r\in\cU$, decompose $r=r_{1}r_{2}\cdots r_{8}r_{9},$
such that $(r_{1},r_{2},\ldots,r_{9})$ corresponds to $(y,y',z,z',\allowbreak w,w',u,i,v)$
in the Hadamard $\PCP.$ 

\textbf{Case 1}: Suppose that $a\in\FF^{n}$ corresponds to the variable
$Y(a)$. $Y(a)$ can be queried in linearity test for $Y$, consistency
test and proximity test. Then 
\begin{align*}
 & \AdjV(a)=E_{1}(a)\cup E_{2}(a)\cup E_{12}(a)\cup E_{5}(a)\cup E_{6}(a)\cup E'_{8}(a)\cup E_{9}(a)\cup E_{89}(a),\\
 & E_{i}(a):=\{r\in\cU:r_{i}=a\},\qquad\qquad\qquad i\in\{1,2,5,6,9\},\\
 & E_{12}(a):=\{r\in\cU:r_{1}+r_{2}=a\},\\
 & E'_{8}(a):=\{r\in\cU:e_{r_{8}}=a\},\\
 & E_{89}(a):=\{r\in\cU:r_{9}+e_{r_{8}}=a\}.
\end{align*}
Now given any proper prefix $p$ for some $r\in\cU$ such that \[p\in\{\varepsilon,r_{1},r_{1}r_{2},\ldots,\allowbreak r_{1}r_{2}r_{3}r_{4}r_{5}r_{6}r_{7}r_{8}\},\]
where $\varepsilon$ stands for the empty string. let 
\begin{align*}
 & \AdjV(a)|_{p,r}:=\AdjV(a)\cap\{ps\in\cU:ps<r\}.
\end{align*}
We want to compute the cardinality of $\AdjV(a)|_{p,r}$. Suppose the
prefix $p$ already implies a query on $Y(a),$ then the suffix $s$
can be anything that makes $ps<r$. If the prefix $p$ does not imply
a query on $Y(a),$ we consider all the following sets. 
\begin{align*}
 & E_{i}(a,p,r):=\{r'<r:r'_{i}=a,p\prec r'\},\qquad\qquad i\in\{1,2,5,6,9\},\\
 & E_{12}(a,p,r):=\{r'<r:r'_{1}+r'_{2}=a,p\prec r'\},\\
 & E'_{8}(a,p,r):=\{r'<r:e_{r'_{8}}=a,p\prec r'\},\\
 & E_{89}(a,p,r):=\{r'<r:r'_{9}+e_{r'_{8}}=a,p\prec r'\}.
\end{align*}
We claim that we can compute the cardinality of the intersection for
an arbitrary combination of the above sets. If this is indeed the
case, the cardinality of $\AdjV(a)|_{p,r}$ can be computed efficiently
using the inclusion-exclusion principle. First, for $r'\in E_{i},$ $r_{i}'$
is fixed to be $a$. For $E'_{8}$, $|E'_{8}|$ is nonzero only if
$a=e_{i}$ for some $i\in[m]$. In that case, it fixes the value of
$r'_{8}.$ For $E_{89},$ there are at most $m$ possible ways of
setting $r'_{8}$ and $r'_{9}.$ When we consider the intersection
of an arbitrary combination of the above sets, we are restricting
the corresponding coordinates to at most $m$ possible assignments,
which we can list efficiently. For each assignment, it is easy to
count the number of assignments to the unrestricted coordinates that
are consistent with $p$ and smaller than $r$. Finally, we take $E_{12}$
into account. If $p$ already fixes $r_{1}'$, then it determines
$r'_{2}.$ Otherwise, for every possible $r'_{1}$, there is one corresponding
$r'_{2}$. When taking intersections with other sets, the corresponding
coordinates $I$ are restricted to at most $m$ possible assignments.
We can exhaust the assignments to $I$, and for all $r_{1}'<r_{1},$
the unrestricted coordinates can have arbitrary values. For the single
special case $r'_{1}=r_{1},$ depending on whether $r'_{2}<r_{2}$
or $r'_{2}=r_{2}$ or $r'_{2}>r_{2}$, we can also count efficiently
the number of assignments to the other coordinates such that $r'<r$. 

The above discussion helps us establish the double explicitness for
$Y(a)$. In particular, \ref{enu:inner-global-to-local} given any
$r\in\cU,$ by computing the cardinality of $\AdjV(a)|_{p,r}$ for $p=\varepsilon$,
we can compute the index $\iota$ of $r$ in $\AdjV(a)$. \ref{enu:inner-local-to-global}
Suppose we are given the index $\iota$. For any prefix $p$, we can
efficiently compute $|\AdjV(a)|_{p,r}|$, by setting $r=1^{2n+2n^2+2n+\ell}\circ m \circ 1^n$. The cardinality only depends on whether
$p$ already queries $Y(a)$, the length of $p$, and $a$. Therefore,
we can compute the $\iota$th randomness in $\AdjV(a)$ by gradually
determine $r_{1},r_{2},\ldots,r_{9}.$

\textbf{Case} \textbf{2}: Suppose $a\in\FF^{n^{2}}$ corresponds to
some $Z(a)$. $Z(a)$ can be queried in linearity test for $Z$, consistency
test and equation test. Then
\begin{align*}
 & \AdjV(a)=E_{3}(a)\cup E_{4}(a)\cup E_{34}(a)\cup E_{56}(a)\cup E_{7}(a),\\
 & E_{i}(a):=\{r\in\cU:r_{i}=a\},\qquad\qquad\qquad i\in\{3,4,7\},\\
 & E_{34}(a):=\{r\in\cU:r_{3}+r_{4}=a\},\\
 & E_{56}(a):=\{r\in\cU:r_{5}\otimes r_{6}=a\},\\
 & E_{7}(a):=\{r\in\cU:A^{T}r=a\}.
\end{align*}
Analogous to case 1, we also consider the version $\AdjV(a)|_{p,r},E_{i}(a,p,r),E_{ij}(a,p,r)$
that are consistent with some prefix $p$. The cardinality of $E_{i}(a,p,r)$
can be computed just like in case 1. The cardinality of $E_{56}$ is
nonzero only when $a$ is a tensor product of some $w,w'$, and $a$
completely determines $w$ and $w'$. For $E_{7}$, we need to solve
the following linear equation such that 
\[
A^{T}u=a.
\]
Since the rows of $A$ are independent, there is at most one solution
for the above equation. This can be found in polynomial time using,
for example, Gaussian elimination. All in all, when considering the intersection
of an arbitrary combination of the above sets, we are restricting
a few coordinates to at most $1$ possible assignment. It is easy
to count the number of assignments on the other unrestricted coordinates
that are consistent with the prefix $p$ and smaller than $r$. To
take $E_{34}$ into account, this is completely analogous to what
happens in case 1. Therefore, we can compute $\AdjV(a)|_{p,r}$ efficiently.

Now it follows the same argument as in case 1, given any $r\in\cU,$
we can compute the index $\iota$ of $r$ in $\AdjV(a)$ and given
any index $\iota$ we can compute the corresponding $\iota$th randomness
$r$.
\end{proof}

\subsection{The \texorpdfstring{$\PCP$}{PCP} Composition}

The final $\PCP$ for the succinct SAT problem will be the composition
of the outer $\PCP$ and the inner $\PCP$. In particular, for any
succinct SAT instance $M$, let $s=\size(M)$. The prover should provide
the proof $\Pi^{\outer}$ for the outer $\PCP$. The outer $\PCP$
verifier samples the randomness $r\in\{0,1\}^{\poly(s)}$. Depending
on $r$, some polynomial-time verification $M_{r}$ will be invoked
to verify a set of variables $I_{r}$ in $\Pi^{\outer}$, denoted
by $\Pi^{\outer}|_{I_{r}}$. $M_{r}$ can be converted into a quadratic
equation instance $(A_{r},b_{r})$ in time $\poly(s)$. The prover
will provide for all possible randomness $r$, a proof $\Pi_{r}^{\inner}$.
Now the inner $\PCP$ will verify $\Pi^{\outer}|_{I_{r}}\circ\Pi_{r}^{\inner}$.
Sample the randomness $r'\in\{0,1\}^{\poly(s)}$ for the inner $\PCP$.
Based on $r'$, there is a polynomial-time verification $M_{r'}^{\inner}$
that verifies $\Pi^{\outer}|_{I_{r}}\circ\Pi_{r}^{\inner}$.

The prover will arrange the proofs as a concatenation of $\Pi^{\outer}\circ\Pi_{0}^{\inner}\circ\Pi_{1}^{\inner}\circ\cdots$.
Note that there are exactly $m_{0}=|\FF|^{2n}$ random
strings for the low-degree tests in the outer $\PCP$. These tests
correspond to the same verification procedure, therefore the inner
$\PCP$s have the same structure. Following the low-degree tests are the
zero tests corresponding to the next $m_{1}=|\FF|^{2m}$
random strings. Finally, the remaining are $m_{2}=|\FF|^{2m}$
consistency tests. We know exactly the size of $|\Pi_{r}^{\inner}|$
for each $r$. Therefore for any variable $v$, it can be computed
efficiently whether $v$ lies in $\Pi^{\outer}$ or $\Pi_{r}^{\inner}$,
and in the latter case, we can computer $r$ in polynomial time. So
when we talk about a variable $v$, we suppose the information is
provided. 
\begin{theorem}
The composed $\PCP$ is doubly explicit.
\end{theorem}
\begin{proof}
 Fix some variable $v$, there are two cases. First, if $v\not\in\Pi^{\outer}$.
This case is straightforward: $v$ is queried only if the random string
$r$ for the outer $\PCP$ is correct. Then the double explicitness
for $v$ follows the double explicitness of the inner $\PCP$. 

Second, if $v\in\Pi^{\outer}$. Now given $r,r'$ the random strings
for the outer and inner $\PCP$s, respectively. From the double explicitness
of the outer $\PCP,$ we know the index $\iota$ of the $r\in\AdjV^{\outer}(v)$.
From which, we can compute the cardinality of $\cR=\{(s,s')\in\AdjV(v):s<r\}$.
This is because by $\iota$ we know exactly the cardinality of $\cR_{i}=\{s\in\AdjV^{\outer}(v): s<r\}\cap T_{i}$
for $i\in\{1,2,3\},$ where $T_{1}, T_{2}, T_{3}$ are the sets of random
strings for the outer $\PCP$ that invoke the low-degree tests, zero
tests, and consistency tests, respectively. Due to the uniformity
of the inner $\PCP,$ for any $s\in\cR_{i}$, the size of the adjacency
list of $v$ for the inner $\PCP$ is the same. Denote $n_{i}$ be
the size of adjacency list of $v$ for any $s\in\cR_{i}$, we have
\[
|\cR|=\sum_{i=1}^{3}n_{i}\cdot|\cR_{i}|.
\]
Now by the double explicitness of the inner $\PCP$, we get the index
$\iota'$ of $r'$. The index of $(r,r')$ for the composed $\PCP$
is therefore $|\cR|+\iota'$. On the other hand, let some index $\iota$
be given. Since we can compute $n_{i}$, it is easy to fix $r$. Then
the double explicitness of the inner $\PCP$ will determine $r'$.

The above argument establishes the double explicitness on adjacency list $\AdjV$. The explicitness of $\AdjC$ is straightforward. Given the random string $r,r'$, fully determined by $r$ the outer $\PCP$ queries a line in one of three tests, the points on which are efficient to list. Look inside the corresponding inner $\PCP$, by $r'$ we can efficiently output the corresponding locations to query. Since we can efficiently output the list of variables that $(r, r')$ queries, it shows the explicitness on the adjacency list $\AdjC$.
\end{proof}
The uniformity of the inner $\PCP$ and outer $\PCP$ together implies
the uniformity of the composed $\PCP.$
\begin{theorem}
In the composed $\PCP$, there are only a constant number of different
types $[N=2^{\poly(s)}]=V_{1}\cup V_{2}\cup\cdots\cup V_{k}$ of variables
in the sense that the size of $\AdjV(v)$ only depends
on which type the variable $v$ is.
\end{theorem}

\begin{proof}
First, consider any variable $v\not\in\Pi^{\outer}$. There are only
3 different kinds of inner $\PCP$ depending on whether the outer
$\PCP$ invokes the low-degree test, zero subcube test, or consistency test.
For each test, by the uniformity of the inner $\PCP$, there are 5
different types of variables. In total, there are 15 different types
of variables. Consider any variable $v\in\Pi^{\outer}$. (Note that
when we discuss the outer $\PCP$, we are using a large alphabet. Here,
a variable has a binary value. So we split one variable from the outer
$\PCP$ into polynomially many.) By the uniformity of the outer $\PCP$,
which and how many of the low-degree tests, zero tests, and consistency
tests are only depending on whether $v$ belongs to $A, P_{0}$ or
$\Pi$. For any $v$ that belongs to the same type, the uniformity of the
inner $\PCP$ tells us that the total number of constraints that queries
$v$ is fixed.
\end{proof}

\end{document}